\documentclass[a4paper,USenglish,cleveref, autoref, thm-restate]{lipics-v2019}

\usepackage{xcolor}
\usepackage{tikz}
\usetikzlibrary{calc}
\usepackage{csquotes}
\usepackage{thmtools}

\usepackage{mdframed}
\usepackage{relsize}

\newcommand{\R}{\mathbb{R}}
\newcommand{\Z}{\mathbb{Z}}
\newcommand{\N}{\mathbb{N}}

\newcommand{\I}{\mathcal{I}}

\newif\ifnotes\notestrue

\ifnotes
\usepackage{color}
\newcommand{\notename}[2]{{\footnotesize{\bf (#1:} {#2}{\bf ) }}}

\newcommand{\jnote}[1]{\textcolor{red}{{\notename{Jim}{#1}}}}

\newcommand{\lnote}[1]{{\textcolor{blue}{\notename{Lasse}{#1}}}}

\else

\newcommand{\notename}[2]{{}}

\newcommand{\anote}[1]{}
\newcommand{\sidford}[1]{}
\newcommand{\sidfordLongerTodo}[1]{}
\newcommand{\lnote}[1]{}
\newcommand{\dnote}[1]{}
\newcommand{\jnote}[1]{}

\fi

\crefname{observation}{observation}{observations}
\Crefname{observation}{Observation}{Observations}

\title{Completeness in the Polynomial Hierarchy and PSPACE for many natural problems derived from NP}

\titlerunning{Completeness in PH derived from NP} 

\author{Christoph Grüne}{Department of Computer Science, RWTH Aachen University, Germany}{gruene@algo.rwth-aachen.de}{https://orcid.org/0000-0002-7789-8870}{Funded by the German Research Foundation (DFG) – GRK 2236/2.}
\author{Berit Johannes}{Flix, Munich, Germany}{berit@mit.edu}{}{}
\author{James B. Orlin}{Sloan School of Management, Massachusetts Institute of Technology, Cambridge, MA}{jorlin@mit.edu}{https://orcid.org/0000-0002-7488-094X}{}
\author{Lasse Wulf}{Section for Theoretical Computer Science, IT University of Copenhagen, Denmark}{lasw@itu.dk}{https://orcid.org/0000-0001-7139-4092}{Funded by the Austrian Science Fund (FWF):W1230 and by the Carlsberg Foundation CF21-0302 ``Graph Algorithms with Geometric Applications''.}

\authorrunning{C. Grüne, B. Johannes, J.B. Orlin and L. Wulf} 

\Copyright{C. Grüne, B. Johannes, J.B. Orlin and L. Wulf} 



\category{} 

\relatedversion{} 

\supplement{}

\nolinenumbers 

\hideLIPIcs  

\EventEditors{John Q. Open and Joan R. Access}
\EventNoEds{2}
\EventLongTitle{42nd Conference on Very Important Topics (CVIT 2016)}
\EventShortTitle{CVIT 2016}
\EventAcronym{CVIT}
\EventYear{2016}
\EventDate{December 24--27, 2016}
\EventLocation{Little Whinging, United Kingdom}
\EventLogo{}
\SeriesVolume{42}
\ArticleNo{23}

\definecolor{darkgreen}{RGB}{0,120,0}

\newcommand{\powerset}[1]{2^{#1}}
\newcommand{\set}[1]{\{ #1 \}}
\newcommand{\fromto}[2]{\set{#1, \ldots, #2}}
\newcommand{\bin}{\set{0,1}}

\newcommand{\NP}{\textsf{NP}}

\newcommand{\PH}{\textsf{PH}}
\newcommand{\PSPACE}{\textsf{PSPACE}}

\newcommand{\NPS}{\textsf{NP-S}}

\newcommand{\leqSE}{\leq_{\mathrm{SE}}}

\newcommand{\U}{\mathcal{U}}

\newcommand{\sol}{\mathcal{S}}
\newcommand{\Uemb}{\U_{emb}}
\newcommand{\Uaux}{\U_{aux}}
\newcommand{\Uembp}{\U'_{emb}}
\newcommand{\Uauxp}{\U'_{aux}}

\newcommand{\sat}{\textsc{Sat}}
\newcommand{\satV}{\textsc{Sat-V}}

\newcommand{\restrictedGameRules}{
\begin{itemize}
  \item \textbf{Rule 1.}  In move $i$ where $i$ is even, the Blocker must block exactly one of $u_{ij}$ or $w_{ij}$ for each $j \in \{1, ..., n\}$. 
  \item \textbf{Rule 2.}  In move $i$, where $i$ is odd and $i < k$, the Protector must protect exactly one of $u_{ij}$ or $w_{ij}$ for each $j \in \{1, ..., n\}$.
  \item \textbf{Rule 3.}  By the end of move $k-1$, the Blocker must have blocked each element of the set $\bigcup_{i \text{ odd, }i \neq k}\ U_i \cup W_i$ that was not protected. 
  \item \textbf{Rule 4.}  When selecting elements in move $k$, the Protector must select each unblocked element of $U_i \cup W_i$ for each $i < k$.  
\end{itemize}
}

\begin{document}

\maketitle
\begin{abstract}
    Many natural optimization problems derived from $\NP$ admit bilevel and multilevel extensions in which decisions are made sequentially by multiple players with conflicting objectives, as in interdiction, adversarial selection, and adjustable robust optimization.
Such problems are naturally modeled by alternating quantifiers and, therefore, lie beyond $\NP$, typically in the polynomial hierarchy or $\PSPACE$.
Despite extensive study of these problem classes, relatively few natural completeness results are known at these higher levels.
We introduce a general framework for proving completeness in the polynomial hierarchy and $\PSPACE$ for problems derived from $\NP$. 
Our approach is based on a refinement of $\NP$, which we call \emph{$\NP$ with solutions} ($\NPS$), in which solutions are explicit combinatorial objects, together with a restricted class of reductions -- \emph{solution-embedding reductions} -- that preserve solution structure.
We define \emph{$\NPS$-completeness} and show that a large collection of classical $\NP$-complete problems, including \emph{Clique}, \emph{Vertex Cover}, \emph{Knapsack}, and \emph{Traveling Salesman}, are $\NPS$-complete.

Using this framework, we establish general meta-theorems showing that if a problem is $\NPS$-complete, then its natural two-level extensions are $\Sigma_2^p$-complete, its three-level extensions are $\Sigma_3^p$-complete, and its $k$-level extensions are $\Sigma_k^p$-complete.
When the number of levels is unbounded, the resulting problems are $\PSPACE$-complete.
Our results contain several previously known completeness results for multilevel optimization problems derived from $\NP$ and yield many new ones simultaneously, demonstrating that high computational complexity is a generic feature of multilevel extensions of $\NP$-complete problems.
    
\end{abstract}

\section{Introduction}

Many natural decision and optimization problems are obtained by taking a classical problem in $\NP$ and extending it with sequential, adversarial, or multilevel decision making.
Such extensions arise, for example, in interdiction, adversarial selection, and multi-stage optimization, where one player acts first and another responds optimally.
From a complexity-theoretic perspective, these problems are no longer decision problems in $\NP$, but are located in higher levels of the polynomial hierarchy or in $\PSPACE$.
The goal of this paper is to systematically understand the complexity of such problems \emph{derived from $\NP$}, and to explain why their multilevel variants are generically complete for these higher complexity classes.

In recent years, extensions of problems in $\NP$ have been motivated by enormous interest in the areas of Bilevel Optimization \cite{dempe2020bilevel, DBLP:journals/ejco/KleinertLLS21}, Robust Optimization \cite{DBLP:books/degruyter/Ben-TalGN09,robook}, Network Interdiction \cite{smith2013modern}, Stackelberg Games \cite{li2017review}, Attacker-Defender games \cite{HUNT2023}, and many other bilevel problems.
These research areas are vital in helping us to understand, which parts of a network are most vulnerable, prevent terrorist attacks, understand economic processes, and to understand and improve the robustness properties of many systems.

What unites these fields is their inherently multilevel structure: decisions are made sequentially by multiple actors with potentially conflicting objectives -- often modeled through alternating quantifiers or game-theoretic interactions.
For example, in bilevel optimization and Stackelberg games, a leader makes a decision anticipating the optimal response of a follower;
in network interdiction and attacker-defender games, one party attempts to disrupt or protect system functionality under budget constraints while anticipating countermeasures from its opponent.
This multilevel nature  elevates such problems into higher levels of the polynomial hierarchy, where the alternation between existential (“there exists a strategy”) and universal (“for all possible attacks”) quantifiers reflects real-world adversarial dynamics.
Accordingly, these problems can be characterized as an abstract two-turn game between two players, the existential player and the universal player.
The first player (from now on called Alice) starts the game and takes an action with the goal of achieving some objective.
Afterwards, the second player (from now on called Bob) responds to Alice.
Typically, but not always, Alice’s and Bob’s goals are opposite of each other.

Let us take as a prototypical example the maximum clique interdiction problem \cite{DBLP:journals/networks/PajouhBP14,DBLP:conf/iscopt/PaulusmaPR16, DBLP:journals/eor/FuriniLMS19, DBLP:journals/anor/Pajouh20}.
In this problem, we are given a graph $G = (V, E)$, a subset $V' \subseteq V$ some budget $\Gamma \in \N$, and some threshold $t \in \N$.
The goal of Bob is to find a clique of size $t$ in the graph.
However, before Bob’s turn, Alice can delete up to $\Gamma$ vertices in $V'$ in order to impair Bob’s objective.
Hence the game is a bilevel problem described by
$$
	\exists B \subseteq V' \text{ with } |B| \leq \Gamma : \forall C \subseteq V: \text{If $C$ is a clique in $G$ of size $t$, then $B \cap C \neq \emptyset$}.
$$
We remark that this problem follows a natural pattern, in which researchers often come up with new problems:
First, some nominal problem is taken (in this case, Clique), and afterwards it is modified into a more complicated multilevel problem by adding an additional component (in this case, the possibility of Alice to interdict).

A natural extension of classical interdiction problems are \emph{protection-interdiction games}, which provide a canonical source of decision problems in the polynomial hierarchy ($\PH$) and in $\PSPACE$.
In such a game, Alice (the attacker) and Bob (the defender) face off; each has their own budget ($\Gamma_A$ for Alice and $\Gamma_B$ for Bob), which is used to attack or protect individual elements from the universe.
The game proceeds in rounds:
first, Bob protects some elements;
then, Alice interdicts unprotected elements;
this alternates until round $k$ is reached, in which Bob must find a valid solution using only the surviving elements.
Bob wins the game if he is able to construct a solution; otherwise, Alice wins.
The number of alternating moves directly determines the complexity of the resulting decision problem, placing it at higher levels of $\PH$ or, when the number of rounds is part of the input, in $\PSPACE$.

\subsection{Structural Properties of NP Reductions}

Many classical $\NP$-completeness reductions preserve substantially more structure than is required to establish $\NP$-hardness.
Rather than merely mapping yes-instances to yes-instances and no-instances to no-instances, they often establish a direct correspondence between the elements of solutions to the source problem and those of the target problem.
Consequently, solutions to the source problem can frequently be reconstructed by selecting an appropriate subset of elements from a solution to the target problem.
This phenomenon is familiar from textbook reductions.
For example, in reductions from {\sc Sat} to problems such as {\sc Vertex Cover}, {\sc Independent Set}, or {\sc Clique}, satisfying assignments are encoded by selecting vertices from variable gadgets, and conversely, every minimal solution induces a valid assignment of the original formula.
While such structure is not required for $\NP$-completeness, it is ubiquitous in practice.
This additional structure becomes essential in multilevel optimization, where solutions themselves become objects of the game.
Players alternately select, protect, block, or otherwise manipulate individual solution elements before a final solution is constructed.
Consequently, the computational problem depends not only on whether a solution exists, but also on how solutions are composed from their constituent elements.
These observations motivate a restricted class of reductions that preserve the solution structure explicitly.
Informally, such a reduction embeds the elements of solutions of one problem into the elements of solutions of another in a way that preserves feasibility in both directions.
We refer to such reductions as \emph{solution-embedding reductions}.

\subsection{NP with Solutions}
Standard formulations of $\NP$ treat solutions as opaque certificates, whose internal structure is irrelevant once feasibility can be verified in polynomial time.
This abstraction is sufficient for single-level decision problems, but becomes inadequate in multilevel optimization settings, where solutions are subject to sequential or adversarial interaction.
To address this, we introduce a refinement of $\NP$ in which solutions are explicit combinatorial objects composed of elements drawn from a well-defined universe.
We refer to this class as \emph{$\NP$ with solutions} ($\NPS$).
Intuitively, a problem in $\NPS$ consists not only of a set of instances, but also of a specification of how solutions are represented and which elements they are composed of.
Different representations of solutions -- such as vertex-based versus edge-based formulations of the same graph problem -- are treated as distinct problems in $\NPS$, reflecting the fact that they give rise to different multilevel variants.
The role of $\NPS$ is not to introduce additional computational power beyond $\NP$, but to expose structural information about solutions that is already implicitly present in many classical reductions.
Within this framework, solution-embedding reductions can be defined precisely, and their consequences for bilevel and multilevel problems can be analyzed rigorously.
In the course of the paper, we show that at least the following 26 classical problems are $\NPS$-complete:
\begin{quote}
        Satisfiability,
        3-Satis\-fiability,
        Vertex Cover, Dominating Set, Set Cover, Hitting Set, Feedback Vertex Set, Feedback Arc Set, 
        Uncapacitated Facility Location, 
        p-Center, p-Median,
        3-Dimensional Matching,
        Independent Set, Clique,
        Subset Sum, Knapsack, Partition, Scheduling,
        Directed/Undirected Hamiltonian Path, 
        Directed/Undirected Hamiltonian Cycle, 
        Traveling Salesman Problem,
        Two Directed Vertex Disjoint Path, 
        $k$-Vertex Directed Disjoint Path,
        Steiner Tree.
\end{quote}

A formal description of all these problems is provided in \Cref{sec:ssp-reductions}.
Furthermore, one can add new problems to the class $\NPS$ with very little work by finding a solution-embedding reduction starting from any problem which is already $\NPS$-complete.
Since we could easily show for many classic problems that they are $\NPS$-complete, we anticipate the addition of many more problems in the future.
The website \url{reductions.network} \cite{DBLP:journals/corr/abs-2511-04308} collects problems for the classes $\NPS$, in which currently 54 problems are registered.

\paragraph*{Lifting to the polynomial hierarchy and beyond.}

Multilevel optimization problems can be interpreted as games with alternating existential and universal choices, directly mirroring the structure of quantified Boolean formulas (QBF), where quantifiers correspond to alternating player moves over solution elements.
As a consequence, QBF naturally serves as the source problem for establishing hardness of multilevel constructions.

We instantiate this principle for Satisfiability by introducing its multilevel variant as a decision problem in which satisfying assignments are constructed through such alternating interactions over literals.
We then give a reduction from quantified Boolean Satisfiability, the canonical complete problem for the polynomial hierarchy (and for $\PSPACE$ when the number of alternations is unbounded), to this multilevel variant of Satisfiability.

Since the multilevel variant of Satisfiability admits a structured, element-wise solution representation, and this structure is preserved under solution-embedding reductions, the hardness transfers uniformly to all $\NPS$-complete problems.
This yields a systematic lifting framework for deriving completeness results for their multilevel counterparts.

The idea of exploiting structure-preserving reductions to lift $\NP$-completeness results to higher levels of the polynomial hierarchy was introduced by Johannes in her doctoral thesis~\cite{johannes2011new}, where it was applied to a variety of adversarial and bilevel problem classes.
Grüne and Wulf~\cite{DBLP:conf/ipco/GruneW25} formalized these ideas in the framework of subset search problems and applied it to several problem classes on the second and third level of the hierarchy.
The present paper is one of two full versions of~\cite{DBLP:conf/ipco/GruneW25} (see also the extended version \cite{DBLP:journals/corr/abs-2311-10540}).
The companion paper \cite{DBLP:journals/corr/abs-2611-11375} treats linear objective problems and min-max regret.
The present paper contains proofs that were omitted from the conference version.
In addition, this paper develops the $\NPS$ framework in full, generalizes the meta-theorems for adversarial selection games, protection-interdiction and adjustable robust problems to all levels of the polynomial hierarchy and to $\PSPACE$, studies alternative $\NPS$ models, and contains a compendium of $\NPS$-complete problems.

\subsection{Our Contribution}
\label{subsec:meta-theorems-2}

The central contribution of this paper is a collection of meta-theorems for three classes of multilevel optimization problems derived from $\NPS$ problems:
adversarial selection games, protection-interdiction games, and multi-stage adjustable robust optimization.

For these classes, we establish completeness results throughout the polynomial hierarchy and $\PSPACE$.
In particular, we recover several previously known problem-specific completeness results in these areas, each of which was established for an individual problem by a problem-specific reduction: interdiction knapsack \cite{DBLP:journals/siamjo/CapraraCLW14,DBLP:journals/corr/abs-2406-01756}, and two-stage adjustable variants of TSP, vertex cover, and independent set \cite{DBLP:journals/dam/GoerigkLW24}.
For maximum clique interdiction \cite{DBLP:journals/amai/Rutenburg93}, we recover the variant with a set of vulnerable elements; the minimum cardinality variant is covered in \cite{DBLP:journals/tcs/GruneW26}.
At the same time, we obtain a uniform extension of these results to broad classes of problems, showing that these completeness phenomena are not isolated but arise systematically across multilevel extensions of $\NP$-complete problems.

As a consequence of this general framework, we resolve open questions raised in \cite{DBLP:journals/dam/GoerigkLW24,DBLP:conf/ipco/GruneW25} concerning the existence of a unified meta-reduction approach for multilevel optimization problems.
We concretely establish the following results.

\subparagraph{Adversarial selection games.}
In \cref{sec:selection-game}, we are concerned with adversarial selection games, which are the natural adaptation of quantified Boolean satisfiability to any $\NPS$ problems under a multilevel optimization regime.
An adversarial selection game is defined over a nominal problem from the class $\NPS$.
Two players, Alice and Bob, alternatingly choose elements from the universe of the nominal problem.
Alice wants to construct a solution, and Bob wants to prevent Alice from doing so.
In this section, we introduce a proof technique custom-fit to solution-embedding reductions, which we call ``dual mimicking''.
Using this dual mimicking proofs, we show that adversarial selection games on $\NPS$-complete problems are $\Sigma^p_k$-complete or, if $k$ is part of the input, $\PSPACE$-complete.

\subparagraph{Protection-interdiction games.}
In \cref{sec:interdiction}, we are concerned with \emph{protection-in\-ter\-dic\-tion games}.
In such a game, there are two players, a protector and an interdictor.
Each of them has a budget: the protector can protect up to $\Gamma_P$ elements, while the interdictor can block up to $\Gamma_I$ elements.
Now, both of the players alternatingly choose elements to protect or block. The interdictor is not able to block protected elements, and the protector is not able to protect blocked elements.
In the last stage, the protector needs to find a solution to the nominal problem (for example, a Hamiltonian cycle, a vertex cover, a maximum clique, etc.) while only using non-blocked elements.
If the protector is able to do so, he wins, otherwise, the interdictor wins.
Our main result is that for every problem $\Pi \in \NPS$, the corresponding $k$-stage protection-interdiction problem is $\Sigma^p_k$-complete or, if $k$ is part of the input, $\PSPACE$-complete.

\subparagraph{Multi-stage adjustable robust optimization.}
In \cref{sec:two-stage}, we are concerned with \emph{multi-stage adjustable robust optimization} with discrete budgeted uncertainty (also called discrete $\Gamma$-uncertainty).
A decision maker has to make a decision on the first stage without full information on the costs of future decisions, or possibly on what decisions are available in the future (here-and-now).
The decision in the following stage is made after some part of uncertainty is revealed (wait-and-see).
In this paradigm, in each stage, new decision elements become available.  At the beginning of the stage, an adversary can block some of the available decision elements. Then the decision maker decides on the remaining elements.
For a problem in $\NPS$, the $k$-stage adjustable robust versions are in $\Sigma^p_{2k-1}$.
Our main result in \cref{sec:two-stage} is that for every problem $\Pi \in \NPS$, the corresponding $k$-stage adjustable problem with discrete budgeted uncertainty is $\Sigma^p_{2k-1}$-complete or, if $k$ is part of the input, $\PSPACE$-complete.

\subparagraph{Alternative $\NPS$ models.}
In \cref{sec:alternative-NPS-models}, we examine alternative ways of modeling problems in $\NPS$ and discuss the robustness of our framework with respect to these choices.
In particular, we consider how different representations of solutions
-- such as alternative universes or encodings induced by different verifiers --
give rise to distinct $\NPS$ problems, even when they correspond to the same underlying $\NP$ decision problem.
We show that our notion of solution embedding and the resulting meta-theorems extend naturally to these alternative $\NPS$ models.
This demonstrates that the framework developed in this paper is not tied to a single canonical representation, but instead captures a structural phenomenon that persists across a broad range of natural solution models.

\subparagraph{Compendium of $\NPS$-complete problems.}
\Cref{sec:ssp-reductions} provides a compendium of $\NPS$-complete problems, serving both as a reference and as a practical guide for applying our meta-theorems.
For each problem, we specify a natural universe and solution structure and indicate how $\NPS$-completeness follows via solution-embedding reductions.
This compendium consolidates a large number of classical $\NP$-complete problems within a single framework and makes explicit which problem representations are suitable for lifting to higher levels of the polynomial hierarchy.
As a result, it enables readers to readily identify new multilevel problems
-- such as adversarial selection, interdiction, or multi-stage variants --
that are complete for $\Sigma_k^p$ or $\PSPACE$ without the need for additional problem-specific reductions.

\subsection{Literature overview}
Due to space restrictions, it is impossible to consider every sub-area of multilevel optimization in a single paper.
In this paper, we restrict our attention to the areas of \emph{interdiction} and \emph{adjustable robust optimization}.
We choose these areas, since they represent highly popular sub-areas of multilevel optimization, and since we can showcase a diverse set of techniques in these different cases, while staying concise at the same time. 
In fact, since there are hundreds of papers in each of these popular areas, it is not possible to give a complete overview.
We split this literature survey into an overview of the general area, as well as an overview of work specifically on $\Sigma^p_k$-completeness.

\paragraph*{General area}
In the area of \emph{network interdiction}, one is concerned with the question, which parts of a network are most vulnerable to attack or failure.
Such questions can be formulated as a min-max problem, and can be imagined as a game between Alice and Bob, where Alice interdicts a limited set of elements from the network in order to disturb Bob's objective.

Other names for slight variants of the network interdiction problem are the most vital nodes/most vital edges problem and the vertex/edge blocker problem.
Network interdiction has been considered for a vast amount of standard optimization problems.
Among these are problems in $\sf P$ such as
shortest path  \cite{malik1989k,bar1998complexity,DBLP:journals/mst/KhachiyanBBEGRZ08},
matching \cite{DBLP:journals/dam/Zenklusen10a},
minimum spanning tree \cite{DBLP:journals/ipl/LinC93},
or maximum flow \cite{WOOD19931}.
These publications show the $\NP$-hardness of the corresponding interdiction problem variant (except \cite{malik1989k} showing that the single most vital arc problem lies in $\sf P$).
Furthermore, algorithms for interdiction problems that are $\NP$-complete were developed, for example for
vertex covers \cite{DBLP:conf/iwoca/BazganTT10, DBLP:journals/dam/BazganTT11},
independent sets \cite{DBLP:conf/iwoca/BazganTT10, DBLP:journals/dam/BazganTT11,DBLP:journals/gc/BazganBPR15, DBLP:conf/tamc/PaulusmaPR17,DBLP:journals/dam/HoangLW23},
colorings \cite{DBLP:journals/gc/BazganBPR15, DBLP:conf/iscopt/PaulusmaPR16, DBLP:conf/tamc/PaulusmaPR17},
cliques \cite{DBLP:journals/networks/PajouhBP14,DBLP:conf/iscopt/PaulusmaPR16, DBLP:journals/eor/FuriniLMS19, DBLP:journals/anor/Pajouh20},
dominating sets \cite{DBLP:journals/eor/PajouhWBP15}, or
1- and $p$-center \cite{DBLP:conf/cocoa/BazganTV10, DBLP:journals/jco/BazganTV13},
1- and $p$-median \cite{DBLP:conf/cocoa/BazganTV10, DBLP:journals/jco/BazganTV13}.
Typically, the complexity of these problems is not analyzed beyond $\NP$-hardness.
A general survey is provided by Smith, Prince and Geunes \cite{smith2013modern}.
Multilevel Interdiction is typically referred to as protection-interdiction or fortification-interdiction games.
In \cite{DBLP:journals/corr/abs-2406-01756}, different variants of knapsack are analyzed. Among them are interdiction knapsack, fortification-interdiction knapsack and multilevel fortification-interdiction knapsack.
Furthermore, the multilevel critical node problem is analyzed in \cite{DBLP:journals/jcss/NabliCH22}.

In the area of \emph{multi-stage adjustable robust optimization}, one is faced with uncertain scenarios.
The decision-maker has to make a multi-step decision, where in each step a partial decision has to be made without full knowledge of the scenario, until in the last stage a partial corrective decision can be made with full knowledge of the scenario.
Two-stage adjustable robust optimization has been considered for a large amount of standard optimization problems, for example in \cite{DBLP:conf/or/KasperskiZ15,kasperski2016robust,DBLP:journals/dam/KasperskiZ17,DBLP:journals/eor/ChasseinGKZ18,DBLP:journals/eor/GoerigkLW22}.
A general survey is provided by Yan{\i}ko{\u{g}}lu, Gorissen, and den Hertog \cite{DBLP:journals/eor/YanikogluGH19}.
Moreover, multi-stage adjustable robust optimization is discussed in the following papers:
the complexity of vertex cover, independent set, and traveling salesman is analyzed in \cite{DBLP:journals/dam/GoerigkLW24}.
Furthermore, multi-stage adjustable robustness is closely related with online optimization with estimates.
Corresponding results can be found in \cite{DBLP:journals/cor/GoerigkGISS15,DBLP:conf/stacs/GehnenLR24,DBLP:journals/corr/abs-2504-21750} for knapsack, \cite{DBLP:journals/corr/abs-2501-18496} for traveling salesman, and \cite{DBLP:journals/corr/abs-2505-09321} for bin packing.

\paragraph*{Closely related work}
Despite a tremendous interest in bilevel optimization, and despite the complexity classes $\Sigma^p_k$ being the natural complexity classes for this type of decision problems, we are aware of only a handful of publications on the matter of $\Sigma^p_k$-completeness with respect to these areas.
Rutenburg \cite{DBLP:journals/amai/Rutenburg93} shows $\Sigma^p_2$-completeness of the maximum clique interdiction problem.
Caprara, Carvalho, Lodi and Woeginger \cite{DBLP:journals/siamjo/CapraraCLW14} show $\Sigma^p_2$-completeness of an interdiction variant of knapsack that was originally introduced by DeNegre \cite{10.5555/2231641}.
Fröhlich and Ruzika \cite{DBLP:journals/tcs/FrohlichR21} show the $\Sigma^p_2$-completeness of two versions of a location-interdiction problem.
Nabli, Carvalho and Hosteins \cite{DBLP:journals/jcss/NabliCH22} study the multilevel critical node problem and prove its $\Sigma^p_3$-completeness.
Goerigk, Lendl and Wulf \cite{DBLP:journals/dam/GoerigkLW24} show $\Sigma^p_3$-completeness of two-stage adjustable variants of TSP, independent set, and vertex cover.
Grüne \cite{DBLP:conf/latin/Grune24} introduces a reduction framework for so-called universe gadget reductions to show $\Sigma^p_3$-completeness of recoverable robust versions of typical optimization problems.
Boggio Tomasaz, Carvalho, Cordone and Hosteins \cite{DBLP:journals/corr/abs-2406-01756} show $\Sigma^p_2$-completeness of an interdiction-knapsack problem and $\Sigma^p_k$-completeness of a $k$-level fortification-interdiction-knapsack problem.

The compendium by Umans and Schaefer \cite{schaefer2008completeness} contains many $\Sigma^p_2$-complete problems, but few of them are related to bilevel optimization.
A seminal paper by Jeroslow \cite{DBLP:journals/mp/Jeroslow85} shows $\Sigma^p_k$-completeness of general multilevel programs.

To the best of our knowledge, this completes the list of known $\Sigma^p_k$-completeness results for bilevel and multilevel optimization.
We find it remarkable that so few results exist in this area, despite being an active research area for well over two decades.
In contrast, $\NP$-completeness proofs are known for a huge number of problems.
One possible reason for this, according to Woeginger \cite{DBLP:journals/4or/Woeginger21}, is that \enquote{at the current moment, establishing $\Sigma^p_2$-completeness is usually tedious and mostly done via lengthy reductions that go all the way back to 2-Quantified Satisfiability}.

We end this subsection with a brief discussion of a related concept in complexity theory.
From a complexity-theoretic perspective, our notion of solution-embedding reductions is related in spirit to classical work on parsimonious reductions and $\sf \#P$-completeness \cite{Val79}.
Parsimonious reductions preserve the exact number of solutions and are central to the study of counting problems.
While both notions, parsimonious reductions and solution-embedding reductions, go beyond standard many-one reductions by preserving additional structure, they address different aspects of computational complexity.
In particular, $\sf \#P$-completeness is concerned with counting solutions, whereas our framework focuses on the behavior of solutions under restriction and alternation in multilevel decision problems.
As a result, solution-embedding reductions are suited to lifting $\NP$-completeness results to the polynomial hierarchy and $\PSPACE$, rather than to counting complexity.

\subsection{Technical Overview}
\label{sec:technicalOverview}

We give a short overview of the techniques and ideas used to obtain our main theorems.
For that purpose it becomes necessary to formally describe to what kind of problems the meta-theorem applies.

\subparagraph{\NPS.}
Problems from the complexity class $\NPS$ are tuples $(\I, \U, \sol)$.
Here, $\I \subseteq \bin^*$ is the set of input instances of the problem encoded as words in binary.
Associated to each input instance $I \in \I$, we assume that there is a \emph{universe} $\U(I)$,
and a \emph{solution set} $\sol(I) \subseteq \powerset{\U(I)}$ containing all subsets of the universe that make up a solution.

Our ideas are best explained using an example.
Consider the problem {\sc Vertex Cover}.
For an instance $I$, we are given an undirected graph $G = (V,E)$ and some threshold $t \in \Z_{\geq 0}$.
A vertex cover of a graph $G$ is a set $C \subseteq V$ of vertices such that each edge $\{u,v\} \in E$ has one endpoint in $C$, i.e., $u \in C$ or $v \in C$.
The question is whether there is a vertex cover of size at most $t$.
Interpreting the vertex cover problem as an $\NPS$ problem in the above sense means the following:
The input instance is given as a tuple $I = (G, t)$ (encoded in binary).
The universe associated to some input instance $I$ is given by $\U(I) = V$, and the solution set is  $\sol(I) = \set{C \subseteq V : C \text{ is a vertex cover in $G$ of size at most $t$}}$.
The question is whether $\sol(I) \neq \emptyset$.

\subparagraph{Solution-embedding Reductions.}
We, further, consider a type of $\NP$-completeness reduction from some problem $\Pi_1$ to some other problem $\Pi_2$ that has a special property, which we call the \emph{solution embedding property} (or, in short, SE property).
This property states that the universe of $\Pi_1$ can be injectively embedded into the universe of $\Pi_2$ in such a way that the following two properties hold: 
(P1) Every solution of $\Pi_1$ corresponds to a partial solutions of $\Pi_2$, 
and (P2) every solution of $\Pi_2$ when restricted to the image of the embedding corresponds to a solution of $\Pi_1$.
We show that that a natural modification of the historically first $\NP$-completeness reduction, (i.e., the reduction used by Cook and Levin to show that \textsc{Satisfiability} is $\NP$-complete) has the SE property (\cref{thm:cook-levin}).
We also show that a large number of $\NP$-completeness reductions that are known from the literature actually have the SE property (\cref{sec:ssp-reductions}).

Let us illustrate the SE property by considering a reduction from \textsc{3-Satisfiability} (\textsc{3Sat}) to \textsc{Vertex Cover}.
The input for \textsc{3Sat} is a Boolean formula $\varphi$ with clauses $c_1, \dots, c_m$. We consider the $\NPS$ version of \textsc{3Sat} in which the universe is the set $L = \fromto{\ell_1}{\ell_n} \cup \fromto{\overline \ell_1}{\overline \ell_n}$ of all literals.
The solution set of $\varphi$ is the set of all the subsets of the literals which encode a satisfying solution, i.e. 
$$
\sol(\varphi) = \set{L' \subseteq L : |L' \cap \set{\ell_i, \overline \ell_i}| = 1 \ \forall i \in \set{1,\dots,n}, L' \cap c_j \neq \emptyset \ \forall j \in \set{1, \dots, m}}.
$$

The following is the classical $\NP$-completeness reduction from \textsc{3Sat} to \textsc{Vertex Cover}  from the book of Garey and Johnson \cite{DBLP:books/fm/GareyJ79}.
Given a \textsc{3Sat} instance consisting of literals $\fromto{\ell_1}{\ell_n} \cup \fromto{\overline \ell_1}{\overline \ell_n}$ and clauses $C$, the reduction constructs a graph $G = (V,E)$ for \textsc{Vertex Cover} in the following way:
The graph contains vertices $W := \fromto{v_{\ell_1}}{v_{\ell_n}} \cup \fromto{v_{\overline \ell_1}}{v_{\overline \ell_n}}$ such that each vertex $v_{\ell_i}$ is connected to vertex $v_{\overline \ell_i}$ with an edge.
Furthermore, for each clause, we add a new triangle to the graph (including three new vertices) and add three additional edges so that the three vertices of the triangle are connected to the corresponding vertices of the literals appearing in the clause.  This reduction is illustrated in \Cref{fig:reduction:3sat-vertex-cover}.

Note that every vertex cover of $G$ has size at least $|L|/2 + 2|C|$ since it must contain at least one of the vertices in $\{v_{\ell_i},v_{\overline \ell_i} \}$ and at least two of the vertices in each triangle.
The following is easily verified:
$G$ has a vertex cover of size $|L|/2 + 2|C|$ if and only if  the \textsc{3Sat} instance is a Yes-instance.
We can see this as follows.
For the vertex cover to have cardinality $|L|/2 + 2|C|$ it must contain exactly one of the vertices in $\{v_{\ell_i},v_{\overline \ell_i} \}$ and exactly two of the vertices in each triangle.
Moreover, the third vertex in each triangle must be incident to a selected vertex not in the triangle.

\tikzstyle{vertex}=[draw,circle,fill=black, minimum size=4pt,inner sep=0pt]
\tikzstyle{edge} = [draw,-]
\begin{figure}[thpb]
\centering
\resizebox{0.67\textwidth}{!}{
\begin{tikzpicture}[scale=1,auto]

\node[vertex] (x1) at (0,0) {}; \node[above] at (x1) {$v_{\ell_1}$};
\node[vertex] (notx1) at (2,0) {}; \node[above] at (notx1) {$v_{\overline \ell_1}$};
\draw[edge] (x1) to (notx1);

\node[vertex] (x2) at (4,0) {}; \node[above] at (x2) {$v_{\ell_2}$};
\node[vertex] (notx2) at (6,0) {}; \node[above] at (notx2) {$v_{\overline \ell_2}$};
\draw[edge] (x2) to (notx2);

\node[vertex] (x3) at (8,0) {}; \node[above] at (x3) {$v_{\ell_3}$};
\node[vertex] (notx3) at (10,0) {}; \node[above] at (notx3) {$v_{\overline \ell_3}$};
\draw[edge] (x3) to (notx3);

\node[vertex] (c1) at (4,-2.25) {}; \node[below] at (c1) {$v^{c_1}_{\overline \ell_1}$};
\node[vertex] (c2) at (5,-1.25) {}; \node[above left] at (c2) {$v^{c_1}_{\overline \ell_2}$};
\node[vertex] (c3) at (6,-2.25) {}; \node[below] at (c3) {$v^{c_1}_{\ell_3}$};
\node[] at (5,-1.92) {$c_1$};
\draw[edge] (c1) to (c2) to (c3) to (c1);
\draw[edge] (notx1) to (c1);
\draw[edge] (notx2) to (c2);
\draw[edge] (x3) to (c3);

\node at ($(x1)+(-1,0)$) {$W$};
\draw[dashed,rounded corners] ($(x1)+(-.5,+.7)$) rectangle ($(notx3) + (.5,-.4)$);

\end{tikzpicture}
}
\caption{Classic reduction of \textsc{3Sat} to \textsc{Vertex Cover} for $\varphi = (\overline \ell_1 \lor \overline \ell_2 \lor \ell_3)$.}
\label{fig:reduction:3sat-vertex-cover}
\end{figure}

We now describe a way in which solutions of \textsc{3Sat} are mapped to solutions of \textsc{Vertex Cover}. Let 
$f : \U \to \U'$ with $f(\ell_i) = v_{\ell_i}$ and $f(\overline \ell_i) = v_{\overline \ell_i}$.
The function $f$ ``embeds'' the universe $\U$ into the universe $\U'$. It maps each literal in $\U$ to a unique vertex in $\U'$.  The function $f$ satisfies the following two properties: Suppose that $W = f(L)$. 
\begin{description}
    \item[(P1)] If $S \subseteq L$ are the literals for a satisfying assignment, then there is a vertex cover of size $|L|/2 + 2|C|$ that restricted to $W$ equals $f(S)$.
    \item[(P2)] If $V'$ is vertex cover of size $|L|/2 + 2|C|$, then $f^{-1}(V'\cap W)$ is the set of literals of a satisfying assignment.
\end{description}

We next give a more compact way of expressing the above two properties.
Let $\sol(\varphi) \subseteq \powerset{\U}$ denote the solutions of \textsc{3Sat} (i.e. the set of all subsets of the literals which encode a satisfying assignment).
Let $k = |L|/2 + 2|C|$, and let $\sol'(G, k) \subseteq \powerset{\U'}$ denote  the solutions of \textsc{Vertex Cover} (i.e. the set of all vertex covers of size at most $k$).
Then the above two properties of $f$ can be expressed in the following equivalent manner: 
\begin{equation}
\label{eq:SSP}
    \set{f(S) : S \in \sol(\varphi) } = \set{S' \cap f(\U) : S' \in  \sol'(G, k)}. 
\end{equation}

If an $\NP$-hardness reduction from $\Pi_1$ to $\Pi_2$ has an associated function $f$ that satisfies property \eqref{eq:SSP}, we say that it is a \emph{solution-embedding} (SE) reduction, and we denote it as $\Pi_1 \leqSE \Pi_2$.
Our meta-theorems are based on SE reductions.
We introduce the class $\NPS$ as an analogue to the class of $\NP$.  And we say that a problem $\Pi \in \NPS$ is $\NPS$-complete if for every other problem $\Pi' \in \NPS$, there is an SE reduction from $\Pi'$ to $\Pi$.  
\cref{sec:ssp-reductions} contains a list of $\NPS$-complete problems.
Since solution-embedding reductions are transitive, as we will show later, it suffices to choose an $\NPS$-complete problem $\Pi'$ from the list and then show that $\Pi' \leqSE \Pi$ to prove that $\Pi$ is $\NPS$-complete.

\subparagraph{The Meta-Theorems.}
Finally, we explain how SE reductions are used to obtain our meta-theorems.
We consider a modification operator $\Phi$ that transforms a problem from $\NPS$ into a multilevel decision problem, such as adversarial selection, protection-interdiction, or adjustable robust optimization.
For each such problem class, the completeness proof for $\Phi(\Pi)$ consists of two steps:
\begin{description}
    \item[(S1) The Sat-V Reduction.]
        We show that $\Phi(\textsc{Sat-V})$ is $\Sigma^p_k$-hard (or \PSPACE-hard) by a reduction from the adversarial selection game on \textsc{Sat-V}, which is the canonical $\Sigma^p_k$-complete problem QSAT$_k$.
    \item[(S2) The Meta-Reduction.]
        We lift the result from \textsc{Sat-V} to every $\NPS$-complete problem $\Pi$, using an SE reduction from \textsc{Sat-V} to $\Pi$ and a dual mimicking argument.
\end{description}
The technical work is concentrated in step (S1), which has to be carried out only once per problem class.
Step (S2) follows the same argument for all problem classes.
Hence our framework can be extended in two directions: 
establishing new $\NPS$-complete problems and determining new multilevel problem classes for which the corresponding \textsc{Sat-V} problem can be proved to be complete for its class.

\subparagraph{The Sat-V Reduction.}
For adversarial selection games, step (S1) is the classical result of Stockmeyer \cite{DBLP:journals/tcs/Stockmeyer76} and Wrathall \cite{Wra76} (\Cref{thm:adversarial-selection-completeness}), since the adversarial selection game on \textsc{Sat-V} is exactly QSAT$_k$.
For protection-interdiction games and adjustable robust problems, step (S1) is considerably more involved (\Cref{thm:combinatorial-interdiction-game-sat-completeness} and \Cref{thm:comb-adjustable-robust-sat-thm}).
These games seem less expressive than QSAT$_k$.
In QSAT$_k$, both players directly set the truth values of their variables, which is not the case in protection-interdiction games and adjustable robust problems.
In a protection-interdiction game, the interdictor can only block elements, which forces the corresponding variables to false, and the protector can only protect elements from being blocked.
Since the budgets are global, both players may distribute their tokens freely over the moves.
In the adjustable robust problem, the adversary can only block vulnerable elements, without exceeding its global budget.
The $k$ stages in an adversarial selection game contain alternating $2k-1$ moves.

To overcome this, we introduce a pseudo-complement $w$ for each variable $u$, so that blocking (or protecting) exactly one element of the pair $\{u, w\}$ encodes an assignment of $u$.
We first consider a restricted game in which the players must follow a set of rules that enforce this encoding, and show by a dual mimicking argument that the restricted game is equivalent to QSAT$_k$ (\Cref{par:restricted-interdiction-game-sat-completeness} and \Cref{par:restricted-adj-rob-game-sat-completeness}).
We then remove the rules using gadgets.
For this, we add cheat-detection gadgets that result in a loss of the first player that violates one of the rules.
The budgets are chosen to be exactly tight.
By the pigeonhole principle, every deviation leaves the opponent with an unblocked or unprotected pair that they can exploit.
As a consequence, optimal play in the unrestricted game is equivalent to optimal play in the game with restricted rules (\Cref{par:gadgets-interdiction-game-sat-completeness} and \Cref{par:gadgets-adj-rob-game-sat-completeness}).

\subparagraph{The Meta-Reduction.}
For step (S2), we utilize the structure of SE reductions (\Cref{Kmoveselection,sec:prot-int-meta,sec:adj-rob-meta}). 
Let $\Pi$ be $\NPS$-complete and let $(g, f)$ be an SE reduction from \textsc{Sat-V} to $\Pi$.
Given an instance of $\Phi(\textsc{Sat-V})$ with underlying formula $I$, we construct an instance of $\Phi(\Pi)$ on $g(I)$ by mapping all sets that restrict the moves of the players via $f$, keeping all budgets, and making the auxiliary universe $\mathcal{U}'_{\mathrm{aux}}$ accessible only in the last move.
All moves before the last move thus take place in the embedded universe, where $f$ yields a one-to-one correspondence between the moves in both games.

We then show that the same player wins both games by a dual mimicking argument:
to show that a winning strategy in $\Phi(\textsc{Sat-V})$ yields a winning strategy in $\Phi(\Pi)$, the winning player translates the optimal moves from $\Phi(\textsc{Sat-V})$ to $\Phi(\Pi)$ via $f$, while the opponent's moves in $\Phi(\Pi)$ are mimicked in $\Phi(\textsc{Sat-V})$ via $f^{-1}$.
The converse direction follows by exchanging the roles of the players.
The SE property is used exactly once, in the last move:
the choices $\hat S$ of the first $k-1$ moves can be completed to a solution in $\Phi(\textsc{Sat-V})$ if and only if $f(\hat S)$ can be completed to a solution in $\Phi(\Pi)$, where the completion in $\Phi(\Pi)$ may additionally use the auxiliary universe.
Here, (P1) yields one direction and (P2) the other.

\section{Preliminaries}
\label{sec:prelim}

A \emph{language} is a set $L\subseteq \bin^*$.
A \emph{many-one-reduction} or \emph{Karp-reduction} from a language $L$ to a language $L'$ is a map $f : \bin^* \to \bin^*$ such that $w \in L$ iff $f(w) \in L'$ for all $w \in \bin^*$. 
A Turing machine $M$ together with an input $x$ of a decision problem $L$ is associated to a predicate $M(x)$ such that $M(x) = 1$ if $x \in L$, and $M(x) = 0$ if $x \notin L$.

A \emph{Boolean variable} $x$ is a variable that takes one of the values 0 or 1. Let $X = \fromto{x_1}{x_n}$ be a set of variables. The corresponding \emph{literal set} is given by $L = \fromto{x_1}{x_n} \cup \fromto{\overline x_1}{\overline x_n}$.
A \emph{clause} is a disjunction of literals. 
A Boolean formula is in \emph{conjunctive normal form} (CNF) if it is a conjunction of clauses.
It is in \emph{disjunctive normal form} (DNF), if it is a disjunction of conjunctions of literals.
We use the convention that a clause may be represented by a subset of the literals, and a CNF formula $\varphi$ may be represented by a set of clauses. 
For example, the set $\set{\set{x_1, \overline x_2},\set{x_2, x_3}}$ corresponds to the formula $(x_1 \lor \overline x_2)\land (x_2 \lor x_3)$.
We write $\varphi(X)$ to indicate that formula $\varphi$ depends only on $X$.
Sometimes we are interested in cases where the variables are partitioned. We indicate this case by writing $\varphi(X,Y,\ldots)$.
An \emph{assignment} of variables is a map $\alpha : X \to \bin$.
The evaluation of the formula $\varphi$ under assignment $\alpha$ is denoted by $\varphi(\alpha) \in \bin$.
We also denote a partial assignment $\alpha$ on $X$ of a formula with partitioned variables $\varphi(X, Y)$ with $\varphi(\alpha, Y)$.
We further denote $\alpha \vDash \varphi$ if $\alpha$ satisfies formula $\varphi$; that is, $\varphi(a)=1$.

A cost function on a set $A$ is a mapping  $c : A \to \R$.
For each subset $A' \subseteq A$, we define the cost of the subset $A'$ as $c(A') := \sum_{a \in A'} c(a)$.
For a map $f : A \to B$ and some subset $A' \subseteq A$, we define the image of the subset $A'$ as $f(A') = \set{f(a) : a \in A'}$. For some $B' \subseteq B$, we denote its preimage by $f^{-1}(B') = \set{a \in A : f(a) \in B'}$.
In this paper, it will be apparent whether a function is a cost function or not, and so the meaning of the notation will be clear from context.
If $A_1, \ldots, A_n$ is a collection of sets, we denote
$\widehat A_i =  \bigcup^i_{j=1} A_j.$

\section{The complexity class NP-S} 
\label{sec:npsandlop}

For problems in the class $\NP$, solutions are often described implicitly.  In this section, we consider an alternative in which the solutions are explicit. Here we define the class {\it $\NP$ with solutions}, or, in short, $\NPS$.
This class is an enhanced version of the well-known class $\NP$ by semantically adding solutions to the instances.
The class $\NPS$ is the basis for subsequent results in this paper. We will establish inherent relations between the $\NP$-completeness of problems and its variants in the polynomial hierarchy.
We will provide two alternative encodings of problems in the class $\NPS$. In one of the encodings, each problem in $\NPS$ is also a problem in $\NP$, thus showing that $\NPS$ may be viewed as a subset of $\NP$.

We first recall a standard definition of $\NP$ and then explain how to modify it to define $\NPS$.
A common perspective of the class $\NP$ is to view it as a collection of decision problems, replete with yes-instances and no-instances.
With that perspective, the complexity class $\NP$ (nondeterministic polynomial time) consists of all decision problems for which a solution -- also called a certificate or witness -- can be verified in polynomial time by a deterministic Turing machine.
Formally, $\NP$ is defined as a collection of languages. 

\begin{definition}[Nondeterministic Polynomial Time]
A language $ L \subseteq \{0,1\}^* $ is in $\NP$ if there exists a deterministic Turing machine $V$ that runs in polynomial time and a polynomial $p(n)$ such that for every string $I \in \{0,1\}^*$,
$$
    I \in L \iff \exists y \in \{0,1\}^* \text{ with } |y| \leq p(|I|) \text{ such that } V(I, y) = 1.
$$
\end{definition}

In this formulation, an encoding of solution $y$ is a binary vector of at most polynomial length.
If desired, certificates could be restricted to have length exactly $p(|I|)$; such details do not affect the class definition.
Furthermore, each solution encoding $y$ can be verified to be a correct solution by a pre-specified deterministic polynomial-time Turing machine $V$, which we also call the \enquote{verifier}.
W.l.o.g., we consider only verifiers that accept strings $y$ of length exactly $p(|I|)$.

We will soon formally define the class {\it Nondeterministic Polynomial Time with Solutions} ($\NPS$) by associating each instance with a set of solutions. 
First, the class $\NPS$ is a refinement of the class $\NP$ in which each certificate of an instance $I$ is represented as a subset of a \emph{universe} $\U(I)$ of \emph{universe elements}.
Without loss of generality, we can represent certificates as subsets of a universe rather than a string in $\{0,1\}^*$. 
For any binary vector $ y \in \{0,1\}^{p(|I|)} $, we can define the corresponding solution $S(y) = \{ u_j \in \U(I) : y_j = 1 \}$.
If we can verify that $y$ is a certificate in polynomial time using a deterministic Turing machine that verifies $V$, 
then we can also recognize subsets of $\U(I)$, and verify that $S(y)$ is a solution in polynomial time using the same Turing machine $V$ (or an easily modified version of $V$).
This subset-based formulation is useful for technical clarity.  It will also be useful for expressing transformations and reductions between problems.
Furthermore, we let $\sol(I)$ denote the set of solutions.
That is, 
$$
     \sol(I) = \{S(y) \subseteq \U(I): V(I, y) = 1\}.
$$
In each case, the question is ``Is $\sol(I) = \emptyset$?''

The universe is implicitly defined by the verifier $V$.
For each instance $I$, we can denote the universe as $\U_V(I)$.
Similarly, the set of solutions $\sol(I)$ of each instance $I$ is implicitly defined by the verifier $V$.
As such, we can denote it as $\sol_V(I)$.

When describing a decision problem in the class $\NPS$, we will specify the universe and describe the set of solutions.
The following are two decision problems in $\NP$ as expressed as problems in $\NPS$. 

\begin{description}
    \item[]\textsc{Hamiltonian Cycle}\hfill\\
    \textbf{Instance  $I$:} Graph $G = (V, E)$\\
    \textbf{Universe  $\U(I)$:} Edge set $E$\\
    \textbf{Solution set $\sol(I)$:} The subsets of edges $S \subseteq E$ that are Hamiltonian cycles in $G$.
\end{description}

\begin{description}
    \item[]\textsc{Clique}\hfill\\
    \textbf{Instance $I$:} Graph $G = (V, E)$, number $k \in \N$.\\
    \textbf{Universe $\U(I)$:} Vertex set $V$\\
    \textbf{Solution set $\sol(I)$:} The subsets $S \subseteq V$ with $|S| \geq k$, such that $\forall i, j \in S \text{ with } i \neq j$, we have $\set{i, j} \in E$.
\end{description}

If a problem in $\NP$ can be expressed with two different universes, they are considered different problems in $\NPS$, even though they represent the same problem in $\NP$.
For example, one can express the clique problem in which the universe is the set of edges rather than the set of vertices.
This version of the clique problem is a different problem in $\NPS$ than the clique problem given above.
In the remainder of the paper, for purposes of clarity, we assume that the semantics of the universe are consistent for all instances of a given problem.
For example, when we consider an $\NPS$ variant of the clique problem, either the universe for every instance is a vertex subset or the universe for every instance is an edge subset.

We represent the \emph{lifting} of a language $L \in \NP$ to a problem $\Pi$ in $\NPS$ by associating a verifier $V$ to it.
We then define $\Pi = (L, \U_V, \sol_V)$ where $\sol_V$ is the solution set such that $V(I, y) = 1$ for all $S(y) \in \sol_V$ (and $V(I, y) = 0$ otherwise).
In order to abstract away from unnecessary details on the actual encoding, we denote the problem $\Pi$ by the triple $(\I, \U, \sol)$, where $\I = L$ is a collection of problem instances, $\U = \U_V$ is a collection of universes, and $\sol = \sol_V$ is a collection of solution sets.
We remark that we drop the subscript $V$ from $\U_V$ and $\sol_V$ when the verifier is known or obvious from the context. We do so only when the verifier is fixed throughout the discussion.
Furthermore, if there are two different verifiers $V$ and $V'$ such that $\U_V = \U_{V'}$ and $\sol_V = \sol_{V'}$, we consider the corresponding $(\I, \U_V, \sol_V)$ and $(\I, \U_{V'}, \sol_{V'})$ to be the same problem.

\begin{definition}[Non-deterministic Polynomial Time with Solutions]
\label{def:NP-S}
A problem $\Pi = (\I, \U, \sol)$ is in $\NPS$ if there is a deterministic polynomial time Turing machine $V$, such that
\begin{itemize}
    \item $\I \subseteq \set{0,1}^*$ is a language. We call $\I$ the set of instances of $\Pi$. 
    \item To each instance $I \in \I$, there is some set $\U(I)$ which we call the universe associated to the instance $I$ such that $|\U(I)| = |I|^{O(1)}$, where $|I|$ denotes the length of the input string. 
    \item To each instance $I \in \I$, there is some (potentially empty) set $\sol(I) \subseteq \powerset{\U(I)}$ that we call the solution set associated to the instance $I$.
    \item For each $C \subseteq \U(I)$, $V(I, C) = 1$ if and only if $C \in \sol(I)$.
\end{itemize}
\end{definition}

Although the solution set $\sol(I)$ is part of the problem representation, it is not part of the input for a problem in $\NPS$.
Rather, the verifier $ V $ encodes the constraints that determine $ \sol(I) $.
The question for a problem in $\NPS$ is: Given instance $I \in \I$, is $\sol(I) \neq \emptyset$?

\Cref{def:NP-S} is the definition that we will use throughout the paper.
We use this definition to abstract away the universe and solution set from its encodings.
This inherently makes it easier to argue about the properties of the problems and handle reductions between them in a more understandable way.
For this reason, when we represent a problem in $\NPS$, we will describe the universe and the solution set.
In this representation, $\NPS$ is an extension of the class $\NP$.

We also consider another representation  of problems in $\NPS$.
In this second representation, each problem instance is represented by the pair $(I, V)$, where $V$ is the polynomial time Turing machine.
Since we derive each problem $(\I, \U, \sol)$ from tuples $(\I, \U_V, \sol_V)$, where the verifier $V$ specifies the universe $\U_V$ and the solution set $\sol_V$, we can represent each problem by just associating the verifier $V$ to each of the instances. 
In this representation, the class $\NPS$ is also an extension of the class $\NP$.
However, based on this representation, we can exploit an aspect of Turing's creation of a universal Turing machine.
As part of this process, Turing showed how to associate a unique integer, the standard description number, with each Turing machine.
By embedding the standard description number into the instance, we are able to show the following lemma.

\begin{lemma}
\label{lem:subsetNP}
    Each problem $\Pi$ in $\NPS$ can be transformed into a language $L_\Pi$ in $\NP$.
\end{lemma}
\begin{proof}
    We construct a language $L_\Pi$ that is contained in $\NP$ by encoding each instance $I \in \I$, the associated universe $\U(I)$, and the associated solution set $\sol(I)$.
    
    For this, we construct a word $w$ for each instance $I$.
    First, we encode the instance $I \in \I \subseteq \{0,1\}^*$ directly by the existing encoding.
    Second, we encode the solution set $\sol(I) = \sol_V(I)$ implicitly by using the standard description number $\langle V \rangle$ of the verifying Turing machine $V$.
    This verifier $V$ only accepts certificates that correspond to a subset of universe elements in the solution set by definition.
    Thus, we also implicitly encode the universe $\U(I)$ with the standard description number $\langle V \rangle$ of verifier $V$.  For any fixed problem $\Pi$, and its corresponding verifier $V$, we have $|\langle V \rangle| = O(1)$.
    Now, we define the word $w = I\#\langle V \rangle$ associated to instance $I$, where $\#$ is a separator symbol, resulting in the language
    $$
        L_\Pi = \{I\#\langle V \rangle \mid I \in \I,\ \sol(I) \neq \emptyset\}
    $$
    as an encoding of the $\NPS$ problem $\Pi = (\I, \U, \sol)$.
    
    In order to prove that $L_\Pi \in \NP$, it suffices to show that there is a deterministic polynomial-time verifying Turing machine $V_\Pi$ for the language $L_\Pi$ such that
    $$
        w \in L_\Pi \iff \exists y \in \{0,1\}^{|\U(I)|} \ V_\Pi(w,y) = 1.
    $$
    In order to recognize the language $L_\Pi$, the machine $V_\Pi$ proceeds as follows. Given some proposed word $w\in\set{0,1}^\star$ and proposed certificate $y \in \set{0,1}^\star$,  $V_\Pi$ first checks whether $w$ is indeed of the form $I\#\langle V \rangle$.
    This is doable because by definition of $\NPS$ the problem $\Pi$ is associated with the fixed machine $V$, and hence $\langle V \rangle$ is a globally fixed string, that is independent of the instances $I$ and has length $O(1)$.
    Second, $V_\Pi$ simulates $V$ by using a universal Turing machine as sub-procedure on the given instance word $I$ and certificate $y$.
    $V_\Pi$ accepts the input if and only if $V(I,y) = 1$ for the simulated Turing machine.
    In other words, the certificate $y$ corresponds to a subset of universe elements that is in the solution set. 
    For a fixed instance $I$ there exists a certificate $y$ that makes $V_\Pi$ accept if and only if there exists a subset $S \subseteq \U(I)$ that makes $V$ accept, if and only if $w \in L_\Pi$.
    It follows that $V_\Pi$ is a polynomial-time verifier for $L_\Pi$. It follows that $L_\Pi \in \NP$.

    Note that if we have two different verifiers $V$ and $V'$ that define the same solution set $\sol_V = \sol_{V'}$, we are able to reduce $L_\Pi$ with $\sol_V$ to $L_\Pi$ with $\sol_{V'}$ by replacing the standard description number of $V$ by that of $V'$ in at most linear time.
\end{proof}

\paragraph*{Different Universes induce a different problem in $\NPS$}

In our framework, each problem instance is accompanied by a universe $\U(I)$, and solutions are subsets of $\U(I)$.
As mentioned in the introduction, if there are two distinct ways of expressing the universe, this leads to two different problems in $\NPS$.

For example, if we model a decision problem in $\NP$ as a 0-1 integer program (IP), then we can associate the universe $\U(I)$ to the decision variables of the IP.
That is, for a given vector $x_1, \ldots, x_n$ of binary integer variables, the universe can be represented as $\fromto{1}{n}$.
If a single problem can be modeled in two distinct ways as a 0-1 integer program using different sets of decision variables, then these two formulations will result in two different problems in $\NPS$.
An analogous remark is true if the decision problem is modeled in two distinct ways as a Boolean formula.
This distinction plays a critical role when we lift problems in $\NPS$ to problems in $\PH$ or $\PSPACE$.
If two extensions of a problem from $\NP$ to $\NPS$ have different universes, the lifted problems are different problems.
They may even have different complexities.

\textit{Example: Clique.}
The clique problem has two natural choices of a universe, as we illustrate next.
We first formulate {\sc Max-Clique} as an integer program.
Rather than a vertex-based formulation, it is a formulation that includes variables for edges and variables for vertices.
We let $x_i=1$ if vertex $i$ is in the clique.
Further, we let $x_{ij}$ be 1 if edge $(i, j)$ is in the clique.
The objective is to maximize the number of edges in the clique, but it would also be acceptable to maximize the number of vertices in the clique.
\begin{align*}
    \textsc{Max-Clique}\\
    \max \quad & \sum_{(i,j) \in E} x_{ij} && \\
    \text{s.t.} \quad & x_i + x_j \leq 1 && \forall \{i,j\} \notin E\\
        & x_{ij} \leq x_{i}  && \forall (i,j) \in E \\
        & x_{ij} \leq x_{j}  && \forall (i,j) \in E \\
        & x_{ij} \in \{0,1\} && \forall i,j \in E\\
        & x_{i} \in \{0,1\} && \forall i \in V
\end{align*}
If we use the variables $x_i$ as our universe, we obtain the problem \textsc{Clique (Vertex-based)} with $\U = V$.
If we use the variables $x_{ij}$, we obtain the problem \textsc{Clique (Edge-based)} with $\U = E$. 
It would also be possible to use the combination of $x_i$ and $x_{ij}$ as a universe $\U = V \cup E$.

\textit{Example: Satisfiability.}
To illustrate, consider the satisfiability problem for Boolean formulas.
The satisfiability problem is fundamental to theoretical computer science, and plays an important role in each of the following sections of this paper.
There are at least two natural ways to define the universe:

\textsc{Satisfiability (Variable-based)} (in short, {\sc Sat-V}).
The universe consists of variables $\U = \{x_1, \dots, x_n\}$.
A solution consists of the variables of a satisfying assignment set to true; i.e., if the truth assignment $\alpha$ is satisfying, then $S(\alpha) = \set{x_i: \alpha(x_i) = 1}$.

\textsc{Satisfiability (Literal-based)} (in short, {\sc Sat-L}).
The universe consists of literals, i.e., $ \U = \{\ell_1, \overline \ell_1, \dots, \ell_n, \overline \ell_n\} $.
A solution consists of the literals of a satisfying assignment; i.e., if the truth assignment $\alpha$ is satisfying, then $S(\alpha) = \set{\ell_i: \alpha(x_i) = 1} \cup \set{\overline \ell_i : \alpha(x_i) = 0}$.

Because their universes (and solution sets) differ, {\sc Sat-V} and {\sc Sat-L} are distinct problems in $\NPS$.

\paragraph*{Different verifiers may induce different problems in $\NPS$: Partition}

Let $\I$ be the instances of a language in $\NP$. 
Suppose that $V$ and $V'$ are two different verifiers for the pair of instances $\I$, and suppose that there is a common universe.  That is, $\U_V = \U_{V'}$.  Let $\U$ denote the common universe.  If $\sol_V \neq \sol_{V'}$, then $\Pi = (\I, \U, \sol_V)$ and $\Pi' = (\I, \U, \sol_{V'})$ in $\NPS$ are different problems in $\NPS$.  This distinction can be quite important, as illustrated by the following two variants of the Partition Problem.  

\begin{description}
    \item[]\textsc{Partition 1}\hfill\\
        \textbf{Instance $I$:} Numbers $\fromto{a_1}{a_n} \subseteq \N$.\\
        \textbf{Universe $\U(I)$:} $\fromto{a_1}{a_n}$.\\
        \textbf{Solution set $\sol_1(I)$:} The set of all subsets $S \subseteq \U(I)$ such that $\sum_{a_i \in S}a_i = \sum_{a_j \notin S}a_j$.
\end{description}

\begin{description}
    \item[]\textsc{Partition 2}\hfill\\
        \textbf{Instance $I$:} Numbers $\fromto{a_1}{a_n} \subseteq \N$.\\
        \textbf{Universe $\U(I)$:} $\fromto{a_1}{a_n}$.\\
        \textbf{Solution set $\sol_2(I)$:} The set of all subsets $S \subseteq \U(I)$ such that $\sum_{a_i \in S}a_i = \sum_{a_j \notin S}a_j$ and $a_n \in S$.
\end{description}
Note that whenever the answer is ``yes'' for {\sc Partition 1}, it is also ``yes'' for {\sc Partition 2}.
When viewed as descriptions of problems in \NP, they are the same problem.
The only difference between these two problems is that the solutions in $\sol_2(I)$ are required to contain the element $a_n$, whereas the solutions in $\sol_1(I)$ are not.
Therefore, when viewed as descriptions of the problems in $\NPS$, these two problems are different.

This distinction has an important consequence.
In the next subsection \ref{NP-S-completeness}, we define the concept of $\NPS$-completeness.
We will show in \Cref{sec:alternative-NPS-models} that {\sc Partition 2} is \NPS-complete, whereas {\sc Partition 1} is not \NPS-complete.

\subsection{\NPS-completeness}
\label{NP-S-completeness}
As a central question, we want to classify the hardest problems in the class $\NPS$.
But what do we mean by the hardest problems?

For this, recall that $\NP$-completeness is defined on the basis of the polynomial-time many-one reductions, which can be adapted to $\NPS$ problems as follows.
Given two problems in $\NPS$, say $\Pi = (\I,\U,\sol)$ and $\Pi' = (\I',\U',\sol')$.
Then a polynomial time many-one reduction from $\Pi$ to $\Pi'$ is a function $g : \I \to \I'$ computable in polynomial time in the input size $|I|$, such that $I$ is a yes-instance of $\Pi$ iff $g(I)$ is a yes-instance of $\Pi'$ (i.e. $\sol(I) \neq \emptyset \iff \sol'(g(I)) \neq \emptyset$).
The universe $\U$ and the solution set $\sol$ of $\Pi$ are mapped implicitly by $g$ to the corresponding universe $\U'$ and solution set $\sol'$ of $\Pi'$.

Our objective is to establish completeness proofs concerning multilevel optimization problems as well as other mappings from $\NPS$ to $\PH$ or even $\PSPACE$.
To accomplish this objective, we next define a special class of polynomial time many-one reductions that we refer to as ``solution-embedding''.
Solution-embedding reductions satisfy the following additional property, which allows us to upgrade reductions on the level of $\NPS$ to reductions between the lifted 
in $\PH$ or $\PSPACE$:
A solution-embedding reduction reduces a problem $\Pi$ to a problem $\Pi'$, such that all solutions to an instance of problem $\Pi$ are transformed to a partial solution in the instance of problem $\Pi'$ and, on the other hand, all solutions to an instance $\Pi'$, are transformed back to a solution in the instance of $\Pi$.
We further define a problem to be $\NPS$-complete if is contained in $\NPS$ and there is a polynomial time solution-embedding reduction from all problems in $\NPS$.
We will explore the implications of this concept in the subsequent sections.

\begin{definition}[Solution-embedding Reduction]
\label{def:se-reduction}
    Let $\Pi = (\I,\U,\sol)$ and $\Pi' = (\I',\U',\sol')$ be two problems in {\sf NP-S}.
    We say that there is a solution-embedding reduction from $\Pi$ to $\Pi'$, and write $\Pi \leqSE \Pi'$, if
    \begin{itemize}
        \item There exists a function $g : \I \to \I'$ computable in polynomial time in the input size $|I|$, such that $\sol(I) \neq \emptyset$ iff $\sol'(g(I)) \neq \emptyset$.
        \item There exist functions $(f_I)_{I \in \I}$ uniformly computable in polynomial time in $|I|$ such that for all instances $I \in \I$, $f_I : \U(I) \to \U'(g(I))$ is an injective function mapping from the universe of the instance $I$ to the universe of the instance $g(I)$ such that the following hold:
        \begin{flalign}
            \forall S \in \sol(I) : \ 
            \exists S' \in \sol'(g(I)) : \ &
            f_I(S) = S' \cap \Uembp \tag{P1} \label{eq:SE:forward}\\
            \forall S' \in \sol'(g(I)) : \ 
            \exists S \in \sol(I) : \ &
            f^{-1}_I(S' \cap \Uembp) = S \tag{P2} \label{eq:SE:backward},
        \end{flalign}
        where $\Uembp = f_I(\U(I))$
        and $\Uauxp = \U'(g(I)) \setminus \Uembp$.
    \end{itemize}
\end{definition}

A solution-embedding reduction embeds the universe of the first problem $\Pi$ into the universe of the second problem $\Pi'$, while leaving the solution structure untouched.
That is, there is a one-to-one correspondence on the solution level between a universe element $u \in \U(I)$ in $\Pi$ and the corresponding solution element $f(u)$ in $\Pi'$.
As a naming convention, we call $\Uembp = f_I(\U(I))$ the {\it embedded universe}, which consists of all the elements that already occurred in the instance of problem $\Pi$.
Furthermore, elements that were not part of the instance of $\Pi$ and thus were introduced to ensure that the one-to-one correspondence between the elements in $\U(I)$ and $\Uembp = f_I(\U(I))$ can be extended to a correct full solution for the instance of $\Pi'$ is called the {\it auxiliary universe} and is defined by $\Uauxp = \U' \setminus \Uembp$.
\Cref{fig:SE-reduction:universe-embedding} schematically depicts how the universes of two $\NPS$ problems are related under an SE reduction.

\begin{figure}[!ht]
    \centering
    \resizebox{0.67\textwidth}{!}{
        \begin{tikzpicture}
            \draw[] ($(-5,0)$) rectangle ($(-2,-2)$);
        	\node[] () at (-4.5,-1.7) {$\U(I)$};

            \draw[dashed] ($(2,0)$) rectangle ($(5,-2)$);
            \draw[] ($(0,0.5)$) rectangle ($(5.5,-3)$);
        	\node[] () at (3.35,-1.7) {$f_I(\U(I)) = \Uembp$};
        	\node[] () at (2,-2.7) {$\U'(g(I)) = \Uembp \dot\cup~\Uauxp$};

            \draw[->, thick] (-2,-0.7) -- (2,-0.7);
            \node[above] () at (-1,-0.7) {$f_I$};
        \end{tikzpicture}
    }
    \caption{
        The relation between the universes by applying an SE reduction between the problem $\Pi = (\I, \U, \sol)$ and $\Pi' = (\I', \U', \sol')$ for a given instance $I \in \I$.
        Let $I \in \I$ be that instance of $\Pi$, then the SE reduction $(g, (f_I)_{I \in \I})$ embeds the universe $\U(I)$ into the universe $\U'(g(I))$ of problem $\Pi'$ such that $f_I(\U(I)) = \Uembp \subseteq \U'(g(I))$.
        Note that $g(I)$ is the instance defined by the usual reduction mapping $g$.
        The function $f_I$ maintains a one-to-one correspondence between the elements of $\U(I)$ and $\Uembp$.
    }
    \label{fig:SE-reduction:universe-embedding}
\end{figure}

One can express \Cref{eq:SE:forward} and \Cref{eq:SE:backward} together equivalently by using the single equation
\begin{equation}
    \set{f_I(S) : S \in \sol(I) } = \set{S' \cap \Uembp : S' \in  \sol'(g(I))}. \tag{SE Property}\label{eq:SE}
\end{equation}
We call this \Cref{eq:SE} or equivalently \ref{eq:SE:forward} together with \ref{eq:SE:backward} the \emph{solution-embedding property}, or, in short, {\it SE property}.
Definition \ref{def:se-reduction} describes SE reductions in terms of solutions of $\Pi$ and $\Pi'$, but it could have been expressed in terms of partial solutions.
Let $I$ be an instance of problem $\Pi = (\I,\U,\sol) \in \NPS$.
We say that $S$ is a \emph{partial solution} of $\sol$ with respect to $T$ if (1) $S \subseteq T \subseteq \U(I)$, and (2) there is a subset $S' \subseteq \U(I) \setminus T$ such that $S \cup S' \in \sol(I)$.

\begin{lemma}
\label{lem:se}
    Let $\Pi = (\I,\U,\sol)$ and $\Pi' = (\I',\U',\sol')$ be two problems in {\sf NP-S}, and let $(g, f)$ be a solution-embedding reduction from $\Pi$ to $\Pi'$.
    Then for each instance $I \in \I$, and for all subsets $S \subseteq T \subseteq \U(I)$, the following is true.
    Subset $S$ is a partial solution of $\sol(I)$ with respect to $T$ if and only if $f(S)$ is a partial solution of $\sol'(I')$ with respect to $f(T)$.
\end{lemma}
\begin{proof}
    Suppose first that $S$ is a partial solution of $I$ with respect to $T$.
    Let $S^* \in \U(I)\setminus T$ be selected so that $S \cup S^* \in \sol(I)$.
    Because the reduction is solution-embedding, there is some subset $V \subseteq \U'_{aux}(I)$ such that $f(S \cup S^*) \cup V \in \sol'(I')$.
    Equivalently, $f(S) \cup (f(S^*) \cup V) \in \sol'(I')$, and thus $f(S)$ is a partial solution of $I'$ with respect to $f(T)$.

    Suppose instead that $f(S)$ is a partial solution of $I'$ with respect to $f(T)$.
    Let $V_1 \subseteq f(\U(I))\setminus f(S)$, and $V_2 \subseteq \U'_{aux}(I')$ be selected so that $f(S) \cup V_1 \cup V_2 \in \sol'(I')$.
    If follows that $f^{-1} (f(S) \cup V_1) \in \sol(I)$.
    Equivalently,  $S \cup f^{-1} (V_1) \in \sol(I)$, and thus $S$ is a partial solution of $I$ with respect to $\U(I)$.
\end{proof}

As we have already shown that $\NPS$ is a subset of $\NP$, we can also naturally expand the concept of completeness from $\NPS$ to $\NP$.

\begin{lemma}
    Every $\NPS$-complete problem is $\NP$-complete.
\end{lemma}
\begin{proof}
    By \Cref{lem:subsetNP}, we know that an $\NPS$ problem $\Pi$ can be encoded such that the corresponding language $L_\Pi$ is in $\NP$.
    It remains to show that a solution-embedding reduction from $\NPS$ problem $\Pi = (\I, \U, \sol)$ to $\NPS$ problem $\Pi' = (\I', \U', \sol')$ also yields a reduction between the corresponding languages $L_\Pi = \{I\#\langle V \rangle \mid I \in \I, \sol(I) \neq \emptyset\}$ and $L_\Pi' = \{I'\#\langle V' \rangle \mid I' \in \I', \sol'(I') \neq \emptyset\}$.
    Let $(g, (f_I)_{I \in \I})$ be such an SE reduction from $\Pi$ to $\Pi'$.
    Now, one can set $I' = g(I)$. We have by the first property of SE reductions that $\sol(I) \neq \emptyset$ iff $\sol(I') \neq \emptyset$.
    The standard description numbers of the verifiers $\langle V \rangle$ and $\langle V' \rangle$ are fixed numbers that exist by definition of $\NPS$.  Since $\Pi$ is a given problem in $\NPS$, these are computable in $O(1)$ time.

    Observe that we do not need the function family $f_I$ to show $\NP$-completeness.
    Nevertheless, the function $f_I$ is implicitly used in the mapping from the verifier $V$ to the verifier $V'$.
    More precisely, a certificate $y$ that is accepted by $V$ is a partial certificate for $V'$.
    That is, each certificate $y \in \{0,1\}^{|\U(I)|}$ that is accepted by $V$ is extendable by a word $z \in \{0,1\}^{|\U'(g(I))| - |\U(I)|}$ to a certificate $y' = yz \in \{0,1\}^{|\U'(g(I))|}$ that is accepted by $V'$.
\end{proof}

As another nice property of the class $\NPS$, we can show that corresponding satisfiability versions are $\NPS$-complete by reusing the proof of the Cook-Levin Theorem.
Their construction allows us to identify a set of Boolean variables that correspond to the certificates of the corresponding instance of the problem at hand.
Thus, we are able to show that all $\NPS$ problems are indeed reducible to satisfiability (with literal or variable universe) by a solution-embedding reduction.

\begin{theorem}[Cook-Levin Theorem Adapted to $\NPS$]
\label{thm:cook-levin}
 \textsc{Satisfiability} (with literal universe or variable universe) is $\NPS$-complete with respect to polynomial-time SE reductions.
\end{theorem}
\begin{proof}    
    We consider the original proof by Cook and show that it is actually a polynomial-time SE reduction.
    Let $\Pi = (\I, \U, \sol)$ be an arbitrary problem in $\NPS$ with universe $\U(I)$ and solution set $\sol(I)$ associated to each instance $I \in \I$ of $\Pi$.
    We have to show that $\Pi \leqSE \textsc{Satisfiability}$. 
    In the following, we show the completeness for \textsc{Sat-V} -- the proof for the completeness of \textsc{Sat-L} is analogous by adapting the universe accordingly.
    Recall that $\textsc{Sat-V} = (\I', \U', \sol')$, where $\I'$ is the set of \textsc{Sat}-instances, and for each formula $\varphi \in \I'$, the set $\U'(\varphi)$ is its variable set, and $\sol(\varphi)$ is the set of subets of variables that when set to true correspond to satisfying assignments.

    Since $\Pi$ is in $\NPS$, due to \cref{def:NP-S}, there exists a deterministic polynomial-time Turing machine $V$, called the verifier, such that for each $X \subseteq \U(I)$, $V(I, X) = 1$ if and only if $X \in \sol(I)$. 
    More precisely, we can assume that this Turing machine $V$ receives as input $n$ binary variables $x_1,\dots,x_n$ encoding the set $X \subseteq \U(I)$ and $n$ binary variables $i_1,\dots,i_n$ describing the instance $i$. (For simplicity, we assume w.l.o.g.\ that the instance and universe description have the same size).
     
    Now, the proof of Cook implies that we can introduce polynomially many auxiliary variables $z_1,\dots,z_{p(n)}$ and compute in polynomial time a Boolean expression $\varphi(x,i,z)$ such that $\varphi(x,i,z)$ captures exactly the behavior of the Turing machine $V$. 
    Let $u_j$ the $j$-th element of universe $\U(I)$. The main idea is that the solution-embedding function will be given by $f(u_j) := x_j$ for all $j \in \{1,\ldots,n\}$, that is, we map the $j$-th universe element of $\U(I)$ to the variable that represents its non-deterministic {\sc Sat} certificate inside Cook's construction.
    It is not important to re-iterate all details of Cook's construction. Instead it suffices to note the following key details.
    
    \begin{itemize}
        \item   For a given instance $I$, and some set $X \subseteq \U(I)$, let $i_I$ denote the encoding of $I$ and $x_X$ denote the encoding of $X$.
        Let the $z$-variables be free.
        Then the partial formula `$\varphi(x_X,i_I,z)$' is satisfiable if and only if $V$ accepts $(I,X)$.
        \item  The partial formula `$\varphi(x,i_I,z)$' has size polynomial in $|I|$ and can be constructed in polynomial time from $I$.
    \end{itemize}
    
    We now claim that this reduction by Cook immediately yields a polynomial-time SE reduction $(g, (f_I)_{I \in \I})$.
    Formally, we let $g(I) := \varphi(x,i_I,z)$ where $i_I$ is the encoding of $i_I$ and $x,z$ are free variables.
    Note that by the properties of Cook's reduction, $I$ is a Yes-instance of $\Pi$ if and only if $\exists x,z \ \varphi(x, i_I, z)$ is satisfiable.
    Hence this is a correct many-to-one reduction.
    For the SE property, we define $f_I(u_j) := x_j$ where $u_j$ is the variable representing the $j$-th universe element of $\U(I)$ for $j \in \fromto{1}{n}$ and $x_j$ is the binary variable representing it.
    It now follows from the above equivalence that this is an SE reduction:
    If $S \in \sol(I)$, then $V$ accepts $(I, S)$ after $|I|^{O(1)}$ steps, and for the corresponding assignment $f_I(S)$ it holds that it can be completed to a satisfying assignment of $\varphi$.
    On the other hand, we have $\Uembp = \set{x_1,\dots,x_n}$ and every satisfying assignment $S' \in \sol(\varphi)$ restricted to the positive literal set $\Uembp$ encodes a set $S = f_I^{-1}(S' \cap \Uembp)$ such that $S \in \sol(I)$.
    This proves the SE property and hence $\Pi \leqSE \textsc{Satisfiability}$.
\end{proof}

As the last important property, we can also show transitivity for SE reductions.
Thus, we are in the comfortable situation that we can extend reduction chains by only considering the two problems at hand and not go all the way back to a satisfiability problem.

\begin{lemma}
\label{lem:SE-transitive}
    SE reductions are transitive, i.e. for $\NPS$ problems $\Pi_1, \Pi_2, \Pi_3$ with $\Pi_1 \leqSE \Pi_2$ and $\Pi_2 \leqSE \Pi_3$, it holds that $\Pi_1 \leqSE \Pi_3$.
\end{lemma}
\begin{proof}
Consider for $i = 1,2,3$ the three $\NPS$ problems $\Pi_i = (\I^i, \U^i, \sol^i)$.
There is an SE reduction $(g_1,f_1)$ from $\Pi^1$ to $\Pi^2$ and an SE reduction $(g_2, f_2)$ from $\Pi^2$ to $\Pi^3$.
We describe an SE reduction from $\Pi^1$ to $\Pi^3$.
For this, we require a tuple $(g, (f_I)_{I \in \I})$.

For the first function $g$, we set $g := g_2 \circ g_1$.
This suffices since $\sol^1(I) \neq \emptyset \Leftrightarrow \sol^2(g_1(I)) \neq \emptyset \Leftrightarrow \sol^3((g_2 \circ g_1)(I)) \neq \emptyset$ and $g$ is polynomial-time computable.

Now, let $I^1 := I$ be the initial instance of $\Pi^1$, $I^2 := g_1(I^1)$ be the instance of $\Pi^2$ and $I^3 := g_2(I^2)$ be the instance of $\Pi^3$.
For the second function $f_I$, we define for each instance $I \in \I_1$ the map $f_I := (f_2)_{I^2} \circ (f_1)_{I}$.
Observe that for each instance $I \in \I^1$, the function $f_I$ is injective and maps to $\U^3_{emb}$ and is polynomial-time computable.
It remains to show that $f$ has the desired SE property.
Since $\Pi^1 \leqSE \Pi^2$ and $\Pi^2 \leqSE \Pi^3$ and since (for injective functions) $f(A \cap B) = f(A) \cap f(B)$, we have
\begin{align*}
    \set{f(S^1) : S^1 \in \sol^1} = &\ \set{f_2(f_1(S^1)) : S^1 \in \sol^1}\\
    = &\ \set{f_2(S^2 \cap \U^2_{emb}) : S^2 \in \sol^2}\\
    = &\ \set{f_2(S^2) \cap \U^3_{emb} : S^2 \in \sol^2}\\
    = &\ \set{S^3 \cap \U^3_{emb} : S^3 \in \sol^3}.
\end{align*}
\end{proof}

We now return to the problem versions of {\sc Partition}.
Under solution-embedding reductions, the problem {\sc Partition 1} ($a_n \in S$ or $a_n \notin S$) is not $\NPS$-complete, while {\sc Partition 2} ($a_n \in S$) is $\NPS$-complete.
The reason can be summarized as follows.
There are instances of {\sc Satisfiability} such that the solution set is not \emph{solution-symmetric};
that is, there is a solution $S \in \sol$ such that $\U \setminus S \notin \sol$.
Any SE reduction from {\sc Satisfiability} will result in an instance that is not solution-symmetric.
However, {\sc Partition 1} is solution-symmetric.
If $S \in \sol$, then $\U \setminus S \in \sol$.
Accordingly, there is no SE reduction from {\sc Satisfiability} to {\sc Partition 1}, and {\sc Partition 1} cannot be \NPS-complete. 

For further discussions on using different universes and solution sets for the same problem, we guide the reader to \Cref{sec:alternative-NPS-models}, in which various consequences are discussed in more detail.

\paragraph*{Weak and strong $\NPS$-completeness}  

The properties of strong and weak $\NPS$-completeness are directly inherited from the properties of strong and weak $\NP$-completeness.
An $\NPS$-complete problem $\Pi$ is called \emph{weakly $\NPS$-complete} if its embedding in $\NP$ is weakly $\NP$-complete.
For example, the $\NPS$ version of the Knapsack Problem is weakly $\NPS$-complete.
Similarly, an $\NPS$-complete problem $\Pi$ is called \emph{strongly $\NPS$-complete} if its embedding in $\NP$ is strongly $\NP$-complete.
For example, the $\NPS$ version of {\sc Traveling Salesman} is strongly $\NPS$-complete.

\section{The Polynomial Time Hierarchy}
\label{sec:ptimehierarchy}

There are several ways to define the \emph{polynomial time hierarchy ($\PH$)}.
One of the most common ways to define the polynomial hierarchy is by using alternating quantifiers.

\begin{definition}[Polynomial Time Hierarchy ($\PH$)]
A language $L$ is in $\Sigma_k^p$ if there exists a polynomial $p$ and a polynomial-time computable predicate $R$ such that:
$$
    x \in L \iff \exists y_1 \forall y_2 \exists y_3 \cdots Q_k y_k \; R(x, y_1, \ldots, y_k),
$$
where each $y_i$ ranges over $\{0,1\}^{p(|x|)}$, and $Q_k$ is an existential quantifier if $k$ is odd and a universal quantifier if $k$ is even.
The class $\Pi_k^p$ is defined similarly, with the order of quantifiers starting with a universal quantifier:
$$
    x \in L \iff \forall y_1 \exists y_2 \forall y_3 \cdots Q_k y_k \; R(x, y_1, \ldots, y_k).
$$
The polynomial hierarchy consists of the union of $\Sigma_k^p$ and $\Pi_k^p$ for all $k \in \N$, i.e.
$$
    \PH = \bigcup_{k \in \N} (\Sigma_k^p \cup \Pi_k^p).
$$
\end{definition}

This definition highlights the logic-based nature of $\PH$: $\Sigma_1^p$ corresponds to existential statements ($\NP$), $\Sigma_2^p$ corresponds to existential-universal statements, and so on.

As with the classes $\NP$ and $\sf coNP$, we can define notions of \emph{hardness} and \emph{completeness} for each level of the polynomial hierarchy.
A problem $L$ is said to be \emph{$\Sigma_k^p$-hard} if every problem in $\Sigma_k^p$ can be reduced to $L$ by a polynomial-time many-one reduction.
If, in addition, $L$ belongs to $\Sigma_k^p$, then $L$ is called \emph{$\Sigma_k^p$-complete}.
Similarly, a problem is \emph{$\Pi_k^p$-hard} if every problem in $\Pi_k^p$ reduces to it, and it is \emph{$\Pi_k^p$-complete} if it is also in $\Pi_k^p$.
Completeness results at higher levels of the hierarchy imply strong intractability.
In particular, if a $\Sigma_k^p$-complete problem were solvable in polynomial time, then the entire hierarchy would collapse to level $k$.
Thus, completeness within $\PH$ is strong evidence of computational hardness beyond $\NP$.

An introduction to the polynomial hierarchy and the classes $\Sigma^p_k$ can be found in the book by Papadimitriou \cite{DBLP:books/daglib/0072413}, the book by Sipser \cite{DBLP:books/daglib/0086373} or in the article by Jeroslow \cite{DBLP:journals/mp/Jeroslow85}.
An introduction specifically in the context of bilevel optimization can be found in the article of Woeginger \cite{DBLP:journals/4or/Woeginger21}.

\subsection{A Game-Theoretic Interpretation of the Polynomial Hierarchy}

The alternating quantifier definition of the polynomial hierarchy lends itself naturally to a game-theoretic interpretation.
We can view membership in a language in $\Sigma_k^p$ or $\Pi_k^p$ as the existence of a winning strategy in a two-player game between an \emph{$\exists$-player} (whom we will refer to as Alice) and a \emph{$\forall$-player} (whom we will refer to as Bob).

In the $\Sigma_k^p$ case, Alice goes first, and alternates moves with Bob until $k$ moves are completed.
We adjust the standard alternating quantifier framework in several ways.
This modified approach can be reduced to the standard approach.
For each instance $I$, we assume that there is a universal set $\U(I)$ such that $|\U(I)| = p(|I|)$ for some polynomial $p$, and that the $i$-th move consists of the player selecting a subset $S_i \subseteq \U(I)$.
Second, there are polynomial-time computable predicates $R_1, R_2, ..., R_k$.
The i-th move $S_i$ is \emph{allowable} if $R_i(I, S_1, \ldots S_i) = 1$; otherwise, it is \emph{forbidden}.   
Alice wins the game if every move is allowable, and if 
$R(I, S_1, \ldots, S_k) = 1$.  Otherwise, Bob wins the game.  
We note that the $k$-move game can be expressed in the usual alternating quantifier form.  One can replace $R$ and $R_1, \ldots R_k$ by a single predicate $R'$, where $R'(I, S_1,\ldots, S_k) = 1$ if Bob is the player who makes the first forbidden move, or if all moves are allowable and $R(I, S_1,\ldots, S_k)=1$. 
Thus, the input $x$ belongs to the language $L$ if and only if Alice has a strategy that guarantees a win -- no matter how Bob plays; that is, Alice can ensure that $R'(x, I, S_1,\ldots, S_k) = 1$.

In a $\Pi_k^p$ problem, the quantifiers again alternate, but begin with a universal quantifier. This means Bob moves first, and the roles of the players are reversed: Bob selects $S_i \subseteq \U(I)$ for odd values of i, and Alice selects $S_i \subseteq \U(I)$ for even values of i.
As before, the predicate $R'(I, S_1, \ldots, S_k)$ determines the winner. Alice wins if $R'(I, S_1, \ldots, S_k) = 1$; otherwise, Bob wins. The input $I$ is in the language $L$ if Bob has no strategy that can force Alice to lose. Equivalently, $I \in L$ if Alice has a response strategy that ensures $R(I, S_1, \ldots, S_k) = 1$ for \emph{every} choice Bob might make.
This two-player view provides intuition for the complexity of problems in the polynomial hierarchy: each additional alternation of quantifiers corresponds to an additional round in a game between opposing players, increasing the strategic depth and computational difficulty.

\subsection{The k-Move Adversarial Selection Game}
\label{sec:selection-game}

In this section, we consider very a natural class of games, which we call \emph{adversarial selection games}. Every game from the class $\NPS$ can be transformed into a $k$-move adversarial selection game for any $k$. In this adversarial selection game, Alice and Bob compete to select elements that do or do not complete a solution. Let us give a simple example, where the base game is the Hamiltonian cycle problem.

\begin{description}
    \item[]\textsc{$k$-Move Adversarial Selection Game on Hamiltonian Cycle}\hfill\\
    \textbf{Instance:} Graph $G = (V, E)$, a number $k$, a partition $E_1\cup\dots \cup E_k$ of $E$.\\
    \textbf{Game:} Alice and Bob take turns alternatingly, selecting some set $S_i \subseteq E_i$ for $i=1,\dots,k$. The player who takes the last turn wins iff $\bigcup_{i=1}^k S_i$ is a Hamiltonian cycle.
\end{description}

We can generalize this example to all problems in $\NPS$.
Suppose $\Pi = (\I, \U, \sol)$ is a problem in $\NPS$.
Then the $k$-move adversarial selection game on $\Pi$ is defined as follows.

The instance is defined by $(I, \mathcal{P})$, where $I$ is an instance of $\Pi$ and where $\mathcal{P} = \{U_1, ..., U_k\}$ is a partition of the universe $\U(I)$.
Two players, Alice and Bob, alternately play against each other and Alice begins.
For all $1 \leq j \leq k$:
If $j$ is odd, Alice selects a subset $S_j \subseteq U_j$, and if $j$ is even, Bob selects a subset $S_j \subseteq U_j$.
After all moves $1 \leq j \leq k$, we arrive at a set $S = \bigcup_{j=1}^k S_i$.
The player who moves last wins the game if $S \in \sol(I)$.
Otherwise, the other player wins the game.  The question is whether Alice (who moves first) wins the game.  We will distinguish the cases according to whether $k$ is odd and Alice goes last or $k$ is even and Bob goes last. If Alice goes last, she wins if she can obtain a solution.  If Bob goes last, Alice wins if Bob cannot obtain a solution.

This description yields the following definition.

\begin{definition}[{\sc $k$-Move Adversarial Selection Game}]
    Let a problem $\Pi = (\I,\U,\sol)$ in $\NPS$ be given.
    We define \textsc{$k$-Move Adversarial Selection Game on $\Pi$} as follows:
    The input is an instance $I \in \I$ together with a partition of the universe $\U(I) = U_1 \cup \ldots \cup U_k$.
    \begin{enumerate}
        \item If $k$ is odd, then the question is whether
            $$
                \exists S_1 \subseteq U_1 \; \forall S_2 \subseteq U_2 \; \cdots \; \forall S_k \subseteq U_k : \bigcup_{i=1}^{k} S_i \in \sol(I).
            $$
        \item If $k$ is even, then the question is whether
            $$
                \exists S_1 \subseteq U_1 \; \forall S_2 \subseteq U_2 \; \cdots \; \exists S_k \subseteq U_k : \bigcup_{i=1}^{k} S_i \notin \sol(I).
            $$
    \end{enumerate}
\end{definition}

We observe the following. 

\begin{lemma}
    For all $\Pi \in \NPS$, the problem {\sc $k$-Move Adversarial Selection Game on $\Pi$} is in $\Sigma_k^p$. 
\end{lemma}
\begin{proof}
The game-theoretic construction yields a problem in $\Sigma^p_k$:
The instance $(I, \mathcal{P}_U)$ is a yes-instance if and only if Alice has a strategy that guarantees a win.
If $k$ is odd, then
$$
    (I, \mathcal{P}_U) \in L \iff \exists S_1 \subseteq U_1 \; \forall S_2 \subseteq U_2 \; \cdots \; \forall S_k \subseteq U_k : \bigcup_{i=1}^{k} S_i \in \sol(I).
$$
If $k$ is even, then
$$
    (I, \mathcal{P}_U) \in L \iff \exists S_1 \subseteq U_1 \; \forall S_2 \subseteq U_2 \; \cdots \; \exists S_k \subseteq U_k : \bigcup_{i=1}^{k} S_i \notin \sol(I).
$$

Since $\Pi \in \NPS$, it is possible to check whether $\bigcup_{i=1}^{k} S_i \stackrel{?}{\in} \sol(I)$ in polynomial time and all $U_i$ are subsets of $\U(I)$, thus, they are encodable in polynomial size.
\end{proof}

The case that $\Pi =$ SAT-V is well studied.  The k-move Adversarial Selection Game on SAT-V is sometimes referred to as {\sc $k$-Quantified Boolean Formula} ($QBF_k$) or $k$-Quantified Satisfiability ($QSAT_k$).
Stockmeyer \cite{DBLP:journals/tcs/Stockmeyer76} and Wrathall \cite{Wra76}, who referred to it as $B_k$, proved the following completeness result.

\begin{theorem}\label{thm:adversarial-selection-completeness}
   {\sc $k$-Move Adversarial Selection Game on {\sc Sat-V}} is $\Sigma_k^p$-complete.
   If $k$ is permitted to grow with $n$, then {\sc Adversarial Selection Game on {\sc Sat-V}} is $\PSPACE$-complete.
\end{theorem}

We next extend the above theorem to every $\NPS$-complete problem.

\begin{theorem}
\label{Kmoveselection}
   Let $\Pi$ denote any $\NPS$-complete problem.
   Then {\sc $k$-Move Adversarial Selection Game on $\Pi$} is $\Sigma_k^p$-complete.
   If $k$ is permitted to grow with $n$, then {\sc Adversarial Selection Game on $\Pi$} is $\PSPACE$-complete.
\end{theorem}
\begin{proof}
    We prove the theorem in the case that $k$ is odd.
    Essentially, the same argument establishes the theorem in the case that $k$ is even.

    Suppose that $I$ is an instance of SAT-V with universal set $\U$ and with solution set $\sol$.
    Let $(I, \mathcal{P})$ be an instance of the $k$-move Adversarial Selection Game on {\sc Sat-V}, where $\mathcal{P} = \{U_1, ..., U_k\}$ is a partition of the universal set $\U$.
    Since $\Pi$ is $\NPS$-complete, there is a solution-embedding reduction $(g, f)$ from $\satV$ to $\Pi$.  
    Let $g(I) = I'$.
    Thus $I'$ is an instance of $\Pi$.
    Let $\U'$ and $\sol'$ denote its universe and solution set.
    Recall that $f : \U \to \U'$ is an injective function that embeds the universe of $\satV$ into the universe of $\Pi$.
    Here, the image $f(\U) =: \Uemb'$ is the embedded universe, and $\Uaux' = \U' \setminus f(\U)$ is the auxiliary universe.  Note that $\U' = \Uemb' \cup \Uaux'$.
    We then create an instance $(I', \mathcal{P'})$ of the $k$-move Adversarial Selection Game on $\Pi$, where $\mathcal{P'} = \{U'_1, \ldots,  U'_k\}$ is defined as follows: (i) for $i \in [k-1]$, $U'_i = f(U_i)$; (ii) $U'_k = f(U_k) \cup \Uaux$.
    
    Suppose that Alice and Bob play instance $I$ of the adversarial selection game on {\sc Sat-V}, while Adam and Barbara play instance $I'$ of the adversarial selection game on $\Pi$.
    To prove the theorem, we need to establish that $I$ is a win for Alice if and only if $I'$ is a win for Adam.
    To establish the if and only if condition, we use a new proof technique that we refer to as \emph{dual mimicking}.  
    We first prove that Adam can guarantee a win if Alice can guarantee a win.  We imagine that Alice and Barbara are experts in their respective games, and that each play perfect strategies.  We also imagine that Adam is mimicking Alice's moves, and Bob (who is helping Adam to win) is mimicking Barbara's moves.    
    For each odd value of $i$ such that $i < k$, if Alice selects $S_i$, then Adam selects $S'_i = f(S_i)$.  After Adam moves, Barbara selects a subset $S'_{i+1}$.  And then Bob selects $S_{i+1} = f^{-1}(S'_{i+1})$.
    Recall that $\widehat S_{k-1} = \bigcup_{i\in [k-1]} S_i$. We have assumed that Alice can win in the $k$-th move of the game. Thus $\widehat S_{k-1}$ is a partial solution of $\sol$ with respect to  $\widehat U_{k-1}$. 
    In the other game, $\widehat S'_{k-1} = f(\widehat S_{k-1})$.  Because the reduction from $\Pi$ to $\Pi'$ is solution-embedding, by Lemma \ref{lem:se}, $f(\widehat S_{k-1})$ is a partial solution of $\sol'$ with respect to  $f(\widehat U_{k-1})$.  Therefore, Adam can win in his $k$-th move.
    Since we have assumed that Adam's opponent plays perfect strategy, $I'$ must have been a guaranteed win for Adam.

Suppose instead that Bob can guarantee a win for instance $I'$.  To show that Barbara can guarantee a win, we reverse which players are experts and which do the mimicking.  Now Adam and Bob  are experts who play perfect strategies.  We imagine that Barbara is mimicking Bob's moves, and Alice (who is helping Barbara to win) is mimicking Adam's moves.
For each odd value of $i$ such that $i < k$, if Adam selects $S'_i$, then Alice selects $S_i = f^{-1}(S_i)$.  After Alice moves, Bob selects a subset $S_{i+1} \subseteq U_{i+1}$.  And then Barbara selects  
$S'_{i+1} = f(S_{i+1})$. 
We have assumed that Alice cannot win in move $k$.  Therefore, $\widehat S_{k-1}$ is not a partial solution with respect to $\widehat U_{k-1}$.  By Lemma \ref{lem:se},  $f(\widehat S_{k-1})$ is not a partial solution with respect to $f(\widehat U_{k-1})$.  Therefore, Adam cannot win in his last move.
Since Adam was assumed to play perfect strategy, $I'$ must have been a guaranteed win for Barbara.
\end{proof}

\paragraph*{On dual mimicking proofs.} We will use dual mimicking proofs to prove a statement of the following form.  Adam has a guaranteed win in his game if and only if Alice has a guaranteed win in her game. 

To prove the if direction, we suppose that Alice has a guaranteed win in her game with Bob, and that both Alice and Barbara play optimally.  We suppose that Adam mimics Alice's moves, and Bob mimics Barbara's moves. We then show that Adam wins his game.  Because Barbara played optimally, a win by Adam could only occur if he had a guaranteed win. 

To prove the only if direction, we suppose that Bob has a guaranteed win in his game with Alice, and that both Bob and Adam play optimally.  We suppose that Alice mimics Adam's moves, and Barbara mimics Bob's moves.   We then show that Barbara wins her game.  Because Adam played optimally, a win by Barbara could only occur if she had a guaranteed win.

\section{Interdiction Problems}
\label{sec:interdiction}

In this section, we showcase the power of our $\NPS$ framework by applying it to the class of \emph{minimum cost blocker problems}, also called \emph{most vital vertex/most vital edge problems} or \emph{interdiction problems}.
We further consider multi-stage variants of these problems.
This is a large and natural class of very well-studied problems in the literature, concerned with finding a small set of elements that simultaneously intersects all solutions of a given nominal problem. 
The following is a natural example of an interdiction problem:

\begin{description}
    \item[]\textsc{Interdiction-Clique}\hfill\\
    \textbf{Instance:} Graph $G = (V, E)$, a vertex set $D \subseteq V$, budget $\Gamma \in \N$, number $k \in \N$.\\
    \textbf{Question:} Is there a set $B \subseteq D$ of size $|B| \leq \Gamma$ such that all cliques in the graph $G \setminus B$ have size at most $k-1$?
\end{description}

The problem \textsc{Interdiction-Clique} can be interpreted as a 2-move game between an \emph{interdictor} (Adam) and a \emph{protector} (Barbara). 
The protector wants to find a large clique (of size at least $k$), but the interdictor has the opposite goal. 
He is allowed to delete up to $\Gamma$ vertices $B \subseteq D$ in order to interfere with the goal of the protector, that is, destroy all cliques of size at least $k$. 
We call the set $D$ the set of \emph{vulnerable elements}. 

As outlined in the introduction, interdiction problems in the past have been studied for a vast number of nominal problems in many variants, in disjoint areas such as operational research, theoretical computer science, and pure graph theory. 
What is the common property of all these studied problems?
In almost all the cases, it is to find a small number of elements that interfere with every solution (for example, a small amount of edges intersecting every Hamiltonian cycle). 
Therefore the following \emph{abstract definition of an interdiction problem} is extremely natural.

\begin{definition}[Interdiction Problem]
\label{def:interdiction}
    Let an $\NPS$ problem $\Pi = (\I, \U, \sol)$ be given. The interdiction problem associated to $\Pi$ is denoted by $\textsc{Interdiction-}\Pi$ and defined as follows: The input is an instance $I \in \I$ together with a set $D \subseteq \U(I)$ and a budget $\Gamma \in \Z$.
    The question is whether
    $$
        \exists B \subseteq D \ \text{with} \ |B| \leq \Gamma : \forall S \in \sol(I):  B \cap S \neq \emptyset.
    $$
\end{definition}

The main result of this section is

\begin{theorem}
\label{interdictPi}
    $\textsc{Interdiction-}\Pi$ is $\Sigma^p_2$-complete for \emph{all} $\NPS$-complete problems $\Pi$.
\end{theorem}

Theorem \ref{interdictPi} is a special case of Theorem \ref {thm:protection-interdiction-main-thm}, which we will prove later in this section.
In the literature, interdiction is sometimes studied with associated \emph{interdiction costs}, which prescribe to each element a different cost to interdict it.
This problem is alternatively called the \emph{minimum cost blocker problem}.

\begin{definition}[Minimum Cost Interdiction]
\label{def:interdiction-cost}
    Let an $\NPS$ problem $\Pi = (\I, \U, \sol)$ be given.
    The minimum cost interdiction problem associated to $\Pi$ is defined as follows:
    The input is an instance $I \in \I$ together with a cost function $c : \U(I) \to \Z$ and a threshold $t \in \Z$.
    The question is whether
    $$
        \exists B \subseteq \U(I) \ \text{with} \ c(B) \leq t : \forall S \in \sol(I):  B \cap S \neq \emptyset.
    $$
\end{definition}

\cref{def:interdiction-cost} includes \cref{def:interdiction} as a special case, as can be seen by setting $c(e) \in \set{1, |\U(I)|+1}$ for each element.
Therefore, the following is a corollary of Theorem \ref{interdictPi}.
For every $\NPS$ complete problem $\Pi$, the minimum cost interdiction problem associated to $\Pi$ is $\Sigma^p_2$-complete. 

\subsection{Technical Remarks}
We make a series of technical remarks, clarifying \cref{def:interdiction}.
For the bigger picture, these remarks can be skipped.

As a first remark, note that in our problem \textsc{Interdiction-$\Pi$}, the underlying instance is not changed.
In particular, we do not delete elements of the universe (e.g. vertices of the graph).
For some problems, there might be a subtle difference between deleting elements and forbidding elements to be in the solution.
The vertex cover problem is one such example:
It makes a big difference of deleting a vertex $v$ and its incident edges, or forbidding that $v$ is contained in the solution, but still having the requirement that all edges incident to $v$ get covered by the vertex cover.
In the first case, the vertex cover interdiction problem can be shown to be contained in $\sf coNP$, hence we can not hope to obtain a general $\Sigma^p_2$-completeness result.
In this paper we only consider the second case.

As the second remark, we note that in the literature, usually there is a slight difference between the minimum cost blocker problem and the most vital nodes/edges problem:
The minimum cost blocker problem asks for a minimum cost blocker which decreases the objective value by a set amount. On the other hand, the most vital nodes asks for the maximum value by which the objective can be decreased, when given a certain cost budget for the blocker.  
The interdiction problem, which we formulated here as a decision problem enables us to capture the $\Sigma^p_2$-completeness of both these variants.
It follows by standard arguments that both of the above problems become $\Sigma^p_2$-complete in our setting.

The last remark is that $\textsc{Interdiction-}\Pi$ can be understood as a game between Alice ($\exists$-player, trying to find a blocker) and Bob ($\forall$-player, trying to find a solution).

\subsection{Multi-stage protection-interdiction games}

We can generalize our main result even further, to not only show $\Sigma^p_2$-completeness of interdiction problems, 
but also to general $\Sigma^p_k$-completeness and $\PSPACE$-completeness of $k$-move multi-stage interdiction games.
These games are motivated by extending the initial interdiction scenario: 
For $k=3$, the protector is now allowed to \emph{protect} up to at most $\Gamma_P$ before the interdictor can interdict up to $\Gamma_B$ elements. 
The interdictor is not allowed to choose any of the elements that were protected. Such problems are called \emph{protection-interdiction problems} in the literature. 
For general $k \in \Z_{\geq 2}$, we can extend this idea, and we obtain a game with the following generalized ruleset: 

\begin{itemize}
    \item We imagine that the protector owns a total of $\Gamma_P$ \emph{protection tokens}, and the interdictor owns a total of $\Gamma_B$ \emph{interdiction tokens}.
    \item For a total of $k$ moves, the protector and interdictor alternatingly place their tokens on the set $\U(I)$. If $k$ is odd, the protector starts; if $k$ is even, the interdictor starts.
    \item The input contains some sequence $C_1 \subseteq C_2 \subseteq \dots \subseteq C_{k-1} \subseteq \U(I)$ of nested sets. In the $i$-th move, the players may place tokens only inside $C_i$. In other word, the players are not allowed to access the whole set $\U(I)$ from the start, instead they are allowed to access a region that grows with every turn.
    \item A token can never be placed on any element that is already occupied with an interdiction or protection token.
    \item In each move, the protector may place as many of their $\Gamma_P$ tokens as they wish, but at the end of the game the total amount may not exceed $\Gamma_P$. The analogous holds for the interdictor and $\Gamma_B$.
    \item In the very last move, the protector is not allowed to place tokens anymore. Instead, they select a subset $S \subseteq \U(I)$ that is disjoint from all interdiction tokens. The protector wins if $S \in \sol(I)$.  Otherwise, the interdictor wins.  Note that $S$ can, but does not need to contain the protection tokens. 
\end{itemize}

For $k=2$, we recover the problem \textsc{Interdiction-$\Pi$} from the previous section. In a formally precise way, we can express the game rules  as follows.

For a given problem $\Pi = (\I, \U, \sol)$ in {\sf NP-S}, define two different $k$-move interdiction games: one in which $k$ is even and one in which $k$ is odd.
To simplify the definition, we use $\widehat A_i$ to denote the cumulative union of the first sets of the protector and the interdictor.
More precisely let $A_1, \ldots, A_n$ be a collection of sets, we denote
$\widehat A_i =  \bigcup^i_{j=1} A_j.$
We use the convention that $A_j = \emptyset$ if $A_j$ is not further specified.
In particular, in the following the term $P_k$ denotes the set of protection tokens placed during the $k$-th move, and the term $\widehat P_k$ denotes the total set of all protection tokens placed after the first $k$ moves. Analogously we define $B_k$ and $\widehat B_k$ for the interdiction tokens. For each $k$, either $P_k = \emptyset$ or $B_k = \emptyset$.

As throughout the paper, we formulate every decision problem from Alice’s perspective, asking whether Alice has a winning strategy. For an even number of moves, Alice wins precisely if the protector (Bob) fails to obtain a feasible solution in the final move.

\begin{definition}[Combinatorial $k$-Move Interdiction Games on $\Pi$]
\label{def:comb-prot-int}
    Let a problem $\Pi = (\I, \U, \sol)$ in {\sf NP-S} be given.
    We define \textsc{Comb. $k$-Move Interdiction Game on $\Pi$} as follows:
    The input is an instance $I \in \I$,
    a nested family of subsets ${C_1, \ldots, C_{k-1}}$ of $\U(I)$
    (i.e., $C_i \subseteq C_{i+1}$, $1 \leq i \leq k-2$ and $C_{k-1}\subseteq \U(I)$) of protectable and blockable elements,
     a  budget $\Gamma_P \in \mathbb{Z}$ for the protector, and a budget  $\Gamma_B \in \mathbb{Z}$ for the interdictor (blocker).
    We distiguish between the cases in which $k$ is odd or even.
    \begin{enumerate}
    \item Protection-Interdiction (k is odd). The question is whether
    \begin{align*}
        \exists P_1 \subseteq C_1,
        \forall B_2 \subseteq (C_2 \setminus \widehat P_1),
        \exists P_3 \subseteq (C_3 \setminus \widehat B_2),
        \ldots ,
        \forall B_{k-1} \subseteq (C_{k-1} \setminus \widehat P_{k-2}):
        \qquad\\
        \exists S \in \sol(I) :
        \text{If } |\widehat B_{k-1}| \le \Gamma_B,  \text{ then } |\widehat P_{k-2}| \le \Gamma_P \text{ and }  \widehat B_{k-1} \cap S = \emptyset,\\
        \text{i.e., Bob can obtain a feasible solution.}
    \end{align*}
    \item Interdiction-Protection (k is even).  The question is whether
    \begin{align*}
        \exists B_1 \subseteq C_1,
        \forall P_2 \subseteq (C_2 \setminus \widehat B_1),
        \exists B_3 \subseteq (C_3 \setminus \widehat P_2),
        \ldots,  \forall B_{k-1} \subseteq (C_{k-1} \setminus \widehat P_{k-2}):
        \qquad\\
        \forall S \in \sol(I) :
        \text{If } |\widehat P_{k-2}| \le \Gamma_P, \text{then } |\widehat B_{k-1}| \le \Gamma_B \text{ and }  \widehat B_{k-1} \cap S \neq \emptyset,\\
        \text{i.e., Bob cannot obtain a feasible solution.}
    \end{align*}
    
    \end{enumerate}
\end{definition}

Note that games with an odd number of moves are different from games with an even number of moves since the question is always asked from the perspective of the first player to move.
In both cases, the protector is the last to move.
For an odd number of moves, the first player is the protector and thus wins if he is able to select a solution $S \subseteq \sol$ in the last move.
For an even number of moves, the first player is the interdictor and thus wins if he is able to deny the protector from choosing a solution $S \subseteq \sol$ in the last move.

If we consider $\Pi = \sat$, the formula corresponds to a CNF-formula for odd $k$ and a DNF-formula for even $k$ (by applying the negations).
Indeed, this also corresponds to the natural definition of $\exists$-quantified {\sc QBF$_k$-Sat}:
If $k$ is odd, then {\sc QBF$_k$-CNF-Sat} is the canonical $\Sigma^p_k$-complete problem.
In contrast, if $k$ is even, then {\sc QBF$_k$-DNF-Sat} is the canonical $\Sigma^p_k$-complete problem.

Finally, we consider one last variant of the protection-interdiction game. 
In the game as described above, both players have a budget $\Gamma_B,\Gamma_P$ that can be used at any point during the game, but can in total not be exceeded. Let us call this a \emph{global budget}. 
In contrast, we could also alter the rules of the game such that $\Gamma_B, \Gamma_P$ are local constraints for each move of the game, meaning that in each move the protector (interdictor, respectively) may place up to $\Gamma_P$ protection tokens ($\Gamma_B$ interdiction tokens, respectively). 
All other rules of the game remain the same. Let us call this variant \emph{local budget}.

\subsection{The Complexity of Protection-Interdiction Games}

We begin this subsection by proving the containment.

\begin{lemma}
\label{lem:comb-protection-containment}
    Let $\Pi$ be in $\NPS$.
    If $k \in \mathbb{N}$ is a constant, then the {\sc Comb. $k$-Move Interdiciton Game on $\Pi$} is contained in $\Sigma^p_{2k-1}$.
    If $k$ is part of the input, then it is contained in $\PSPACE$.
\end{lemma}
\begin{proof}
    Let $\Pi = (\I, \U, \sol)$ is in $\NPS$ and $I \in \I$ one of its instances.
    \Cref{def:comb-prot-int} directly describes the problem by alternating quantifiers.
    Furthermore, all sets $S_i$, $\widehat P_i$, $P_i$, $\widehat B_i$, and $B_i$ can be encoded in polynomial space since they are subsets of $\U(I)$.
    At last, the condition $\widehat B \cap S \stackrel{?}{=} \emptyset$ can be checked in polynomial time since $\Pi$ is in $\NPS$.
\end{proof}

The main results of this section is the next theorem.

\begin{restatable}{theorem}{mainThmInterdiction}
\label{thm:protection-interdiction-main-thm}
    Let $\Pi$ be an $\NPS$-complete problem.
    If $k\in \mathbb{N}$ is a constant, then {\sc Comb. $k$-Move Interdiction Game on $\Pi$} is $\Sigma^p_k$-complete.  If $k$ is part of the input, then it is $\PSPACE$-complete. (This holds for both the global and local budget variant.)
\end{restatable}

We will first show that {\sc Comb. $k$-Move Interdiction Game on $\Pi$} is $\Sigma^p_k$-complete if $\Pi = \textsc{Sat-V}$.
For this proof, we will use a reduction from {\sc $k$-Move Adversarial Selection Game on Sat-V}.
We provide the reduction and proof in the case that $k$ is odd.
After the proof, we point out how the reduction and proof would be modified for the case that $k$ is even.
After proving the $\Sigma^p_k$-completeness for $\Pi = \textsc{Sat-V}$, we will show that {\sc Comb. $k$-Move Interdiction Game on $\Pi$} is $\Sigma^p_k$-complete if $\Pi$ is any $\NPS$-complete problem. 

The reduction from {\sc $k$-move Adversarial Selection Game on Sat-V} is quite intricate.  
To explain the reduction, we divide the transformation into two parts.
In the first part, we reduce an arbitrary instance $(I, U_1, U_2, \ldots, U_k)$ of {\sc $k$-Move Adversarial Selection Game on Sat-V} to a \emph{restricted} interdiction game in which there are four additional rules that the protector Adam and the blocker (interdictor) Barbara must satisfy.
We will show that Adam has a guaranteed win for this instance of the restricted interdiction game if and only if Alice has a guaranteed win for the original instance of the adversarial selection game. 
In the second part of the proof, we will show how to relax the four additional rules.
That is, we will reduce the instance of the restricted interdiction game into an instance of the (original) interdiction game.
We will show that Adam has a guaranteed win for this instance of the interdiction game if and only if he has a win for the restricted interdiction game.
Finally, we show the completeness of the protection-interdiction game for \emph{all} $\NPS$-complete problems.

\subsection{The complexity of Protection-Interdiction Sat-V}

In this subsection, we establish the complexity of {\sc Comb. $k$-Move Interdiction $\satV$}.
We remark that the proof is quite involved.
It is possible for the reader to skip this subsection, assuming only its main result. 

\begin{theorem}\label{thm:combinatorial-interdiction-game-sat-completeness}
    The {\sc Comb. $k$-Move Interdiction Game on Sat-V} is $\Sigma^p_k$-complete for all constants $k$.
    If $k$ is part of the input, then it is $\PSPACE$-complete.
\end{theorem}

\subsubsection{The Restricted Interdiction Game on $\satV$}\label{par:restricted-interdiction-game-sat-completeness}

In this section, we will focus mainly on the variant with global budget, since its hardness proof is more involved.
The variant with local budget will follow as an easy consequence.
We begin by describing the restricted interdiction game on $\satV$.
Our hardness reduction starts from the $\Sigma^p_k$-complete {\sc $k$-move Adversarial Selection Game on $\satV$}, as defined in \cref{sec:selection-game}.
We show the transformation here for the case that $k$ is odd.

Let  $I_\text{sel} := (I, \mathcal P)$ be an instance of {\sc $k$-Move Adversarial Selection Game on $\satV$}, where $I$ is the corresponding \textsc{Sat} instance and $\mathcal P = \{U_1, \ldots, U_k\}$ is a partition of the universe $\U(I)$.  
Without loss of generality, we may assume that $|U_i| = n$ for all $i = 1$ to $k$.  (This can be accomplished by adding dummy elements to the smaller $U_i$.)
Let $U_i = \{u_{i1}, \dots, u_{in}\}$.
We recall the rules of the adversarial selection game:
Alice plays against Bob, and Alice goes first.
For $i \equiv 1 \pmod{2}$, Alice selects $S_i \subseteq U_i$. For $i \equiv 0 \pmod{2}$, Bob selects $S_i \subseteq U_i$.
Alice wins if $\bigcup_{i=1}^k S_i \in \mathcal{S}$.
Otherwise, Bob wins.
Given the instance $I_\text{sel}$, we define a new instance $I_\text{restr}$ of the restricted interdiction game on $\satV$, described as follows.  

\begin{itemize}
    \item For each element $u_{ij}$, we create another Boolean variable $w_{ij}$, which is called the \emph{pseudo-complement} of $u_{ij}$.
    For each $i = 1$ to $k$, we let $W_i := \set{w_{ij} : j \in [n]}$.
    \item The nested sets of the $k$-move interdiction game are $C_1, ..., C_k$, where for each $i$, $C_i = \widehat U_i \cup \widehat W_i$.
    \item  The clauses of $I_\text{restr}$ are the same clauses as in $I_\text{sel}$.
    Note that none of the clauses contains a variable from $W = \bigcup _{i=1}^k W_i$.
    \item The budgets are $\Gamma_P =n(k-1)/2$ and $\Gamma_B = n(k-1)$.
\end{itemize}

We now impose additional constraints on the behavior of Adam and Barbara in this restricted interdiction game. We impose four rules that are designed so that the following correspondences will hold when the players use dual mimicking. For $i = 1$ to $k-1$ and for $j = 1$ to $n$, a token is placed on $u_{ij}$ by Alice or Bob in the adversarial selection game if and only if Barbara does not place a token on $u_{ij}$ in the interdiction game.  We will use this correspondence to show that Alice can win the adversarial selection game if and only if Adam can win the interdiction game.

 \restrictedGameRules

Although we have specified the above as rules that should be satisfied, we consider the possibility that the rules are violated.  In the case that a rule is violated, the first player to violate a rule loses.  If no rule is violated, then Adam wins if he can select a set $S'$ of elements satisfying Rule 4 such that  $S'\cap U$ is a solution to the original {\sc Sat} formula, i.e., $S' \cap U \in \sol(I)$.  Otherwise, Barbara wins.
We refer to the $k$-move interdiction game with the above additional rules as the {\sc $k$-Move Restricted Interdiction Game}. 

\begin{lemma}
\label{lem:restricted-protect-interdict}
    Adam has a winning strategy for the instance $I_\text{restr}$ of {\sc $k$-move Restricted Interdiction Game on $\satV$} if and only if Alice has a winning strategy for the instance $I_\text{sel}$ of {\sc $k$-move Adversarial Selection Game on $\satV$}.
\end{lemma}
\begin{proof}

We assume without any loss of generality that Adam and Barbara satisfy all four rules.  

Let $M_P = \{i: i \text{ is odd and } 1 \le i \le k-2\}$ be the iterations on which Adam is protecting elements.  Let $M_B = \{i: i \text{ is even and } 1 \le i \le k-1\}$ be the iterations on which Barbara is blocking elements.
Consider first the case that $I_\text{sel}$ is a win for Alice. We now show that $I_\text{restr}$ is a win for Adam even if Barbara plays optimally. Our proof uses a dual mimicking strategy argument. Suppose that Alice plays optimally, and is thus guaranteed a win. The dual mimicking strategies for Adam and Bob on moves $1$ to $k-1$ are as follows.

\begin{itemize}
  \item \textbf{Adam:} For each $i \in M_P$ and $j \in \{1, ..., n\}$, if Alice selects $u_{ij}$ on move $i$, then Adam protects $u_{ij}$. If Alice does not select  $u_{ij}$, then Adam protects $w_{ij}$.  (The rule guarantees that Adam will satisfy Rule 2 and protect exactly one of $u_{ij}$ and $w_{ij}$.)  
  \item \textbf{Bob:} For each $i \in M_B$ and $j \in \{1, ..., n\}$, if Barbara blocks $u_{ij}$, then Bob does not select $u_{ij}$.  If Barbara blocks $w_{ij}$, then Bob selects $u_{ij}$. (By Rule 1, Barbara blocks exactly one of $u_{ij}$ and $w_{ij}$.)
\end{itemize}

Let $S_i \subseteq U_i$ denote the subset of elements selected by Alice or Bob in move $i$ of the adversarial selection game. For $i \in M_P$, let $S'_i \subseteq (U_i \cup W_i)$ denote the elements protected by Adam.  By Rule 3, for $i \in M_P$, Barbara blocks $T'_i = (U_i \cup W_i)\setminus S'_i$.  For $i \in M_B$, let $T'_i \subseteq (U_i \cup W_i)$ denote the elements blocked by Barbara.

We now claim that for each $i = 1$ to $k-1$, $S_i = U_i\setminus T'_i$.  If the claim is true, then Adam can satisfy Rule 4 and select a subset $S^*$ in move $k$ that contains all unblocked elements of $\bigcup _{i=1}^{k-1} (U_i \cup W_i)$ (thus satisfying Rule 4), and such that $S^* \cap U = \bigcup_{i=1}^k S_i \in \sol(I)$ (thus winning the game).
To see that the claim is true for even values of $i$, note that Bob selects $S_i = U_i\setminus T'_i$. 
To see that the claim is true for odd values of $i$, note that Adam selects $S'_i$  such that $S'_i \cap U = S_i$. By Rule 3, Barbara blocks $T'_i = S'_i\setminus (U_i \cup W_i)$.  Therefore, $S_i = U_i\setminus T'_i$.
Because the claim is true, Adam has a guaranteed win.

We now consider the case that $I$ is a loss for Alice and a win for Bob. We will show that $I'$ is a loss for Adam and a win for Barbara, even if Adam plays optimally. 
We assume that Adam and Bob play optimal strategies for their respective games.  Alice and Barbara use the following mimicking strategies.

\begin{itemize}
  \item \textbf{Alice:} If $i \in M_P$ and Adam selects $u_{ij}$, then Alice selects $u_{ij}$. If Adam selects $w_{ij}$, then Alice does not select $u_{ij}$.
  \item \textbf{Barbara:} If $i\in M_B$ and Bob selects $u_{ij}$, then Barbara blocks $w_{ij}$. If Bob does not select $u_{ij}$, then Barbara blocks $u_{ij}$.  In addition, Barbara blocks the pseudo-complement of every element selected by Adam in his previous move.
\end{itemize}

As before, we let $S_i$ denote the elements of $U_i$ selected by Alice or Bob. For $i \in M_P$, we let $S'_i$ denote elements of $U_i \cup W_i$ that are protected by Adam.  And for $i \in (M_P \cup M_B)$, we let $T'_i$ denote elements of $U_i \cup W_i$ that are blocked by Barbara. 

We now claim that for each $i = 1$ to $k-1$, $S_i = U_i\setminus T'_i$.  If the claim is true, then Adam cannot win the game.  To see why, note that in the last move of the game, by Rule 4, for $i \in \{1, ..., k-1\}$, Adam must select $(U_i \cup W_i)\setminus T'_i$, and thus the elements of $\bigcup_{i=1}^{k-1} U_i$ that Adam selects is $\bigcup_{i=1}^{k-1} S_i$.  Since there is no $S_k \subseteq U_k$ such that $S = \bigcup_{i=1}^k S_i  \in \sol(I)$, Adam cannot win the game.  
To see that the claim is true for $i \in M_B$, note that Barbara's mimicking strategy ensures that $U_i\setminus T'_i = S_i$. 
To see that the claim is true for $i \in M_P$, note that Alice's mimicking strategy ensures that $S_i = S'_i \cap U_i$, and Rule 3 ensures that $S'_i \cap U_i = U_i \setminus T'_i$. 
As before, for $i \in M_P$, we suppose that Alice selects $S_i \subseteq U_i$ and Adam protects $S'_i \subseteq U_i \cup W_i$. For $i \in M_B$, we suppose that Bob selects $S_i \subseteq U_i$, and Barbara blocks $T'_i \subseteq U_i \cup W_i$.  For $i \in M_B$, let $S'_i$ denote the unblocked elements.  Because Bob has a winning game, there is no $S_k \subseteq U_k$ such that $S = \bigcup_{i=1}^k S_i  \in \sol(I)$. 

Because of Alice's mimicking strategy, for $i \in M_P$, $S_i = S'_i \cap U_i$.  Because of Barbara's mimicking strategy, for each $i \in M_B$, $S'_i \cap U_i = S_i$.  
\end{proof}

We now consider the case that $k$ is even.  Alice and Adam go first in their respective games, and Bob and Barbara go last.  This means that Adam is the blocker and Barbara is the protector.    This leads to a number of small changes in the problem instance and in the proof.

\begin{enumerate}
    \item $\Gamma_P = n(k-2)/2$;  $\Gamma_B = n(k-1)$.
    \item Adam wins if Bob cannot select a subset $S$ of unblocked elements with $S \in \sol(I)$.
    \item The blocker moves when the iteration number $i$ is odd, and the protector moves when $i$ is even.  
\end{enumerate}

As before, we would prove that $I_\text{sel}$ is a win for Alice if and only if $I_\text{restr}$ is a win for Adam.  The proofs take into account that when $k$ is even, we need to prove that Adam (the blocker) wins when Alice wins. This proof looks similar to the proof that when $k$ is odd, Barbara (the blocker) wins when Bob wins.  Other than these few changes, the proof is essentially the same.

\subsubsection{Introducing Gadgets for the rules}\label{par:gadgets-interdiction-game-sat-completeness}

Having shown the hardness of the restricted version, we can extend the result to show the hardness of the unrestricted version.

\begin{theorem}
    \textsc{Comb. $k$-Move Interdiction Game on \textsc{Sat-V}} is $\Sigma^p_k$-complete for all constants $k$, and $\PSPACE$-complete for $k$ part of the input.
\end{theorem}
\begin{proof}
    We show that there exists a reduction from {\sc $k$-Move Adversarial Selection Game on Sat-V} to {\sc $k$-Move Interdiction Game on Sat-V}, such that all the additional constraints of the \emph{restricted} interdiction game can be modeled by certain gadgets in the reduction. 
    Thus, together with \cref{lem:restricted-protect-interdict}, this proves the theorem.
    
    We again first assume that $k$ is odd, and discuss the case where $k$ is even at the end.
    Given an instance $I_\text{sel} := (I, \set{U_1, \dots, U_k})$ of {\sc $k$-Move Adversarial Selection Game on $\satV$}, 
    we use it to define an instance $I_\text{restr}$ of the restricted interdiction game just as above. 
    We then further transform the instance $I_\text{restr}$ into an instance $I_\text{int}$ of the unrestricted interdiction game.
    The main idea is that we transform the initial CNF-{\sc Sat} formula $\varphi$ of the instance $I_\text{restr}$ into a new CNF-{\sc Sat} formula $\psi$ for the instance $I_\text{int}$, 
    such that whenever Adam and Barbara play the unrestricted game on $\psi$, then in any optimal strategy Adam and Barbara respect rules 1 -- 4. Therefore, for this choice of $\psi$, the restricted and the unrestricted game are equivalent.
    
    We first give the formal definition of $\psi$, and then explain the intuition behind its different components.
    Given the initial formula $\varphi$ on variables $u_{ij}$ for $i \in [k], j \in [n]$, we introduce new variables $w_{ij}$ for $i \in [k], j \in [n]$, which represents the pseudo-complement of $u_{ij}$  from the restricted game.
    We furthermore introduce additional helper variables $s_{ij}$ for all $i \in [k], i$ even, $j \in [n]$, variables $t_{ij}$ for all $i \in [k-1]$, $i$ odd, $j \in [n]$, and two final helper variables $s,t$. 
    
      Here, we are using the notation $U$, $W$, $u_{ij}$ and $w_{ij}$ in two different ways.  When used in a Boolean expression such as $\varphi(U)$, $U$ (and implicitly $u_{ij}$ and $w_{ij}$) are treated as Boolean vectors and variables. When treated as Boolean variables, the literals $\overline u_{ij}$ and $\overline w_{ij}$ are well-defined.  When we write $S \in \sol$, we interpret it as follows. We let $u_{ij} = true$ if $u_{ij} \in S$, and we let $u_{ij} = false$ otherwise. Then $\varphi(U) = true$ if every clause is satisfied.
    The formula $\psi$ is then defined as:

    \textbf{Formula $\psi$:}
        \begin{align}
        \left(\varphi(U) \lor s \lor t \right) \label{eq:psi-1} \tag{$\ell_1$}\\ 
        \land \left( \bigwedge_{i\in M_P} \ \bigwedge_{j=1}^n(u_{ij} \lor w_{ij} \lor s)\right)
        \land \left(\bigwedge_{i\in M_B}\ \bigwedge_{j=1}^n(u_{ij} \lor w_{ij} \lor s \lor t)\right) \label{eq:psi-2}\tag{$\ell_2$}\\
        \land \left(\bigwedge_{i\in M_B}\ \bigwedge_{j=1}^n (s_{ij} \rightarrow (u_{ij}\land w_{ij}))\right)
        \land \left(\bigwedge_{i\in M_P}\ \bigwedge_{j=1}^n (t_{ij} \rightarrow (u_{ij}\land w_{ij}))\right) \label{eq:psi-3}\tag{$\ell_3$}\\
        \land (\ \overline t \lor \bigvee_{i\in M_P}\ \bigvee_{j=1}^n t_{ij}) \land (\overline s \lor \bigvee_{i \in M_B}\ \bigvee_{j=1}^n s_{ij}) \label{eq:psi-4}\tag{$\ell_4$}\\
        \land \bigwedge_{i=1}^k\bigwedge_{j=1}^n
        ((\overline u_{ij} \leftrightarrow w_{ij}) \lor s \lor t). \label{eq:psi-5}\tag{$\ell_5$}
    \end{align}
\end{proof}

The nested sets of the unrestricted interdiction game, as well as the budgets of protector and interdictor remain exactly the same as in the restricted interdiction game, 
that is, $C_i = \widehat U_i \cup \widehat W_i$ for $i \in [k-1]$, and $\Gamma_P = n(k-1)/2$ and $\Gamma_B = n(k-1)$.
We clarify that from the four literals $u_{ij}, \overline u_{ij}, w_{ij}, \overline w_{ij}$, only the two positive literals $u_{ij},w_{ij}$ are contained in $C_i$ for $i \leq k-1$.  
As a consequence, for all $i \in [k-1]$, in the $i$-th move of the game the protector and interdictor are allowed to place their tokens on the set of positive literals $C_i$. 
Only in the last move, that is, the $k$-th move, the protector is allowed to freely choose an assignment of $\psi$, represented as a subset of all literals (including the literals corresponding to helper variables) that are not occupied by an interdiction token.
We now explain why the reduction is correct. 
Recall the rules 1 -- 4 from the restricted game, which we repeat here for convenience.

\begin{itemize}
  \item \textbf{Rule 1.}  In move $i \in M_B$, the Blocker (Barbara) must block exactly one of $u_{ij}$ or $w_{ij}$ for each $j \in \{1, ..., n\}$. 
  \item \textbf{Rule 2.}  In move $i \in M_P$, the Protector (Adam) must protect exactly one of $u_{ij}$ or $w_{ij}$ for each $j \in \{1, ..., n\}$.
  \item \textbf{Rule 3.}  By the end of move $k-1$, the Blocker must have blocked each unprotected element of the set $\bigcup_{i \in M_P}\ (U_i \cup W_i)$. 
  \item \textbf{Rule 4.}  When selecting elements in move $k$, the Protector must select each unblocked element of $\bigcup_{i=1}^{k-1} (U_i \cup W_i)$.  
\end{itemize}

The main idea is to introduce so-called \emph{cheat-detection variables} $s,t$.  The variables and clauses of the interdiction instance are designed so that the first person to cheat (i.e., violate one of the four rules) will lose the interdiction game. 
The cheat-detection works roughly as follows (we later explain this in more detail). If Barbara (the blocker) is the first player to cheat, and if she cheats by breaking Rule 1 in move $i \in M_B$, then Adam can protect elements in move $i+1$ that will ensure his ability to let $s=1$ in move $k$.  And if $s = 1$, then Adam wins.
If Adam cheats first by breaking Rule 2 in move $i \in M_P$, then Barbara can block variables in move $i+1$ that preclude Adam from winning the game.
If Barbara cheats first by violating Rule 3 at the end of move $k-1$, then Adam can set $t = 1$ in move $k$ and guarantee a win. Finally, if Rules 1, 2, and 3 are satisfied in moves 1 to $k-1$, then any satisfying Boolean expression will satisfy Rule 4. Thus, to have a possibility of winning, Adam must satisfy Rule 4.

To implement the behavior described above, we require two different cheat-detection variables $s,t$, with $s$ being responsible for even $i$, and $t$ being responsible for odd $i$. 
Under this interpretation, the definition of $\psi$ can be interpreted as follows.
Line \eqref{eq:psi-1} states that Adam needs to satisfy $\varphi(U)$, except if he detects a cheat by Barbara. 
Line \eqref{eq:psi-2} states that Adam needs to select at least one literal from each pair of pseudo-complements, unless he has detected a cheat, with slight difference in power of cheat-detection for odd $i$ and even $i$.
Line \eqref{eq:psi-3} states that setting a variable $s_{ij}$ or $t_{ij}$ is only possible, if a pair of pseudo-complement variables $u_{ij}, w_{ij}$ are both true.
Line \eqref{eq:psi-4} states that setting the cheat-detection variable $s$ to true is only possible if at least one variable $s_{ij}$ is true, 
for even $i$. Likewise, setting $t$ to true is only possible if at least one variable $t_{ij}$ is true, for odd $i \neq k$.
Finally, line \eqref{eq:psi-5} states that Adam's assignment of variables $w_{ij}$ needs to respect the pseudo-complements assignment, except if he detected a cheat.

Let us now argue that for optimal play of Adam and Barbara, Rules 1 -- 4 are indeed followed. We make a case distinction upon the \emph{first} rule that is broken.

 \textbf{Case 1:} \textit{Rule 1 is the first one broken.} Let us assume that for some even $i$, all the rules in moves $1,\dots,i-1$ are followed, but in move $i$ for some $j \in [n]$, Barbara does not block exactly one of $u_{ij}$ nor $w_{ij}$. This means she blocks either both, or neither of them. 
Observe that since Adam played according to the rules so far, at this point he has spent $n \cdot i/2$ of his total $\Gamma_P = n(k-1)/2$ protection budget so far.

\textbf{Case 1a:} Assume that Barbara blocked neither $u_{ij}$ nor $w_{ij}$ in move $i$. 
Then  in case $i+1 < k$, Adam protects both $u_{ij}$ and $w_{ij}$ in the immediate next move $i+1$. Note that this is possible since Adam still has two units of budget left. In case $i+1 = k$, Adam simply selects both $u_{ij}$ and $w_{ij}$ in move $i+1 = k$. 
In all cases, in the final move $k$, Adam can set the cheat-detection variable $s$ to true (and $t$ to false). Since the variable $s$ appears in every clause except \eqref{eq:psi-4}, Adam trivially wins the instance $I_\text{prot}$.

\textbf{Case 1b:} Assume that Barbara blocked both $u_{i',j'}$ and $w_{i',j'}$ in move $i'$. 
Then Adam employs the following strategy: Beginning from move $i'+1$, Adam still makes sure to follow Rule 1 by protecting exactly one of $u_{ij}$ and $w_{ij}$ for all $i \in M_P$ and all $j \in [n]$. (Adam's choice whether to protect $u_{ij}$ or $w_{ij}$ is allowed to be arbitrary.)
This is possible since Adam has sufficient budget left.
By the end of move $k-1$, Alice has blocked at most $\Gamma_B = n(k-1)$ elements, including both elements in pair $u_{i',j'}$ and $w_{i',j'}$.  By the pigeonhole principle, there is a pair of unblocked pseudo-complements in $S := \set{u_{i,j}, w_{i,j} : i \in [k-1], j \in [n]}$.

This means that Adam can set either the cheat-detection variable $s$ or $t$ to true. Then, Adam wins the game because line \eqref{eq:psi-2} for  $i \in M_P$ is true since Adam always followed Rule 1. All the other clauses are true because either $s$ or $t$ is true.

\textbf{Case 2:} \textit{Rule 2 is the first one broken.} 
Let us assume that for some $i \in M_P$, all the rules in moves $1,\dots,i-1$ are followed, but in move $i$ for some $j \in [n]$, Adam does not protect exactly one of $u_{ij}$ and $w_{ij}$.
This means Adam either protects both of them or neither of them.
Since Barbara has followed the rules so far, note that she has spent at most $n(i-1)$ of her budget $\Gamma_B$.

Then, Barbara can employ the following strategy starting from move $i+1$:
Let us assume first that Adam protected neither of $u_{ij}$ or $w_{ij}$. Then in the immediate next move $i+1$, Barbara can block both $u_{ij}$ and $w_{ij}$. And in every move after i, Barbara follows Rule 1.  (Ultimately, this will cause Barbara to break Rule 3, but she will win nonetheless.)

Because Barbara follows Rule 1 each time, Adam will not be able to set $s$ to true, and the clause $u_{ij} \lor w_{ij} \lor s$ from \eqref{eq:psi-2} will not be satisfied.  Thus Adam loses the game.
Finally, we consider the case that Adam protected both of $u_{ij}$ and $w_{ij}$.  Given that Adam's budget is $\Gamma_P = |M_P|n$,  by the pigeon hole principle, Adam will leave some unprotected pair $u_{i',j'}$ and $w_{i',j'}$ for $i' \in M_P$ and $j' \in [n]$. Barbara can follow Rule 1 in each move and also block both of the pair $u_{i',j'}$ and $w_{i',j'}$ in move $i'+1$.  As before, Adam will lose.

\textbf{Case 3:} \textit{Rule 3 is the first rule broken.} Note a breaking of Rule 3 can only be identified  at the end of move $k-1$.  By assumption, Rules 1 and 2 were satisfied in each move. Since Adam has maintained Rule 2, at least one of $u_{ij}, w_{ij}$ has been protected for all $i \in M_P$ and all $j \in [n]$. 
Suppose Barbara does not block at least one of $u_{ij}, w_{ij}$ for $i \in M_P$.  Then, in move $k$, Adam can select both $u_{ij}, w_{ij}$, and therefore set the cheat-detection variable $t$ to true. 
Then, Adam trivially satisfies all clauses, except the clause \eqref{eq:psi-2} for $i\in M_P$. 
But these clauses can also be satisfied by Adam, since at least one of every pair $u_{ij}, w_{ij}$ has been protected for all odd $i$, $i \neq k$. In total, in this case Adam can trivially win.

\textbf{Case 4:} \textit{Rule 4 is the first one broken.} If Rules 1 -- 3 are maintained, then Adam must set the cheat-detection variables $s$ and $t$ to false.
Since rules 1 -- 3 have been followed, and since $\Gamma_B = n(k-1)$, for every $i\in [k-1]$ and $j \in [n]$, exactly one of $u_{ij}, w_{ij}$ has been blocked. 
For each $i \in [k-1]$ and $j \in [n]$, to satisfy clauses \eqref{eq:psi-5}, Adam must select the unblocked element in $\{u_{ij}, w_{ij}\}$. Hence Adam has to follow rule 4 to fulfill the \sat-formula.

We finally arrive at the following conclusion: The optimal strategy for Adam and Barbara in the unrestricted game $I_\text{int}$ is as follows: 
\begin{itemize}
    \item For each $i \in M_P$, in move $i$, Adam protects exactly one of  $u_{ij}, w_{ij}$ for each $j \in [n]$.
    \item For each $i \in M_B$, in move $i$, Barbara blocks exactly one of $u_{ij}, w_{ij}$ for each $j \in [n]$.
    \item By the end of move $k-1$, for each $i \in M_P$ and $j \in [n]$, Barbara blocks the unprotected element in $\{u_{ij}, w_{ij}\}$.
    \item In move $k$, for all $i < k$ and $j \in [n]$, Adam selects the unblocked element in $\{u_{ij}, w_{ij}\}$.
\end{itemize}
We conclude that the gagdets we introduced indeed reduce the instance $I_\text{in}$ of the unrestricted game to the instance $I_\text{restr}$ of the restricted game. In total, we have that Adam can win if and only if the original instance of the adversarial selection game on {\sc Sat-V} is a yes-instance. 

This completes the hardness proof for odd $k$ and global budgets $\Gamma_P,\Gamma_B$. We now discuss the case of local budgets. We can copy the hardness proof of the global budgets, 
and now set $\Gamma_P = n$ and $\Gamma_B = 2n$ instead. 
Following the logic of the above proof, this is not a restriction in any way, i.e.\ already in the construction for the global budget in an optimal strategy the protector places exactly $n$ tokens per turn, and the interdictor places exactly $2n$ tokens per turn. Hence by exactly the same proof, we can prove hardness of the variant with local budgets.

Finally, we discuss the case that $k$ is even. In this case, we keep the exact same construction, with precisely the same gadgets, except that now the interdictor is the first player, and plays all the moves with odd indices. 
Likewise, the protector moves second, and moves with all even indices. 
To account for this change, the role of \enquote{odd} and \enquote{even} is swapped in rules 1 -- 4 and in the gadgets. Furthermore, the budgets are modified to $\Gamma_P = n(k-2)/2$ and  $\Gamma_B = n(k-1)$.
The rest of the proof stays essentially the same. It can be seen, that the same logic applies, and also in the case of even $k$, the restricted game can be reduced to the unrestricted game.

In summary, we obtain that \textsc{Comb. $k$-Move Interdiction Game on \textsc{Sat-V}} is $\Sigma^p_k$-hard, if $k$ is a constant, and $\PSPACE$-hard, if $k$ is part of the input.
Containment in these classes follows from \Cref{lem:comb-protection-containment}.

\subsection{The Meta-Reduction}\label{sec:prot-int-meta}
In the previous subsection, we have established that the \textsc{Comb. $k$-Move Interdiction Game} in the special case of $\satV$ is $\Sigma^p_k$-hard.
In this subsection, we now make use of the structural properties of the $\NPS$ framework to extend this result to \emph{all} $\NPS$-complete problems, thereby proving \cref{thm:protection-interdiction-main-thm}. 
For convenience, we repeat the statement.

\mainThmInterdiction
    
\begin{proof}
    The containment in $\Sigma^p_k$ or $\PSPACE$ is established by \Cref{lem:comb-protection-containment}.
    
    We show the hardness by using a dual mimicking proof.
    Since $\Pi$ is $\NPS$-complete, there is a solution-embedding reduction $(g, f)$ from $\satV$ to $\Pi$.
    Let an instance of the \textsc{Comb. $k$-Move Interdiction Game on $\satV$} be given for an arbitrary $k$. 
    We define a new instance of \textsc{Comb. $k$-Move Interdiction Game on $\Pi$}.
    In order to do this, let $I$ be the underlying instance of $\satV$ in {\sc Comb. $k$-Move Interdiction Game on $\satV$}. 
    Then we can consider the instance $g(I) = I'$, which is an instance of $\Pi$.
    
    We turn $I'$ into an instance of \textsc{Comb. $k$-Move Interdiction Game-$\Pi$} in the following way. The nested sets 
    $C'_i$ are given by $C'_i := f(C_i)$ for all $i \in [k-1]$. 
    The budgets $\Gamma_P$ and $\Gamma_B$ remain exactly the same. This defines an instance of \textsc{Comb. $k$-Move Interdiction Game on $\Pi$}. This instance can be computed in polynomial time due to our assumptions on $(f,g)$.
    Recall that $f : \U(I) \to \U(I')$ is an injective function that embeds the universe of $\satV$ into the universe of $\Pi$.
    Here, we called its image $f(\U(I)) =: \Uemb$ the embedded universe, and $\Uaux = \U(I') \setminus \U(I)$ the auxiliary universe. 
    Note that by the construction of the instance of \textsc{Comb. $k$-Move Interdiction Game on $\Pi$}, the first part of the game, where tokens are placed, takes place only in $\Uemb$.
    
    We now claim, that in both these instances the same player has a winning strategy. If we can show this we have proved the theorem.
    For the sake of clarity, let us say that Alice and Bob play the combinatorial interdiction game for $\satV$, while Adam and Barbara play the combinatorial interdiction game for $\Pi$.

    Suppose Alice has a winning strategy in the game on $\satV$. Then Adam also has a winning strategy in the game on $\Pi$.  He partners with Bob, and both Bob and Adam will adopt mimicking strategies in their respective games.  The mimicking works as follows: 
    Whenever Barbara makes the $i$-th move of the game on $\Pi$, i.e.\ she places her tokens on some set $T'_i \subseteq C'_i$, Bob makes the move $T_i = f^{-1}(T'_i)$ in the game on $\satV$. 
    Then Adam observes Alice's best response $T_{i+1}$ in the game on $\satV$, and then places tokens on the set $T'_{i+1} = f(T_i)$.
    This means that Adam effectively steals the strategy of Alice in the game on $\satV$ while Bob is stealing Barbara's strategy in the game on $\Pi$. Adam uses the injective function $f$ in moves 1 to $k-1$. 
    In the final move $k$, this is not possible.
    Instead, by the rules of the game the protector wins the game on $\Pi$ if and only if they can find some $S \in \sol(I')$ which is disjoint from all interdiction tokens. 
    Let $T$ be the set of all interdiction tokens placed in the game on $\satV$, and $T'$ be the set of all interdiction tokens placed in the game on $\Pi$. Then $T' = f(T)$.
    By the solution embedding property, there exists some $S' \in \sol(I')$ with $S' \cap T' = \emptyset$ if and only if there exists $S \in \sol(I)$ with $S \cap f^{-1}(T') = S \cap T = \emptyset$. Because Alice can win her game,  Adam can also win his game. 
    
    Let us now assume that Bob has a winning strategy in the game on $\satV$. Then Barbara also has a winning strategy in the game on $\Pi$. This is because by an analogous argumentation, she can steal Bob's strategy by applying the injective function $f$.  And Alice can steal Adam's strategy by applying the function $f^{-1}$.  In move $k$ of the game on $\satV$, Alice will have no way of winning because Bob has a guaranteed win.  Because of the solution embedding property, in move $k$ of the game on $\Pi$, Adam will have no way of winning.   

    \begin{figure}[!ht]
        \centering
        \begin{tikzpicture}[scale=0.82]
            \node[below right] at (-5,-2.8) {$\U_{\textsc{\satV{}}}$};
            \draw[] ($(-5,0)$) rectangle ($(-2,-0.67)$);
            \node[above left] at (-2,-0.67) {$C_1$};
            \draw[] ($(-5,-0.67)$) rectangle ($(-2,-1.33)$);
            \node[above left] at (-2,-1.33) {$\ldots$};
            \draw[] ($(-5,-1.33)$) rectangle ($(-2,-2)$);
            \node[above left] at (-2,-2) {$C_{k-1}$};
            \draw[] ($(-5,-2)$) rectangle ($(-2,-2.8)$);
            \node[above left] at (-2,-2.8) {$\U_{\textsc{\satV{}}} \setminus \hat{C}_{k-1}$};
                        
            \draw[] ($(0,0.2)$) rectangle ($(5.2,-4)$);
            \node[below right] at (2,-2.8) {$f(\U_{\textsc{\satV{}}})$};
            \node[below right] at (0,-4) {$\U_\Pi$};
            \draw[] ($(1,0)$) rectangle ($(5,-0.67)$);
            \node[above left] at (5,-0.67) {$f(C_1)$};
            \draw[] ($(1,-0.67)$) rectangle ($(5,-1.33)$);
            \node[above left] at (5,-1.33) {$\ldots$};
            \draw[] ($(1,-1.33)$) rectangle ($(5,-2)$);
            \node[above left] at (5,-2) {$f(C_{k-1})$};
            \draw[] ($(1,-2)$) rectangle ($(5,-2.67)$);
            \node[above left] at (5,-2.8) {$f(\U_{\textsc{\satV{}}} \setminus \hat{C}_{k-1})$};
                        
            \draw[->, thick] (-2,-0.7) -- (1,-0.7);
            \node[above] () at (-1,-0.7) {$f$};
        \end{tikzpicture}
        \caption{
            The relation between the universes when applying the meta-reduction from the \textsc{Comb. $k$-Move Interdiction Game} for $\satV$ to the \textsc{Comb. $k$-Move Interdiction Game} for the problem $\Pi$.
            The function $f$ maintains a one-to-one correspondence between the elements of $C_i$ and $f(C_i)$.
        }
        \label{fig:schematic-meta-reduction}
    \end{figure}
\end{proof}

\section{Adjustable Robust Problems}
\label{sec:two-stage}

In this section, we consider adjustable robust optimization problems with discrete budgeted uncertainty.
We show that for every single $\NPS$-complete problem, the corresponding two-stage adjustable problem is $\Sigma^p_3$-complete.
Moreover, we generalize these concepts of two-stage adjustable robustness to a multi-stage setting and show that $k$-stage adjustable problems are $\Sigma^p_{2k-1}$-complete and, if $k$ is part of the input, $\PSPACE$-complete.

In two-stage adjustable optimization, a decision maker has to find a solution under uncertainty.
Instead of having one stage of decision making, as in usual robust optimization, which may be too conservative, the decisions are divided into two stages:
first-stage decisions, also called here-and-now decisions, and second-stage decisions, also called wait-and-see decisions.
The first-stage decisions must be made before the uncertainty has been realized.
Then the uncertainty for the second stage is realized, after which the second-stage decisions are made.

We will consider two variants of adjustable robust problems.
In the first version, the costs on the universe elements are part of the uncertainty set.
In the second version, the available decisions are part of the uncertainty set.
Here, the unavailable decisions can be considered as blocked, or perhaps that the costs are so high as to make the decisions too costly.
We first present the version in which the uncertainty sets are defined with discrete budgeted uncertainty over the cost function.

The decision maker is given a cost function $c_1: \U(I) \to \Z$ for the first stage decisions.
The input for the second stage decisions includes lower bound costs $\underline c_2 \in \Z$ and upper bound costs $\overline c_2 \in \Z$ as well as an uncertainty parameter $\Gamma \in \Z$.
The uncertainty set for the second stage consists of all scenarios with $\underline c_2(u) \le c_2(u) \le \overline c_2(u)$ for all $u \in \U(I)$ such that at most $\Gamma$ elements from $u \in \U(I)$ have costs greater than $\underline c_2$.
The second stage costs are set by an adversary.  Whenever, the adversary sets a cost $c_2(e) > \underline c_2(e)$, it is optimal for the adversary to let $c_2(e) = \overline c_2(e)$.
Thus, we can formally define the uncertainty set $C_\Gamma$ of discrete budgeted uncertainty over the cost functions by all cost functions such that at most $\Gamma$ elements $u \in \U(I)$ have costs of $\overline c_2(u)$, while all other have $\underline c_2(u)$.
$$
    C_\Gamma := \{c_2 \mid \forall u \in \U(I) : c_2(u) = \underline c_2(u) + \delta_u(\overline c_2(u) - \underline c_2(u)), \ \delta_u \in \bin, \sum_{u \in \U(I)} \delta_u \leq \Gamma \}.
$$

We observe that if $c_2(u) \ge c_1(u)$ for all $u \in \U(I)$, it would be optimal for the decision maker to choose a complete decision in the first stage, thus making the second stage irrelevant.
However, for many problems of interest, $\underline c_2(u) < c_1(u)$ for some $u \in \U(I)$, and the second stage decisions are needed.

We model two-stage adjustable optimization as a game with three moves, one of which occurs in the first stage, and two of which occur in the second stage.
In the first move of the game, the decision maker selects a set $S_1 \subseteq \U(I)$ that incurs costs of $c_1(S_1)$.
In the second move, the adversary selects a cost function $c_2 \in C_\Gamma$.
In the third move, the decision maker selects an additional subset $S_2 \subseteq \U(I)$.
The decision maker wins if the solution $S_1 \cup S_2$ is valid, i.e. $S_1 \cup S_2 \in \sol(I)$, and has total cost is at most some specified threshold $t_{TS}$; i.e, $c_1(S_1) + c_2(S_2) \le t_{TS} $.
Otherwise, the adversary wins.

The following is an illustration of a two-stage adjustable robust problem derived from the TSP.

\begin{description}
    \item[]\textsc{Two-Stage Adjustable TSP}\hfill\\
    \textbf{Instance:}
    Graph $G = (V, E)$,
    first-stage cost function $c_1 : E \to \Z$,
    second-stage cost functions $\underline c_2 : E \to \Z$, $\overline c_2 : E \to \Z$,
    threshold $t_{TS} \in \Z$,
    and budget $\Gamma \in \Z$.\\
    \textbf{Question:}
    Is there a subset of edges $S_1 \subseteq E$ such that for all cost functions $c_2 \in C_\Gamma$, there is a an additional subset of edges $S_2 \subseteq E \setminus S_1$ such that $S_1 \cup S_2$ is a Hamiltonian cycle of weight $c_1(S_1) + c_2(S_2) \leq t_{TS}$?
\end{description}

When \textsc{Two-Stage Adjustable-TSP} is interpreted as a 3-move, 2 person game, it is easy to see that 
{\sc Two-Stage Adjustable TSP} is in $\Sigma^p_3$.
We next define two-stage adjustable problems for every problem in $\NPS$.

\begin{definition}[Two-Stage Adjustable Problem]
\label{def:two-stage}
    Let a problem $\Pi = (\I, \U, \sol)$ in $\NPS$ be given.
    The two-stage adjustable problem associated to $\Pi$ is denoted by \textsc{Two-Stage Adjustable-$\Pi$} and defined as follows:
    The input is an instance $I \in \I$ together with a first stage cost function $c_1: \U(I) \rightarrow \Z$, two second stage cost functions $\underline c_2: \U(I) \rightarrow \Z$ and $\overline c_2: \U(I) \rightarrow \Z$, a threshold $t_{TS} \in \Z$, and an uncertainty parameter $\Gamma \in \Z$.
    The question is whether
    $$
        \min_{S_1 \subseteq \U(I)} \max_{c_2 \in C_\Gamma} \min_{\substack{S_2 \subseteq \U(I) \setminus S_1\\ S_1 \cup S_2 \in \sol(I)}} c_1(S_1) + c_2(S_2) \leq t_{TS}.
    $$
\end{definition}

The main result of this section is the following theorem.

\begin{theorem}
\label{2stageadjustable}
    For each $\NPS$-complete problem $\Pi$, the two-stage adjustable variant \textsc{Two-Stage Adjustable-$\Pi$} is $\Sigma^p_3$-complete.
\end{theorem}

We defer the proof of Theorem \ref{2stageadjustable} to the next subsection, where we prove a more general result.
The above theorem remains true even if every cost is either 0 or 1.
Such problems are equivalent to \emph{Combinatorial Two-Stage Adjustable Problems}, which we define next.

\begin{definition}[Combinatorial Two-Stage Adjustable Problem]
\label{def:comb-two-stage}
    Let an $\NPS$ problem $\Pi = (\I, \U, \sol)$ be given.
    The combinatorial two-stage adjustable problem associated to $\Pi$ is denoted by \textsc{Comb. Two-Stage Adjustable-$\Pi$}.
   Its input is an instance $I \in \I$ together with three sets
    $U_1 \subseteq \U(I)$,
    $U_2 = \U(I) \setminus U_1$, and
    $D \subseteq U_2$,
    and an uncertainty parameter $\Gamma \in \N_0$.
    The question is whether
    \begin{align*}
        & \exists S_1 \subseteq U_1 :
        \forall B \subseteq D \ \text{with} \ |B| \leq \Gamma :
        \exists S_2 \subseteq (U_2 \setminus B) :
        S_1 \cup S_2 \in \sol(I).
    \end{align*}
\end{definition}

We refer to the elements in $B$ as \emph{vulnerable} elements, as they are elements that can be blocked by the adversary.
In a combinatorial two-stage adjustable problem, the adversary can block the elements in $B$ from being selected by the decision maker.
It can be reformulated as a problem in which every cost is either $0$ or $1$ and where the decision maker must select a solution with cost $0$.
First, we let $c_1(e) = 0$ for each element $e \in U_1$ and $c_1(e') = 1$ for each element $e' \in U_2$.
Second, for each element $e \in U_2$, we let  $\underline c_2(e) = 0$.
Furthermore, for each element $e \in D$, we let $\overline c_2(e) = 1$ and for each element $e' \in U_2 \setminus B$, we let $\overline c_2(e') = 0$. 
Finally, we set the threshold to $t_{TS} = 0$.
Since the decision maker wins by selecting a solution with cost 0, the adversary is able to block $\Gamma$ elements from $B$ by setting costs to $1$.
These are the same conditions under which the decision maker wins the combinatorial version.

\subsection{Multi-stage Adjustable Robust Problems}
In this subsection, we extend the definition on two-stage adjustable robustness to a definition of multi-stage adjustable robustness.
Based on this generalization, we prove that $k$-stage problem adjustable robust problems are $\Sigma^p_{2k-1}$-complete.
If $k$ is part of the input, we prove that the problem is $\PSPACE$-complete.

To motivate our discussion of multi-stage adjustable robustness, we first describe an online variant of the well-known knapsack problem, which is called online knapsack with estimates \cite{DBLP:journals/corr/abs-2504-21750}.
Given a knapsack with capacity $C \in \N$ and objects $O = \{1, \ldots, n\}$ with profits $p: O \to \N$ and weights between $\underline w: O \to \N$ and $\overline w: O \to \N$, as well as an uncertainty parameter $\Gamma \in \N$, a decision maker wants to solve the following task.
The weight of each object $i$ is known to be between $\underline w_i$ and $\overline w_i$. At most $\Gamma$ objects will a have weight of more than $\underline w$.
Now, the objects arrive one at a time.
The decision maker is then able to weigh the arriving object and identify its actual weight.
Given the actual weight of the object, the decision maker has to irrevocably decide whether to pack the object in the knapsack or discard the object.
The decision maker wants to maximize the value in the knapsack.
This problem falls into a family of online optimization problems with estimates that have been studied:
knapsack \cite{DBLP:journals/corr/abs-2504-21750},
graph exploration \cite{DBLP:journals/corr/abs-2501-18496},
bin packing \cite{DBLP:journals/corr/abs-2505-09321}.
We can model the above problem as a multi-stage adjustable robust problem.
Thus, the theorems of this section provide complexity results for online problems with estimates as well.

For the rest of this section, we describe these problems as games between the two players, a decision maker and an adversary, who we refer to as the interdictor.
We further focus our complexity analysis to combinatorial multi-stage adjustable robust problems, which include a set of vulnerable elements $D$.
As already indicated above, the cost version of multi-stage adjustable robustness includes the combinatorial version as a special case.   We state the general complexity of the combinatorial version in \Cref{cor:adjustable-robust-main-thm}.
For each $k \in \Z_{\geq 2}$, we obtain a game with the following rule set:
\begin{itemize}
    \item The input includes a partition of the universe $U$ into $k$ disjoint subsets $U_1, U_2, \ldots, U_k$, a  subset $D \subseteq U$ of vulnerable elements $D$, and a budget $\Gamma$ for the interdictor.
    \item In the first round, the agent selects a subset $S_1 \subseteq U_1$.
    \item At the beginning of each round $i \in \fromto{2}{k}$, the interdictor blocks a subset $B_i \subseteq (U_i \cap D)$.
    Afterwards, the agent selects a subset $S_i \subseteq U_i \setminus B_i$.
    \item In each round, the interdictor is able to block as many elements as he wants, but at the end of the game the number of blocked elements may not exceed $\Gamma$, i.e., $\sum_{2 \leq i \leq k} |B_i| \leq \Gamma$.
    \item The agent wins if $S_1, \ldots, S_k$ such that $\bigcup^{k}_{i = 1} S_i \in \sol(I)$.
    Otherwise, the interdictor wins.
\end{itemize}

This game is similar to the protection-interdiction game  game presented in Section \ref{sec:interdiction}.  In contrast to the protection-interdiction game, the decision maker in this adjustable robust game directly selects a subset of non-blocked elements instead of protecting them by placing tokens.

\begin{definition}[Combinatorial $k$-stage Adjustable Robust Problem]
\label{def:combinatorial-multi-stage-adjustable-robust-problem}
    Let an $\NPS$ problem $\Pi = (\I, \U, \sol)$ be given.
    The combinatorial $k$-stage adjustable robust game associated to $\Pi$,
    denoted by {\sc Comb. $k$-stage Adjustable Robust $\Pi$ Game}, is defined as follows:
    The input is an instance $I \in \I$ together with
    a partition  of the universe $\U(I)$ into subsets $U_1,  \ldots, U_k$,
    a set of vulnerable elements $D$,
    and an uncertainty parameter $\Gamma \in \N_0$.
    The question is whether
    \begin{align*}
        \hspace{-0.25cm}\exists S_1 \subseteq U_1, \forall B_2 \subseteq (D \cap U_2), \exists S_2 \subseteq (U_2 \setminus B_2), \ldots, \forall B_k \subseteq (D \cap U_k), \exists S_k \subseteq (U_k \setminus B_k) :\\
        \text{if } \sum^k_{i = 1} |B_i| \leq \Gamma, \text{ then } \bigcup^k_{i = 1} S_i \in \sol(I).
    \end{align*}                    
\end{definition}

\subsection{The Complexity of Multi-Stage Adjustable Robust Problems}

We first state the containment in $\Sigma_{2k-1}^p$ as a lemma.

\begin{lemma}
\label{lem:comb-adjustable-robust-containment}
    Let $\Pi$ be a problem in $\NPS$.
    If $k \in \mathbb{N}$ is a constant, then the {\sc Comb. $k$-Stage Adjustable Robust $\Pi$} is contained in $\Sigma^p_{2k-1}$.
    If $k$ is part of the input, then it is contained in $\PSPACE$.
\end{lemma}
\begin{proof}
    Let $\Pi = (\I, \U, \sol)$ is in $\NPS$ and $I \in \I$ one of its instances.
    \Cref{def:combinatorial-multi-stage-adjustable-robust-problem} directly describes the problem by alternating quantifiers.
    Furthermore, all sets $S_i$ and $B_i$ can be encoded in polynomial space since they are subsets of $\U(I)$.
    The condition $\bigcup^k_{i = 1} S_i \in \sol(I)$ can be checked in polynomial time since $\Pi$ is in $\NPS$.
\end{proof}

The main results of this section are the next two theorems.

\begin{theorem}
\label{thm:comb-adjustable-robust-main-thm}
    Let $\Pi$ be an $\NPS$-complete problem.
    If $k \in \mathbb{N}$ is a constant, then the {\sc Comb. $k$-Stage Adjustable Robust $\Pi$} is $\Sigma^p_{2k-1}$-complete.
    If $k$ is part of the input, then it is $\PSPACE$-complete.
\end{theorem}

Theorem \ref{thm:comb-adjustable-robust-main-thm} can be used to prove the corresponding theorem for multi-stage adjustable optimization with cost functions instead of a set of vulnerable elements.

\begin{corollary}
\label{cor:adjustable-robust-main-thm}
    Let $\Pi$ be an $\NPS$-complete problem.
    If $k \in \mathbb{N}$ is a constant, then the {\sc $k$-Stage Adjustable Robust $\Pi$} is $\Sigma^p_{2k-1}$-complete.
    If $k$ is part of the input, then it is $\PSPACE$-complete.
\end{corollary}
\begin{proof}
    Similar to $k=2$, we generalize combinatorial adjustable robustness based on problems in $\NPS$ to multi-stage adjustable robust problems, which include cost functions $\underline c_i : \U \to \Z$ and $\overline c_i : \U \to \Z$ for each stage $1 \leq i \leq k$.
    For this, let $t_{TS} = 0$. 
    On the one hand, we set $\underline c_i(u) = 0$ for all $u \in U_i$.
    Otherwise, $\underline c_i(u) = 1$.
    On the other hand, we set $\overline c_i(u) = 0$ for all $u \in U_i \setminus D$.
    Otherwise, $\overline c_i(u) = 1$.
    Thus, only non-blocked elements of the current stage can be selected by the decision maker.
\end{proof}

Our proof for {\sc Comb. $k$-Stage Adjustable Robust $\Pi$} follows along the same lines as our proof for the Combinatorial $k$-Move Protection-Interdiction Game.
We first reduce an instance of {\sc Adversarial Selection $\satV$ Game} to an instance of restricted version of {\sc Comb. $k$-Stage Adjustable Robust $\satV$}, a version that includes additional rules.
We then modify this reduction by adding gadgets to prove that the original {\sc Comb. $k$-Stage Adjustable Robust $\satV$} is $\Sigma^p_{2k-1}$-complete.
Finally, we show that there is a meta-reduction from {\sc Comb. $k$-Stage Adjustable Robust $\satV$} to {\sc Comb. $k$-Stage Adjustable Robust $\Pi$} for any $\Pi$ that is $\NPS$-complete.

\subsection{The complexity of Adjustable Robust \satV{}}

In this subsection, we establish the complexity of {\sc Comb. $k$-Stage Adjustable Robust $\satV$}.

\begin{theorem}
\label{thm:comb-adjustable-robust-sat-thm}
    The \textsc{Comb. $k$-Stage Adjustable Robust Game $\Pi$} for $\Pi = \textsc{SAT-V}$ is $\Sigma^p_{2k-1}$-complete for all constants $k$.
    If $k$ is part of the input, then it is $\PSPACE$-complete.
\end{theorem}

\subsubsection{The Restricted Adjustable Robust \textsc{Sat-V}}\label{par:restricted-adj-rob-game-sat-completeness}
We begin by describing the restricted adjustable robust problem on \textsc{Sat-V}. 
Our hardness reduction starts from the $\Sigma^p_{2k-1}$-complete {\sc $(2k-1)$-Move Adversarial Selection Game on $\satV$}, as defined for $k \in \N$ in \cref{sec:selection-game}.
Let $\satV = (\I_{\sat},\U_{\sat},\sol_{\sat})$.
Further, let  $I_\text{sel} := (I, U_1, \dots, U_{2k-1})$ be an instance of {\sc $(2k-1)$-Move Adversarial Selection Game $\satV$}, $k \in \N$, where $I \in \I_{\sat}$, and $U_1 \cup \dots \cup U_{2k-1}$ is a $(2k-1)$-partition of its set of variables $\U_{\sat}(I)$.  
Without loss of generality, we may assume that $|U_i| = n$ for all $i \in \fromto{1}{2k-1}$.
At last, let $U_i = \{u_{i1}, \dots, u_{in}\}$.

We recall the rules of a $(2k-1)$-move adversarial selection game:
Alice plays against Bob, and Alice goes first.
For $i \equiv 1 \pmod{2}$, Alice selects $S_i \subseteq U_i$.
For $i \equiv 0 \pmod{2}$, Bob selects $S_i \subseteq U_i$.
Alice wins if $\bigcup_{i=1}^{2k-1} S_i \in \mathcal{S}(I)$.
Otherwise, Bob wins.

Given the instance $I_\text{sel}$, we define a new instance $I_\text{restr}$ of restricted $k$-stage adjustable robust $\satV$.
The construction is a modification of our construction used in the completeness proof for {\sc $(2k-1)$-Move Interdiction Game on $\satV$}.
For each $U_i := \set{u_{ij} : j \in [n]}$, let $W_i := \set{w_{ij} : j \in [n]}$.
The element $w_{ij}$ is called the \emph{pseudo-complement} of $u_{ij}$.
Let $W = \bigcup_{i=1}^{2k-1} W_i$.
We construct $I_\text{restr}$ as follows:

\begin{itemize}
    \item The instance $I_\text{restr}$ contains the elements
        $$
            \bigcup_{i=1}^{2k-1} U_i \cup \bigcup_{i=1}^{2k-1} W_i,
        $$
        where $U$ is also the universe of $I_\text{sel}$.
     \item For $i \in [k]$, we define $C_i$ as follows.  $C_1 = U_1 \cup W_1$.  For each $i \in \fromto{2}{k}$, 
        $$
            C_i = (U_{2i-2} \cup W_{2i-2}) \cup (U_{2i-1} \cup W_{2i-1}).
        $$
    \item $D = \bigcup_{i=2}^{k} (U_{2i-2} \cup W_{2i-2})$.
    \item The budget is $\Gamma = n(k-1)$.
\end{itemize}

We refer to the players of the Restricted Adjustable Robust Game as Adam (the selector) and Barbara (the blocker).  As in the restricted Protection-Interdiction game, we now impose additional rules on Adam and Barbara which will result in a natural correspondence between the moves of Alice and Adam as well as a correspondence between the moves of Bob and Barbara. In the latter case, Barbara will block the elements of $U_i$ that are not selected by Bob.

\begin{itemize}
  \item \textbf{Rule 1.} In Stage 1 (and move 1), Adam must select exactly one of $u_{1,j}$ and $w_{1,j}$ for each $j \in \{1, \ldots, n\}$.
  \item \textbf{Rule 2.} In Stage $i \geq 2$ (and move $2i-2$), Barbara must block exactly one of $u_{2i-2,j}$ and $w_{2i-2,j}$ for each $j \in \{1, \ldots, n\}$. 
  \item \textbf{Rule 3.} In Stage $i \geq 2$ (and move $2i - 1$), Adam must select each unblocked element of $U_{2i-2} \cup W_{2i-2}$ and exactly one of $u_{2i-1,j}$ and $w_{2i-1,j}$ for each $j \in \{1, ..., n\}$. 
\end{itemize}

Although we have specified the above as rules that should be satisfied, we consider the possibility that the rules are violated.  The first player to violate a rule loses the game. If no rule is violated, let $S'$ denote the union of the items selected by Adam.  Then Adam wins if $S' \cap U \in \sol(I_\text{SAT})$.  Otherwise, Barbara wins.

We refer to the $k$-stage adjustable robust game with the above additional rules as {\sc restricted $k$-Stage Adjustable Robust Game}.

\begin{lemma}
\label{lem:restricted-adjustable-robust}
    Adam has a winning strategy for the instance $I_\text{restr}$ of {\sc $k$-Stage Restricted Adjustable Robust Game on $\satV$} if and only if Alice has a winning strategy for the instance $I_\text{sel}$ of {\sc $(2k-1)$-Move Adversarial Selection Game on $\satV$}.
\end{lemma}
\begin{proof}

Let $S'_1$ denote the elements of $U_1 \cup W_1$ selected by Adam in Stage 1.
For $2 \le i \le k$, let $T'_{2i-2}$ denote the elements that Barbara selects from $U_{2i - 2} \cup W_{2i-2}$.
Further, let $S'_{2i-2} = (U_{2i - 2} \cup W_{2i-2})\setminus T'_{2i-2}$ denote the elements that Adam selects from $U_{2i - 2} \cup W_{2i-2}$. Let $S'_{2i-1}$ that Adam selects from $U_{2i - 1} \cup W_{2i-1}$.

We now prove that $I_\text{restr}$ is a win for Adam if $I_\text{sel}$ is a win for Alice.
We first assume that $I_\text{sel}$ is a win for Alice, and we will prove that it is also a win for Adam.
Our proof uses a dual mimicking strategy argument.
Suppose that Alice and Barbara play optimally in their respective games.
The mimicking strategies for Adam and Bob on stages $1$ to $k$ are as follows.

\begin{itemize}
  \item \textbf{Adam:}  For each $i = 1$ to $k$,  if Alice selects $u_{2i-1,j}$ at stage $i$, then Adam selects $u_{2i-1,j}$. Otherwise, Adam selects $w_{2i-1,j}$.  In addition, for $i \ge 2$, Adam selects 
  $S'_{2i-2} = (U_{2i-2} \cup W_{i-2})\setminus T'_{2i-2}$.
  \item \textbf{Bob:} For each $i = 1$ to $k$, $S_{2i-2}=U_{2i-2}\setminus T'_{2i-2}$.
\end{itemize}

Because Adam is mimicking Alice, it follows that for $i = 1$ to $k$, $S'_{2i-1} \cap U_{2i-1} = S_{2i-1}$.
Because Bob is mimicking Barbara, $S_{2i-2}=U_{2i-2}\setminus T'_{2i-2}$, and thus $S_{2i-2} = U_{2i-2} \cap S'_{2i-2}$.  
It follows that $S' \cap U = S$, and thus Adam wins the game.
We now assume that $I_\text{sel}$ is a win for Bob, and we will prove that it is also a win for Barbara.  In this case, the proof relies on mimicking strategies by Alice and Barbara.

\begin{itemize}
  \item \textbf{Alice:}  For each $i = 1$ to $k$, Alice selects $S_{2i-1} = S'_{2i-1} \cap U_{2i-1}$. 
  \item \textbf{Barbara:} For each $i = 1$ to $k$, and for each $j = 1$ to $n$, if Bob selects $u_{2i-2, j}$, then Barbara blocks $w_{2i-2, j}$.
  If Bob does not select $u_{2i-2, j}$, then Barbara blocks $u_{2i-2, j}$.
\end{itemize}

Because Alice is mimicking Adam, it follows that $S'_{2i-1} \cap U_{2i-1} = S_{2i-1}$ for $i = 1$ to $k$.
Because Barbara is mimicking Bob, it follows that $U'_{2i-2} \setminus T'_{2i-2} = S_{2i-2}$.
Since Adam is choosing the pseudo-complements of Barbara's blocked elements, $S'_{2i-2} =  (U_{2i-2} \cup W_{2i-2})\setminus T'_{2i-2}$, and thus $S'_{2i-2} \cap U_{2i-2} = S_{2i-2}$.
We conclude that $S' \cap U = S$.
Since the game is a winning game for Bob, $S \notin \sol$, and thus $S' \cap U \notin \sol$, and it is a winning game for Barbara.
\end{proof}

\subsubsection{Introducing Gadgets for the rules}\label{par:gadgets-adj-rob-game-sat-completeness}
Having established the hardness of the restricted adjustable robust $\satV$, we next introduce gadgets that enforce the rules in a general adjustable robust $\satV$ instance.
In doing so, we derive the following theorem.

\begin{theorem}
    \textsc{Comb. $k$-Stage Adjustable Robust \textsc{Sat-V}} is $\Sigma^p_{2k-1}$-complete for all constants $k$, and $\PSPACE$-complete if $k$ is part of the input.
\end{theorem}
\begin{proof}
    Containment is established by \Cref{lem:comb-adjustable-robust-containment}.
    For the hardness, we use gadgets to reduce {\sc $(2k-1)$-Move Adversarial Selection Game on Sat-V} to the \emph{unrestricted} version of {\sc $k$-Stage Adjustable Robust Sat-V}, such that the three additional constraints of the \emph{restricted} $k$-stage adjustable robust problem are guaranteed to be satisfied if the players use optimal strategies. 
    Thus, together with \cref{lem:restricted-adjustable-robust}, this proves the theorem.
    
    Given an instance $I_\text{sel} := (I_\text{SAT}, U_1, \dots, U_{2k-1})$ of {\sc $(2k-1)$-Move Adversarial Selection Game on $\satV$}, 
    we have defined an instance $I_\text{restr}$ of {\sc Restricted $k$-Stage Adjustable Robust Sat-V}. 
    We then further transform the instance $I_\text{restr}$ into an instance $I_\text{adj}$ of the unrestricted {\sc $k$-Stage Adjustable Robust Sat-V}, such that whenever Adam and Barbara play the unrestricted game, then in any optimal strategy Adam and Barbara respect rules 1 -- 3.
    Therefore, the restricted and the unrestricted problem are equivalent.
    We first give the formal definition of $\psi$, and then explain the intuition behind its different components.
    Given the initial formula $\varphi$ on variables $u_{ij}$ for $i \in [k], j \in [n]$, we introduce pseudo-complements $w_{ij}$ for $i \in [k], j \in [n]$.
    In addition, we introduce additional helper variables $s_{ij}$ for all $i \in [k], i$ even, $j \in [n]$, variables $t_{ij}$ for all $i \in [k-1]$, $i$ odd, $j \in [n]$, and two final helper variables $s,t$.
    The formula $\psi$ is then defined as:

    \textbf{Formula $\psi$:}
    \begin{align}
        & \left(\varphi(U) \lor s \right) \label{eq:adj:psi-1} \tag{$\ell_1$}\\
        & \qquad \land \bigwedge_{i \text{ even}} \ \bigwedge_{j=1}^n (s_{ij} \rightarrow (u_{ij}\land w_{ij})) \label{eq:adj:psi-2} \tag{$\ell_2$}\\
        & \qquad \land (\overline s \ \lor \bigvee_{i \text{ even}} \ \bigvee_{j=1}^n s_{ij}) \label{eq:adj:psi-3} \tag{$\ell_3$}\\
        & \qquad \land \bigwedge_{i=1}^{2k-1}\bigwedge_{j=1}^n 
         ((u_{ij} \vee w_{ij}\lor s)\land
        (\overline u_{ij} \vee \overline w_{ij} \vee s)) . \label{eq:adj:psi-4} \tag{$\ell_4$}
    \end{align}

    The partition $C_1, \ldots, C_{2k-1}$, the set of vulnerable elements $D$ as well as the budget $\Gamma$ stay the same as in the restricted game.
    That is, $C_1 = U_1 \cup W_1$ and $C_i = (U_{2i-2} \cup W_{2i-2}) \cup (U_{2i-1} \cup W_{2i-1})$ for $1 \leq i \leq k$, $D = \bigcup_{i=2}^{k} (U_{2i-2} \cup W_{2i-2})$, and $\Gamma = n(k-1)$.
    Note that the interdictor can only assign the value false by blocking variables $u_{ij}, w_{ij}$ contained in the set $D$.
    However, the decision maker can assign the value true or false in each of his moves.

    As in the proof for the interdiction game, we introduce a cheat-detection variable $s$ for situation in which the interdictor cheats.  For this proof, there is no need to introduce a second cheat-detection variable.  Recall  the rules of the restricted game.
    
    \begin{itemize}
      \item \textbf{Rule 1.} In Stage 1 (and move 1), Adam must select exactly one of $u_{1,j}$ and $w_{1,j}$ for each $j \in \{1, \ldots, n\}$.
      \item \textbf{Rule 2.} In Stage $i \geq 2$ (and move $2i-2$), Barbara must block exactly one of $u_{2i-2,j}$ and $w_{2i-2,j}$ for each $j \in \{1, \ldots, n\}$. 
      \item \textbf{Rule 3.} In Stage $i \geq 2$ (and move $2i - 1$), Adam must select each unblocked element of $U_{2i-2} \cup W_{2i-2}$ and exactly one of $u_{2i-1,j}$ and $w_{2i-1,j}$ for each $j \in \{1, ..., n\}$.
    \end{itemize}
    
    We claim that if one player cheats on a move, then the other player has a guaranteed win.
    We again consider the importance of all rules following a hierarchy over the moves.
    More precisely, Rule 1 is the most important as it is valid for the first stage.
    Then, for each stage $i \in \fromto{2}{k}$, we consider Rule 2 (Barbara's move) more important than Rule 3 (Adam's move) in the same stage $i$ and the rules of stage $i-1$ more important than the rules of stage $i$.
    Therefore, the first player to break a rule will lose the game.

    Line \ref{eq:adj:psi-1} is the original formula $\varphi(U)$ of the original variables $U$ together with the cheat-detection variable $s$.
    If $s$ is set to true by Adam, he trivially satisfies this sub formula.
    Line \ref{eq:adj:psi-2} implements the detection of a cheat by the interdictor Barbara.
    If she blocks both $u_{ij}$ and $w_{ij}$, she has not enough budget to block one of $u_{i'j'}$ and $w_{i'j'}$.
    Then Adam can assign $u_{i'j'}$ and $w_{i'j'}$ the value true, and assign $s_{ij}$ the value true.
    This is propagated to Line \ref{eq:adj:psi-3} such that Adam can set $\overline s$ to false.
    Line \ref{eq:adj:psi-4} acts as a cheat-detection gadget for Adam, which guarantees that he selects exactly one of $u_{ij}$ and $w_{ij}$ for all $i \in \fromto{1}{2k-1}$ and $j \in \fromto{1}{n}$.
    In conclusion, we derive the following cases.

    \textbf{Case 1:} \textit{Rule 1 is the first one broken.}
    Accordingly, Adam sets both $u_{1j}$ and $w_{1j}$ either to true or to false for some $j \in \fromto{1}{n}$.
    Barbara obviously did not cheat beforehand.
    From that point on, Barbara just needs to adhere to rule 2 and she trivially wins the game, since Adam cannot set the variable $s$ to true.
    
    \textbf{Case 2:} \textit{Rule 2 is the first one broken.}
    We consider the first case in which Barbara blocked neither $u_{ij}$ nor $w_{ij}$.
    Then, Adam can set both $u_{ij}$ and $w_{ij}$ to true and is able to satisfy \ref{eq:adj:psi-2} by setting $s_{ij}$ to true and $s$ to true.
    Thus, \ref{eq:adj:psi-1} and \ref{eq:adj:psi-4} are trivially fulfilled and Adam wins the game.

    In the second case, Barbara blocked both $u_{ij}$ and $w_{ij}$.
    Since Barbara's budget is only $\Gamma = n(k-1)$, there are some $i' \geq i$ and $j \in \fromto{1}{n}$ such that neither $u_{i'j'}$ nor $w_{i'j'}$ are blocked.
    Then, Adam can win the game by playing along the rules until stage $i'$ and use the same strategy as in case 1.
    
    \textbf{Case 3:} \textit{Rule 3 is the first one broken.}
    Barbara played along the rules until stage $i$ so that Adam selects neither $u_{ij}$ nor $w_{ij}$ for some $j \in \fromto{1}{n}$.
    Barbara now has a winning strategy by just playing along the rules.
    Then, Adam cannot satisfy \ref{eq:adj:psi-4} since $s$ needs to be set to false and both $u_{ij}$ and $w_{ij}$ are set to false and loses the game.

    By this case distinction, we conclude that it is optimal for both Adam and Barbara to follow all three rules.
    If all three rules are followed, then the clauses in \ref{eq:adj:psi-2},\ref{eq:adj:psi-3}, and \ref{eq:adj:psi-4} are satisfied, and $s$ is forced to be 0.
    Thus, $\psi(U)$ reduces to $\varphi(U)$.
    In conclusion, Adam can win the restricted adjustable robust $\satV$ on instance $I_{restr}$ if and only if he can win the unrestricted adjustable robust $\satV$ on instance $I_{adj}$, completing the proof.
\end{proof}

\subsection{The Meta-Reduction}\label{sec:adj-rob-meta}

As the last step, we want to establish the completeness of all adjustable robust problems that are based on an $\NPS$-complete problem.
For this, we show that there is a meta-reduction starting at {\sc Comb. Adjustable Robust $\satV$} using the underlying solution-embedding reduction from $\satV$ to $\Pi$.

\begin{theorem}
    If $\Pi$ is $\NPS$-complete, then {\sc Comb. $k$-Stage Adjustable Robust $\Pi$} is $\Sigma^p_{2k-1}$-complete, if $k$ is a constant, and $\PSPACE$, if $k$ is part of the input.
\end{theorem}
\begin{proof}

    Containment is analogous to \Cref{lem:comb-adjustable-robust-containment}.
    We further show the hardness by using the dual mimicking strategy proof.
    Since $\Pi$ is $\NPS$-complete, there is a solution-embedding reduction $(g, f)$ from $\satV$ to $\Pi$.
    Let $I$ be the instance of $\satV$ and $g(I) = I'$ be the instance of $\Pi$.
    Alice and Bob play the adjustable robust $\satV$ game.
    Adam and Barbara play the adjustable robust $\Pi$ game.
    For each move of Adam $i \in \fromto{1}{2k-1}$, $i$ even:
    If Alice selects $S_i \subseteq U_i$, Adam selects $S'_i = f(S_i)$.
    For each move of Barbara $i \in \fromto{2}{2k-2}$, $i$ even:
    If Bob selects $S_i \subseteq U_i$, Barbara selects $B'_i = f(U_i \setminus S_i)$.
    Adam then selects $S'_i = f(U_i \setminus B_i)$.
    The optimality of these strategies can be argued analogously to the proof of \Cref{thm:protection-interdiction-main-thm}.
    
\end{proof}

\section{Alternative NP-S models }
\label{sec:alternative-NPS-models}

In this section, we investigate different models for $\NPS$ problems.
The universe and the corresponding solution sets of $\NPS$ problems have a great influence not only on the corresponding problems in the polynomial hierarchy and $\PSPACE$ but also on possible solution-embedding reductions between those problems.
On the one hand, the universe and solutions implicitly define different versions of problems in the polynomial hierarchy, for example for interdiction:
Here, we are asked to find a blocker $B \subseteq D$, where $D \subseteq \U$, that is, the set of vulnerable elements $D$ directly depends on the choice of the universe $\U$ in the $\NPS$ problem.
On the other hand, a great many $\NP$-completeness reductions are (perhaps implicitly) designed to be solution-embedding, as evidenced by the problems described in Section \ref{sec:ssp-reductions}.
Even in an $\NP$-completeness reduction involving complex \enquote{gadgets}, there often is an underlying solution-embedding mapping $f$.
However, there are $\NP$-complete problems that are not $\NPS$-complete and there are $\NP$-completeness reductions that are not solution-embedding.
Thus in this section, we further deal with such problems and reductions and argue that most of these can be interpreted as $\NPS$-complete problems or solution-embedding reductions under different models.

In the first part of this section, we start our investigation by analyzing different version of \textsc{Satisfiability} by discussing the use of complementary variables in the universe for a problem $\Pi \in \NPS$. 
Complementary variables, which are the counterpart of the existing variables, can be useful as part of a solution-embedding reduction;
we further can define literal-based models with them.
This can be extended to other combinatorial problems, which we will capture in the notion of \emph{dual problems}.
All of these models may also lead to desirable properties in the interdiction problems created from $\NPS$-complete problems.
In the second part, we investigate $\NPS$ problems that are \emph{universe covering}.
Such problems may be partitioning and packing problems, which are not directly $\NPS$-complete under the typical subset interpretation.
Indeed, these problems can be modeled differently so that they become $\NPS$-complete.
In the third part, we discuss a technique that we call \emph{2-step verification} and apply it to the problem {\sc Sequencing to Minimize Tardy Tasks} to show that there is an $\NPS$-complete model for this problem.

\subsection{Complementary variables and duals of problems}
\label{complementing}

In \Cref{sec:npsandlop}, we defined \textsc{Satisfiability (Variable-based)}, denoted by {\sc Sat-V}, in which the universe is the set of variables, and \textsc{Satisfiability (Literal-based)}, denoted by {\sc Sat-L}, in which the universe is the set of literals.
Additionally, it is possible to define the problem \textsc{Satisfiability (Complement Variable-based)}, in which the universe is again the set of variables, however, the solutions are defined as the variables that are assigned false.

Since all known reductions start from \textsc{Satisfiability}, the solution structure for $\NPS$-complete problems is similar to one of these problems, such that we generalize these notions.
Based on this, we discuss how the use of complementary models and literal-based models may enrich the modeling of the corresponding problems in the polynomial hierarchy and $\PSPACE$.

Let $\Pi = (\I, \U, \sol)$ be an $\NPS$ problem.
We may interpret $\U(I)$ as a set of binary variables for some instance $I$ as already discussed in \Cref{sec:npsandlop}.
We denote $\overline \U(I)$ the set of complementary variables to $\U(I)$.
That is, for $x_j \in \U(I)$, we have $\overline x_j = 1 - x_j$.
With this observation, we derive general notions for models of $\NPS$ problems.

\paragraph*{Complement-based models}
We define the \emph{complement-based} model of $\Pi$ by $\Pi_C = (\I, \U_C, \sol_C)$.
For this, we define $\U_C = \overline \U$ and modify the solution set $\sol$ appropriately, i.e., for each instance $I \in \I$: if $S \in \sol(I)$, then $\{\overline u \mid u \in S\} \in \sol_C(I)$, and if $S_C \in \sol_C(I)$, then $\{u \mid \overline u \in S_C\} \in \sol(I)$. 

For example, consider the problem {\sc Clique (Edge-based)}, in which we are given a graph $G=(V,E)$ and a number $k$ and we are asked to find a set $S \subseteq E$ such that if $\{u,v\}, \{v,w\} \in E$, then $\{u,w\} \in E$ and $|S| \geq k$.
Furthermore, consider the problem {\sc Independent Set (Edge-based)}, in which we are given a graph $G'=(V',E')$ and a number $k'$ and we are asked to find a set $S \subseteq E$ such that if $\{u,v\}, \{v,w\} \in E$, then $\{u,w\} \notin E$ and $|S| \geq k$.
The problem {\sc Clique (Edge-based)} is the complement-based model of {\sc Independent Set (Edge-based)}, by interpreting the complement of the edges $E$ of {\sc Clique (Edge-based)} as the edges $E' = \overline E$ in {\sc Independent Set (Edge-based)}.

\begin{lemma}
    The complement-based model of {\sc Sat-V} is $\NPS$-complete.   
\end{lemma}
\begin{proof}
    Recall that in Theorem \ref{thm:cook-levin}, we established that {\sc Sat-V} is $\NPS$-complete.
    In that proof, we could have replaced each variable $x_i$ by a different variable $\overline y_i$.
    Thus each occurrence of $x_i$ in a clause would be replaced by $\overline y_i$, and each occurrence of $\overline x_i$ would be replaced by $y_i$.
    The resulting proof shows that the complement-based model of {\sc Sat-V} is $\NPS$-complete.    
\end{proof}

\begin{theorem}
\label{complement-based}
    Suppose that $\Pi$ is an $\NPS$-complete problem in which the universe consists of variables.
    Then the complement-based model of $\Pi$ is also $\NPS$-complete. 
\end{theorem}
\begin{proof}
    The solution-embedding reduction from {\sc Sat-V} to $\Pi$ is also a solution-embedding reduction from the complement-based model of {\sc Sat-V} to the complement-based model of $\Pi$.
    To see this, let the reduction from {\sc Sat-V} to $\Pi$ map $x_i$ to $f(x_i)$, then we can map a variable $\overline x_i$ in $\textsc{Sat-V}_C$ to $f(\overline x_i)$ in $\Pi_C$.
    Since $S \in \sol(I)$, then $\{\overline u \mid u \in S\} \in \sol_C(I)$, and vice versa, we have a correct solution-embedding reduction.
\end{proof}

\paragraph*{Literal-based models}
The \emph{literal-based} model of $\Pi$ is obtained by redefining $\Pi = (\I, \U, \sol)$ to $\Pi_{L} = (\I, \U_L, \sol_L)$, where $\U_L(I) = \U(I) \cup \overline \U(I)$ for each instance $I \in \I$.
We further replace $S \in \sol(I)$ by $S' \in \sol_L(I)$ as follows:
(i) $x_j \in S$, iff $x_j \in S'$, and (ii) $x_j \notin S$, iff $\overline x_j \in S'$.

\begin{theorem}\label{literal-based}
    Suppose that $\Pi$ is a variable-based model of some $\NPS$-complete problem.
    Then the literal-based variant $\Pi_L$ is $\NPS$-complete.
\end{theorem}
\begin{proof}
    In the reduction of $\Pi = (\I, \U, \sol)$ to $\Pi_L = (\I, \U_L, \sol_L)$, we define $g(I) = I$ for all $I \in \I$.
    Since $\U_L(I)$ includes the set $\U(I)$, we define $f(x_i) = x_i$ for each variable $x_i \in \U(I)$.
    This reduction is solution-embedding.
\end{proof}

\begin{lemma}\label{satvandlequiv}
    The following are true: {\sc Sat-V} $\leqSE$ {\sc Sat-L} and {\sc Sat-L} $\leqSE$ {\sc Sat-V}.
\end{lemma}
\begin{proof}
    {\sc Sat-V} $\leqSE$ {\sc Sat-L} follows from \Cref{literal-based}.
    To show {\sc Sat-L} $\leqSE$ {\sc Sat-V}, we define the following reduction.
    For each pair of literals $\ell, \overline \ell \in L$, we create a variable pair $x^t, x^f \in X$.
    We further define $f(\ell) = x^t$ and $f(\overline \ell) = x^f$.
    Moreover, we substitute each occurrence of $\ell$ in some clause $c \in C$ by $x^t$ and each occurrence of $\overline \ell$ in some clause $c \in C$ by $x^f$.
    At last, we add the clauses $\{x^t \lor x^f \}$ and $\{\overline x^t \lor \overline x^f\}$ to the clause set.
\end{proof}

\paragraph*{Dual problems}
For each literal-based problem in \NPS, we define another problem that we call the ``dual''.
For a literal-based problem $\Pi \in \NPS$, the \emph{dual} of $\Pi = (\I, \U, \sol)$, denoted by $\Pi_D = (\I, \U, \sol_D)$, is defined by $\sol_D(I) = \{S \subseteq \U(I) : (\U(I) \setminus S) \in \sol(I) \}$.

As an example, consider the problem {\sc Vertex Cover}, in which we are given a graph $G=(V,E)$ and a number $k$ and we are asked to find a set $S \subseteq V$ such that $|S \cap e| \geq 1$ for all $e \in E$ and $|S| \leq k$.
In the literal-based variant of {\sc Vertex Cover}, the universe is $\U(I) = V \cup \overline V$.
For each vertex $v \in V$, a solution contains $v \text{ or } \bar v$, according to whether $v$ is selected or not.
Its dual problem is the literal-based version of {\sc Independent Set}, in which we are given a graph $G=(V,E)$ and a number $\ell$, and we are asked to find a subset $S \subseteq V$ such that $|S| \ge \ell$ and for each edge $(v, w) \in E$, $|S \cap \{v, w\}| \leq 1$.

To see the direct correspondence, suppose that $\sol_{VC}$ is the solution set of {\sc Vertex Cover} for the instance with graph $G = (V, E)$ and threshold $k$.
Let $\sol_{IS}$ be the solution set of {\sc Independent Set} for the instance with graph $G = (V, E)$ and $\ell = n - k$.
Then $\sol_{IS} = \{S \subseteq (V \cup \overline V) : ((V \cup \overline V) \setminus S) \in \sol_{VC}\}$.

\begin{theorem}
    If a literal-based problem $\Pi$ is $\NPS$-complete, then the dual of $\Pi$ is also $\NPS$-complete.
\end{theorem}
\begin{proof}
    Suppose that $\Pi = (\I, \U, \sol)$ is a literal-based model of an $\NPS$-complete problem and that $\Pi_D = (\I, \U, \sol_D)$ is its dual.
    Since $\Pi$ is a literal-based model, the universe is a partition $\U(I) = U \cup \overline U$ for each instance $I \in \I$.
    Furthermore, by definition of a literal-based model, for each $S \in \sol(I)$, $u \in S$ if and only if $\overline u \notin S$.
    Thus, for all instances $I \in \I$, we can map $f(u) = \overline u$ and $f(\overline u) = u$ for all $u \in \U(I)$.
    It follows that $\U(I) = f(\U(I))$ and $S \in \sol(I)$ if and only if $(\U(I) \setminus S) \in \sol_D(I)$.
    This mapping is a bijection, and it is also solution-embedding.
\end{proof}

As discussed earlier, different models of the same problem in $\NP$ lead to different problems in $\NPS$, and they are mapped to different problems in $\PH$.    
Consider, for example, the problem \textsc{Interdiction-Knapsack} that can arise.
When the universe $\U$ consists of variables, then the blocker, at her turn, can block items from being in the knapsack.
If, instead, the set $\U$ contains both variables and their complements, then the blocker can also force items to be in the knapsack.
Literal models and variable models (be it complementary or not) thus lead to different interdiction games.

\subsection{Universe-Covering Problems}
A class of problems in $\NPS$ that cannot be $\NPS$-complete are those that have the property of \enquote{universe covering}.
We say that a problem $\Pi$ in $\NPS$ is \emph{universe covering} if for every instance $I \in \I$ of $\Pi$ such that $\sol(I) \neq \emptyset$, we have $\bigcup _{S\in \sol(I)} S = \U(I)$.

\begin{lemma}
\label{lem:ucovering}
     If a problem $\Pi$ is universe covering, then the problem cannot be $\NPS$-complete.
\end{lemma}
\begin{proof}
    We carry out a proof by contradiction.
    Suppose that $\Pi = (\I_\Pi, \U_\Pi, \sol_\Pi)$ is both universe-covering and $\NPS$-complete.
    Then there must be an SE reduction from $\textsc{Sat-V} = (\I_{\textsc{Sat}}, \U_{\textsc{Sat}}, \sol_{\textsc{Sat}})$ to $\Pi$.
    Now, let $I \in \I_{\textsc{Sat}}$ be an instance of \textsc{Sat-V} that is chosen with the property that $x_1 \in \U_{\textsc{Sat}}(I)$, and $x_1$ is not in any solution $S \in \sol_{\textsc{Sat}}(I)$.
    Furthermore, let $g(I)$ be the instance obtained from $I$ by the SE reduction.
    Because $\Pi$ is universe covering, there is a solution $S' \in \sol_\Pi(g(I))$ such that $f(x_1) \in S'$.
    Because the reduction is SE, there is a solution $S = f^{-1}(S'\ \cap f(\U_\Pi(I))) \in \sol_{\textsc{Sat}}$ that contains $x_1$, contrary to our assumption.     
\end{proof}

We follow up with different natural and important problems that are universe covering.
The first example of a universe covering problem in $\NPS$ is the problem \textsc{3-Coloring} expressed as follows.

\begin{samepage}
    \begin{mdframed}
    	\begin{description}
            \item[]\textsc{3-Coloring (IP-based model)} \hfill\\
            \textbf{Instances $\I$:} \quad Undirected graph $G = (V, E)$.\\
            \textbf{Universe $\U$:} \quad $V \times \{1,2,3\}$.\\
            \textbf{Solutions $\sol$:} \quad The set of all $S \subseteq V \times \{1,2,3\}$, where $(v, i) \in S$ means that vertex $v$ is given color $i$, such that if $(v, i) \in S$ and $(w, i) \in S$, then $\{v, w\} \notin E$. 
        \end{description}
    \end{mdframed}
\end{samepage}
If $S$ is any feasible solution, then another solution $S'$ can be obtained by swapping color 1 with color 2 or with color 3.  The symmetries in the solution sets cause the problem to be universe-covering, and prevent it from being $\NPS$-complete.  

Another example of a universe-covering problem in $\NPS$ is the problem \textsc{Bin Packing} as described next. 
\begin{samepage}
    \begin{mdframed}
    	\begin{description}
            \item[]\textsc{Bin Packing (IP-based model)} \hfill\\
            \textbf{Instances $\I$:} \quad  Set $A \subseteq \N$, bin capacity $B \in \N$, number $k \in N$. \\
            \textbf{Universe $\U$:} \quad $\{(a, j) : a \in A \text{ and } j \in \{1, \ldots, k\} \}$ representing assignments of an element of $A$ to one of the bins $\{1, \ldots, k\}$.\\
            \textbf{Solutions $\sol$:}  \quad The set of all $S \subseteq \{(a, j): a \in A \text{ and } j \in \{1, \ldots, k\} \}$, such that for each $a \in A$ there is exactly one $j$ with $(a, j) \in S$ and $\sum_{(a, j) \in S} a \le B$ for all $j \in \{1, \ldots, k\}$. 
        \end{description}
    \end{mdframed}
\end{samepage}
In this $\NPS$ model of {\sc Bin Packing}, for every feasible solution $S$, there are $k!$ different solutions, each obtained by permuting the labels on the bins.
By Lemma \ref{lem:ucovering}, this model of {\sc Bin Packing} is not $\NPS$-complete.
Both, the model of {\sc 3-Coloring} and the model of {\sc Bin Packing}, are universe-covering because each solution set $\sol$ is invariant upon permutation.

In this subsection, we develop models of {\sc 3-Coloring}, and {\sc Bin Packing} that are all $\NPS$-complete.
In \Cref{coloring}, we discuss a modeling approach based on equivalence relations that can be used to create $\NPS$-complete partitioning and coloring problems.  
In \Cref{subsetgen}, we discuss a modeling approach in which $\U$ is a collection of subsets of a ground set.
We refer to these models as \emph{subset-generating}.
We show that the subset-generating models of the problems {\sc 4-Partition} and {\sc Packing with Triangles} are both $\NPS$-complete.

\subsubsection{Coloring and Partitioning Problems}
\label{coloring}

If an IP-based model for a partitioning problem is universe-covering, then, by Lemma \ref{lem:ucovering}, this model of the problem cannot be $\NPS$-complete. 
Here, we introduce a way of modeling the parts (or colors) as  equivalence classes in an equivalence relation.
This type of modeling breaks the usual symmetry of IP-based models.
We show that the following problems are all $\NPS$-complete: {\sc 3-Coloring}, {\sc 4-Partition}, and {\sc Bin Packing}.

\paragraph*{Modeling equivalence relations and equivalence classes.}   

An \emph{equivalence relation} on a set $S$ is a relation $\sim$ that satisfies the following three properties for all $a,b,c \in S$:
\emph{reflexivity} ($a \sim a$), \emph{symmetry} (if $a \sim b$, then $b \sim a$), and \emph{transitivity} (if $a \sim b$ and $b \sim c$, then $a \sim c$).
Given an equivalence relation $\sim$ on a set $S$, the \emph{equivalence class} of an element $a \in S$ is defined as
$[a] = \{ x \in S : x \sim a \}$.  Every element belongs to exactly one equivalence class, and no two distinct equivalence classes overlap.
Thus, the equivalence classes of $S$ form a partition of $S$.

We show that the equivalence relation model of the 3-coloring problem is $\NPS$-complete.

\begin{samepage}
    \begin{mdframed}
    	\begin{description}
            \item[]\textsc{3-Coloring (Equivalence relation model)} \hfill\\
            \textbf{Instances $\I$:} \quad Undirected graph $G = (V, E)$.\\
            \textbf{Universe $\U$:} \quad $\{(v, w) : v,w \in V \text{ and } \{v, w\} \notin E\}$.\\
            \textbf{Solutions $\sol$:} \quad An equivalence relation $S \subseteq \U$ with at most three equivalence classes.
        \end{description}
    \end{mdframed}
\end{samepage}

\begin{theorem}
    The equivalence relation model of {\sc 3-Coloring} is $\NPS$-complete.
\end{theorem}

\begin{proof}
    We show that the reduction from {\sc 3Sat} to {\sc 3-Coloring} described in Garey and Johnson \cite{DBLP:books/fm/GareyJ79} is an SE reduction from {\sc 3Sat-V} to the {\sc 3-Coloring (Equivalence Relation Model)}.
    
    Let \textsc{3Sat-V} $= (\I, \U, \sol)$ and {\sc 3-Coloring} $= (\I', \U', \sol')$.
    For each instance $I \in \I$, the reduction in Garey and Johnson creates an instance $g(I)$ in which there is an undirected graph $G = (V, E)$.
    For each variable $x_i$ of $\U(I)$, there is a vertex $v_i \in V = \U'(g(I))$.
    The vertex set $V$ also includes three additional vertices labeled $T$, $F$, and $D$; $T$ stands for ``true'', $F$ stands for ``false'', and $D$ stands for ``dummy''.
    There are also three vertices for each clause of $I$.
    The proof that the reduction is correct shows the following:
    $\alpha$ is a satisfying truth assignment of \textsc{3Sat} if and only if there is a 3-coloring of the vertices of $G$ such that:
    (1) if $\alpha(x_i) = 1$, the corresponding vertex $v_i$ has the same color as vertex $T$, and
    (2) if $\alpha(x_i) = 0$, the corresponding vertex $v_i$ has the same color as vertex $F$.
    This reduction is solution-embedding for the equivalence relation model of {\sc 3-Coloring}.
    The function $f$ is defined as follows.
    For every variable $x_i \in \U(I)$, $f(x_i) =  (v_i, T)$.
    This can be interpreted as meaning:
    if $x_i$ is true, then vertex $v_i$ is the same color as vertex $T$. 
\end{proof}

We next define the equivalence relation model of the problem {\sc $4$-Partition} and show that it is $\NPS$-complete.

\begin{samepage}
    \begin{mdframed}
    	\begin{description}
            \item[]\textsc{4-Partition Problem (Equivalence Relation Model)} \hfill\\
            \textbf{Instances $\I$:} \quad A set $A = \{a_1, a_2, ..., a_{4n}\}$, integer $B$, and integer sizes $s(a_i)$ for each $i \in \{1, ..., 4n\}$ such that $B/5 < s(a_i) < B/3$.\\
            \textbf{Universe $\U$:} \quad $\{(a_i, a_j) : i,j \in \{1, \ldots, 4n\}$. \\
            \textbf{Solutions $\sol$:} \quad The set of all $S \subseteq \{(a_i, a_j) : i,j \in \{1, \ldots, 4n\}\}$ such that $S$ is an equivalence relation, and for each equivalence class $C$ of $S$, $\sum_{a_i \in C} s(a_i) = B$. 
        \end{description}
    \end{mdframed}
\end{samepage}

\begin{theorem}
    The equivalence relation model of {\sc 4-Partition} is strongly $\NPS$-complete.
\end{theorem}
\begin{proof}
   The proof is based on the reduction from {\sc 3-Dimensional Matching} to {\sc 4-Partition} as given in Garey and Johnson \cite{DBLP:books/fm/GareyJ79}.
   (We establish the $\NPS$-completeness of {\sc 3-Dimensional Matching} in \Cref{sec:ssp-reductions}.)

    The reduction takes an instance $I$ of 3DM $= (\I, \U, \sol)$ and creates an instance $g(I)$ of {\sc 4-Partition} $= (\I', \U', \sol')$ in which all the element sizes are bounded by a polynomial function of the total number of elements.
    In particular, $\max \{a_i: i \in \{1, ..., 4n\} \} < 2^{16}\,|A|^4$.
    
    Let $W = \{w_1,w_2,\dots,w_q\}$, $X = \{x_1,x_2,\dots,x_q\}$,
    $Y = \{y_1,y_2,\dots,y_q\}$, and let \\$M \subseteq W \times X \times Y$
    denote an arbitrary instance of 3DM.
    The corresponding instance of {\sc 4-Partition} has $4|M|$ elements, one for each occurrence of a member of $W \cup X \cup Y$, in a triple in $M$ and another element for each triple in $M$.
    
    If $z \in W \cup X \cup Y$ occurs in $N(z)$ triples of $M$, then the corresponding elements of $A$ are denoted $z[1], z[2], \dots, z[N(z)]$.
    The element $z[1]$ is referred to as the ``actual'' element corresponding to $z$.
    The elements $z[2]$ through $z[N(z)]$ are referred to as ``dummy'' elements corresponding to $z$.
    The sizes of these elements depend on which of $W,X,Y$ contain $z$, and also depend on whether the element is actual or dummy.
    The parameter $r$ used below is defined to be $32q$, and the sizes of elements of $A$ are chosen as follows:
    
    \[
        \begin{aligned}
            s(w_i[1]) &= 10^4 i + r + 1, && 1 \leq i \leq q, \\
            s(w_i[l]) &= 11r^4 + i r + 1, && 1 \leq i \leq q, \; 2 \leq l \leq N(w_i), \\
            s(x_j[1]) &= 10^4 j^2 + 2, && 1 \leq j \leq q, \\
            s(x_j[l]) &= 11r^4 + j r^2 + 2, && 1 \leq j \leq q, \; 2 \leq l \leq N(x_j), \\
            s(y_k[1]) &= 10^4 + k r^3 + 4, && 1 \leq k \leq q, \\
            s(y_k[l]) &= 8 r^4 + k r^3 + 4, && 1 \leq k \leq q, \; 2 \leq l \leq N(y_k).
        \end{aligned}
    \]
    
    If $m_{\ell} = (w_i, x_j, y_k) \in M$, then the size of $m'_{\ell} \in A$ is 
    $s(m'_{\ell}) = 10^4 - k r^3 - j r^2 - i r + 8$.
    Finally, $B = 40 r^4 + 15$.
       
    Suppose that $m'_{\ell} \in A$ corresponds to  $m_{\ell} = (w_i, x_j, y_k) \in M$.   Garey and Johnson's proof establishes that there are only two ways for the element $m'_{\ell}$ to be in a subset $T \subseteq A$ that sums to $B$. Either $T= \{m'_{\ell}, w_i[1], x_j[1], y_k[1]\}$, or else T contains $m'_{\ell}$ and ``dummy'' elements corresponding to $w_i, x_j$ and $y_k$.    
    
    To see that the reduction is solution-embedding, for each $\ell = 1$ to $q$, we let $f(m_{\ell}) = \langle m_{\ell}, w_i[1] \rangle$.
    If $M'$ is a three dimensional matching, then $f(M')$ can be extended to a 4-partition.
    Similarly, if $S' \subseteq \U'(g(I))$ is a feasible equivalence relation, then $f^{-1}(S' \cap f(\U(I)))$ is a three dimensional matching.
\end{proof}

The final problem in this subsection is {\sc Bin Packing}.
We describe the corresponding equivalence relation model.

\begin{samepage}
    \begin{mdframed}
    	\begin{description}
            \item[]\textsc{Bin Packing (Equivalence Relation Model)} \hfill\\
            \textbf{Instances $\I$:} \quad  Set $A$; integer size $s(a)$ for $a \in A$, a positive integer bin capacity $B$, and an integer $k$. \\
            \textbf{Universe $\U$:} \quad $\{(a, a'): a, a' \in A$ and $a \neq a'\}$.\\
            \textbf{Solutions $\sol$:}  \quad  The set of all equivalence relations $S \{(a, a'): a, a' \in A$ and $a \neq a'\}$ such that for each class $T$ of $S$, $s(T) \le B$ and at most $k$ equivalence classes exist.
        \end{description}
    \end{mdframed}
\end{samepage}

\begin{theorem}
    {\sc Bin Packing} is strongly $\NPS$-complete.
\end{theorem}
\begin{proof}
    There is a straightforward solution-embedding reduction from the equivalence relation model of {\sc 4-Partition} by interpreting each of the four parts as a bin.
\end{proof}

\subsubsection{A universe consisting of subsets of a ground set} 
\label{subsetgen}
In some partition or coloring problems, the maximum number of elements in a part (or with the same color) may be a constant.
When the number of possible parts is polynomially bounded, we can let the universe $\U$ consist of all possible parts.
We refer to this type of representation as \emph{subset generation}.  

Subset generation is another method for breaking the symmetry inherent in some integer programs.
We illustrate with {\sc Covering by Triangles} and {\sc Partition into Triangles}.
Each instance of {\sc Covering by Triangles} consists of a graph $G = (V, E)$ and an integer $k$.
A \emph{triangle} is a set of three vertices of $V$ that are mutually adjacent.
The set $\U$ consists of all of the triangles in $G$.

The solution set $\sol$ consists of any collection of triangles that include all of the vertices of $V$ with at most $k$ triangles.
{\sc Partition into Triangles} is similarly defined, with the difference that the solution set $\sol$ consists of any collection of triangles that include all of the vertices of $V$ exactly once.

\begin{samepage}
    \begin{mdframed}
    	\begin{description}
            \item[]\textsc{Covering by Triangles (Subset Generation Model)} \hfill\\
            \textbf{Instances $\I$:} \quad Undirected graph $G = (V, E)$, integer k.\\
            \textbf{Universe $\U$:} \quad $\{ \{ w, v, x \} :  \{w, v\}, \{w, x\}, \{v, x\} \in E\}$.\\
            \textbf{Solutions $\sol$:} \quad The set of all $S \subseteq \U$ such that $\bigcup _{S_i \in S} S_i = V$ and $|S| \le k$.
        \end{description}
    \end{mdframed}
\end{samepage}

\begin{samepage}
    \begin{mdframed}
    	\begin{description}
            \item[]\textsc{Partition into Triangles (Subset Generation Model)} \hfill\\
            \textbf{Instances $\I$:} \quad Undirected graph $G = (V, E)$, and $n = |V|/3$.\\
            \textbf{Universe $\U$:} \quad $\{ \{ w, v, x \} :  \{w, v\}, \{w, x\}, \{v, x\} \in E\}$.\\
            \textbf{Solutions $\sol$:} \quad The set  of all $S \subseteq \U$ such that $|S| = n/3$, and $\bigcup _{S_i \in S} S_i = V$.
        \end{description}
    \end{mdframed}
\end{samepage}

\begin{theorem}
    The subset generation models of {\sc Covering by Triangles} and {\sc Partition into Triangles} are both $\NPS$-complete. 
\end{theorem}

\begin{proof}
    {\sc Partition into Triangles} is a special case of {\sc Covering by Triangles}.
    So it is sufficient to prove that {\sc Partition into Triangles} is $\NPS$-complete.
    Let $\Pi$ denote the subset generation model of {\sc Partition into Triangles}.
    We will prove that $\Pi$ is $\NPS$-complete based on the reduction from 3DM to $\Pi$, presented in Garey and Johnson \cite{DBLP:books/fm/GareyJ79}.
    (Actually, their transformation is from {\sc Exact Cover by 3-Sets (X3C)}; however, the transformation from 3DM is virtually the same.)
    Let $W, X$, and $Y$ denote the three sets of the 3DM instance $I$.
    Each of these three sets has $n$ distinct elements.
    Let $M$ denote the set of $m$ triples $(w_i, x_j, y_k)$ in the instance $I$.
    A solution consists of $n$ triples of $M$ that include every element of $W \cup X \cup Y$. 

    The reduction from 3DM to $\Pi$ creates an instance $I'$ of $\Pi$ in which the instance $I'$ of $\Pi$ is a graph $G = (V, E)$, where $V$ has the following properties.
    $|V| = 3n + 9m$;
    There is one vertex in $V$ for each element of $W \cup X \cup Y$, and there are 9 vertices in |V| for each triple of $M$.
    If $m_{\ell} = (w_i, x_j, y_k) \in M$, then the 9 additional vertices are labeled $a_{\ell}[1], ..., a_{\ell}[9]$.
    
    The proof establishes a one-to-one correspondence between the matchings $M' \subseteq M$ and the partitions into triangles.
    Moreover, if a matching $M'$ contains $m_{\ell} = (w_i, x_j, y_k)$ if and only if the corresponding partition into triangles contains the edge $(a_{\ell}[1], w_i)$.
    By letting $f(m_{\ell} = (a_{\ell}[1], w_i)$, one can see that the reduction from 3DM to the subset generation model of {\sc Partition into Triangles} is solution-embedding.   
\end{proof}

\subsection{2-step verification models.}  
\label{twostep}

When describing solutions (or certificates) for a problem in $\NP$, usually the solutions are ``clearly specified,'' and the verification task is fairly straightforward.
For {\sc Hamiltonian Cycle}, the solution is a Hamiltonian cycle.
For {\sc Clique}, the solution is the vertices of the clique.
(The set of edges also clearly specifies the solution.)
For {\sc Bin Packing}, a solution specifies the partition of the elements of $A$ into bins.
For {\sc Knapsack problem}, a solution is the subset of elements to be put into the knapsack.

Our final strategy models the set $\sol$ in a manner in which the solutions are only partially specified.
We refer to this modeling approach as \emph{2-step verification}.
It relies on the existence of a polynomial-time algorithm that transforms a partially specified solution $S \in \sol$ into a fully specified solution.
We illustrate on the problem of {\sc Sequencing to Minimize Tardy Tasks}.
We first present an IP-based model.

\begin{samepage}
    \begin{mdframed}
    	\begin{description}
            \item[]\textsc{Sequencing to Minimize Tardy Tasks (IP-based model)} \hfill\\
            \textbf{Instances $\I$:} \quad  Set $T$ of tasks, partial order $\prec$ on elements of $T$; for each task $t$, a length $\ell(t)$, and a deadline $d(t)$, and a positive integer $k \le T$. \\
            \textbf{Universe $\U$:} \quad $\{(t, j): \text{ for } t \in T \text{ and } j \in \{1, ..., n\} \}$ representing the assignment of tasks to the starting times.\\
            \textbf{Solutions $\sol$:}  \quad  The set of all $S \subseteq \{(t, j): \text{ for } t \in T \text{ and } j \in \{1, ..., n\} \}$, such that the tasks are scheduled to start at different times in $\{0, ..., n\}$, and $t$ starts prior to $t'$ whenever $t \prec t'$ and such that the number of tasks that complete after their deadlines is at most $k$.
        \end{description}
    \end{mdframed}
\end{samepage}

The IP-based model is not universe-covering.
Nevertheless, we conjecture that it is not $\NPS$-complete.

One can create a 2-step verification model by making use of the following observation.
There is a feasible schedule for an instance $I$ of {\sc Sequencing to Minimize Tardy Tasks} if there is a subset $S \subseteq T$ of $|T| - k$ tasks that can be scheduled to satisfy all deadlines.
One can obtain a feasible schedule for $S$ (if one exists) as follows.
First, adjust the deadlines of jobs in $S$ so that if $t, t' \in S$ and $t \prec t'$, then $d(t) \le d(t') - 1$.
Subsequent to modifying deadlines, it suffices to schedule jobs in order of their deadline.

The following is a revised model for the problem that makes use of the above observation.
Note that in the description below, the solution consists of jobs scheduled on time.

\begin{samepage}
    \begin{mdframed}
    	\begin{description}
            \item[]\textsc{Sequencing to Minimize Tardy Tasks (2-step verification model)} \hfill\\
            \textbf{Instances $\I$:} \quad  Set $T$ of $n$ tasks, for each task $t$, a length $\ell(t)$, a partial order $\prec$ on $T$, for each task $t \in T$ a deadline $d(t)$, and an integer $k \le n$.\\
            \textbf{Universe $\U$:} \quad The set of tasks $T$.\\
            \textbf{Solutions $\sol$:}  \quad The set of all $S \subseteq T$ such that it is possible to schedule all jobs in $S$ to meet their deadlines and $|S| \ge |T|-k$.
        \end{description}
    \end{mdframed}
\end{samepage}

Given a solution $S \in \sol$, one can create a feasible schedule in polynomial time. We view this polynomial time construction as the "second step" in the verification procedure.

\begin{theorem}
    The 2-step verification model of {\sc Sequencing to Minimize Tardy Tasks} is $\NPS$-complete.
\end{theorem}

\begin{proof}
    The proof is based on a reduction from {\sc Clique (Edge-based)}, see Garey and Johnson \cite{DBLP:books/fm/GareyJ79}.
    (We establish the $\NPS$-completeness of {\sc Clique (Edge-based)} in \Cref{sec:ssp-reductions}.)
    Consider an instance $I$ of determining whether there is an edge set $S$ of size $k$ defining a clique in a graph $G = (V, E)$.
    That is, any solution $S$ consists of the $k=j(j-1)/2$ edges of a clique with $j$ vertices.
    The transformed instance $g(I)$ of the scheduling problem is the following.
    For each vertex $v \in V$, there is a task $v \in T$, with deadline $|V| + |E|$.
    We let $V'$ denote this subset of tasks.
    For each edge $\{v, w\} \in E$, there is a task $vw \in T$ with deadline $k+j$.
    We let $E'$ denote this subset of tasks.
    In addition, $v \prec vw$, and $w \prec vw$.
    Thus, $T = V' \cup E'$, and $|T| = n + m$.
    The mapping from edges to tasks is: $f(\{v, w\}) = vw$.
    
    Note that all tasks in $V'$ can be scheduled on time regardless of their starting times.
    At  most $k$ tasks in $E'$ can be scheduled on time.
    And one can schedule a set $S' \subseteq E'$ of $k$ tasks on time only if there is a $f^{-1}(S')$ is a clique of $E$ with $k$ edges.
    Thus, the reduction is solution-embedding.
\end{proof}

\section{Conclusion}

We have presented a general framework for understanding the complexity of multilevel
optimization problems derived from $\NP$.
By introducing the class \emph{$\NP$ with solutions} ($\NPS$) and formalizing solution-embedding reductions, we isolate a structural property of many classical $\NP$-completeness reductions that is crucial for lifting hardness results beyond
$\NP$.
This perspective allows us to define $\NPS$-completeness and to identify a large collection of natural problems -- ranging from satisfiability and graph problems to knapsack and routing problems -- that satisfy this stronger notion of completeness.

The central consequence of this framework is a collection of meta-theorems that systematically lift $\NP$-completeness results to higher levels of the polynomial hierarchy and to $\PSPACE$.
We apply these meta-theorems to several fundamental classes of multilevel problems, including adversarial selection games, interdiction and protection-interdiction problems, and multi-stage adjustable robust optimization.
For $\NPS$-complete problems, we show that natural two-level variants -- such as adversarial selection and interdiction -- are
complete for $\Sigma_2^p$, while three-level variants are complete for $\Sigma_3^p$.
More generally, we establish that $k$-stage extensions of these problems are complete for $\Sigma_k^p$, and that allowing the number of stages to grow with the input size leads to $\PSPACE$-completeness.

These results contain several previously known completeness results for multilevel optimization problems derived from $\NP$ as special cases, and they yield many new ones simultaneously.
In contrast to earlier work, which typically establishes hardness for individual problems via problem-specific reductions, our framework provides a unified explanation for the complexity of entire families of multilevel problems.
In particular, adversarial selection, interdiction, and adjustable robust optimization emerge as closely related manifestations of the same underlying phenomenon:
the interaction between solution structure and sequential decision making.

Beyond providing new completeness results, our work helps explain why multilevel extensions of $\NP$-complete problems are often significantly harder than their single-level counterparts.
The explicit treatment of solutions in $\NPS$ clarifies that different representations of solutions -- even for the same underlying $\NP$ decision problem -- can lead to fundamentally different multilevel problems with different complexities.
This observation highlights the importance of solution structure in complexity theory and suggests that $\NP$-hardness alone is
often insufficient for understanding the true difficulty of multilevel optimization problems.

Several directions for future research remain.
While $\NPS$ captures a broad class of combinatorial problems, it would be interesting to further refine this framework or to identify alternative abstractions that characterize additional forms of multilevel or dynamic optimization.
Moreover, our results are primarily complexity-theoretic; understanding how these hardness results inform algorithm design, approximation guarantees, or parameterized
approaches for adversarial selection, interdiction, and adjustable robust optimization remains an important open direction.
More broadly, the close connection between solution embedding, alternation, and space-bounded computation suggests that similar structural techniques may apply to other settings in which problems are lifted from $\NP$ to richer computational models.

Overall, our results demonstrate that high computational complexity is not an exception but a generic feature of multilevel extensions of $\NP$-complete problems.
By unifying adversarial selection, interdiction, and multi-stage adjustable robust optimization within a single complexity-theoretic framework, we hope this work contributes to a deeper structural understanding of optimization problems in the polynomial hierarchy and beyond.

\appendix
\newpage
\section{A Compendium of NP-S-complete Problems}
\label{sec:ssp-reductions}

In this section, we present a multitude of $\NPS$-complete problems.
For each problem, we provide the definitions of the set $\I$ of instances, the universe $\U(I)$, and the solution set $\sol(I)$ associated to each of the instances $I \in \I$.
Additionally, we state the solution-embedding reduction together with a reference.
For some problems, we provide additional comments directly below the definition, in which we describe modifications of the original reduction.
All of the presented and referenced reductions in this section are already known and can be found in the literature.
In \Cref{fig:reductions}, the tree of all presented reductions can be found, beginning at \textsc{Satisfiability (Variable-based)}.
We heavily rely on the transitivity of SE reductions (\cref{lem:SE-transitive}).

We remark that the website \emph{reductions.network} \cite{DBLP:journals/corr/abs-2511-04308} collects solutions embedding reductions for problems in $\NPS$, which also includes results by Pfaue \cite{DBLP:journals/corr/abs-2411-05796} and Bartlett \cite{DBLP:journals/corr/abs-2506-12255}.

With the tree of SE reductions shown in \Cref{fig:reductions}, we derive the following theorem.
\begin{samepage}
    \begin{theorem}
        The following problems are $\NPS$-complete:
        \textsc{Satisfiability (Literal-based)},
        \textsc{Satisfiability (Variable-based)},
        \textsc{3-Satis\-fiability (Literal-based)},
        \textsc{3-Satis\-fiability (Variable-based)},
        \textsc{Vertex Cover},
        \textsc{Dominating Set},
        \textsc{Set Cover},
        \textsc{Hitting Set},
        \textsc{Feedback Vertex Set},
        \textsc{Feedback Arc Set},
        \textsc{Uncapacitated Facility Location},
        \textsc{p-Center},
        \textsc{p-Median},
        \textsc{3-Dimensional Matching},
        \textsc{Independent Set},
        \textsc{Clique (Vertex-based)},
        \textsc{Clique (Edge-based)},
        \textsc{Subset Sum},
        \textsc{Knapsack},
        \textsc{Partition},
        \textsc{Scheduling},
        \textsc{Directed Hamiltonian Path},
        \textsc{Directed Hamiltonian Cycle},
        \textsc{Undirected Hamiltonian Cycle},
        \textsc{Traveling Salesman Problem},
        \textsc{Two Directed Vertex Disjoint Path},
        \textsc{$k$-Vertex Directed Disjoint Path},
        \textsc{Steiner Tree}
    \end{theorem}
\end{samepage}

\begin{figure}[!ht]
	\centering
	\scalebox{0.67}{
	\begin{tikzpicture}[scale=1]
		\node[text width=1.5cm,align=center](sat) at (0, 0.75) {\textsc{Sat-L}};
		\node[text width=1.5cm,align=center](satV) at (3, 0.75) {\textsc{Sat-V}};
		\node[text width=1.5cm,align=center](0) at (0, 0) {\textsc{3Sat-L}};
		\node[text width=1.5cm,align=center](0V) at (3, 0) {\textsc{3Sat-V}};
        \path[->] (sat) edge (satV);
        \path[->] (0) edge (0V);
		\node[](1) at (-8, -1) {\textsc{VC}};
        \node[](01) at (-5, -1) {\textsc{3DM}};
		\node[](2) at (-2.75, -1) {\textsc{IS}};
		\node[](3) at (0, -1) {\textsc{Subset Sum}};
		\node[](4) at (6.5, -0.85) {}; \node[text width=1cm,align=center](4-1) at (7, -1.19) {\textsc{Steiner Tree}};
		\node[](5) at (3.25, -1) {\textsc{DHamPath}};
		\node[](6) at (5.5, -1) {\textsc{2DDP}};
		\node[](11) at (-12.5, -2) {\textsc{DS}};
		\node[](12) at (-11.5, -2) {\textsc{SC}};
		\node[](13) at (-10.5, -2) {\textsc{HS}};
		\node[](14) at (-9.5, -2) {\textsc{FVS}};
		\node[](15) at (-8.5, -2) {\textsc{FAS}};
		\node[](10) at (-7.6, -2) {\textsc{UFL}};
		\node[](16) at (-6.25, -2) {\textsc{p-Center}};
		\node[](17) at (-4.5, -2) {\textsc{p-Median}};
		\node[](21) at (-2.75, -2) {\textsc{Clique-V}};
		\node[](31) at (-1, -2) {\textsc{Knapsack}};
		\node[](32) at (1, -2) {\textsc{Partition}};
		\node[](51) at (3.25, -2) {\textsc{DHamCyc}};
		\node[](61) at (5.5, -2) {\textsc{kDDP}};
		\node[](211) at (-2.75, -2.75) {\textsc{Clique-E}};
		\node[](321) at (1, -2.75) {\textsc{Scheduling}};
		\node[](511) at (3.25, -2.75) {\textsc{UHamCyc}};
        \node[](5111) at (3.25, -3.5) {\textsc{TSP}};
		\path[->] (sat) edge (0);
		\path[->] (0) edge (1);
		\path[->] (0) edge (01);
		\path[->] (0) edge (2);
		\path[->] (0) edge (3);
		\path[->] (0) edge (4);
		\path[->] (0) edge (5);
		\path[->] (0) edge (6);
		\path[->] (1) edge (10);
		\path[->] (1) edge (11);
		\path[->] (1) edge (12);
		\path[->] (1) edge (13);
		\path[->] (1) edge (14);
		\path[->] (1) edge (15);
		\path[->] (1) edge (16);
		\path[->] (1) edge (17);
		\path[->] (2) edge (21);
		\path[->] (21) edge (211);
		\path[->] (3) edge (31);
		\path[->] (3) edge (32);
		\path[->] (32) edge (321);
		\path[->] (5) edge (51);
		\path[->] (51) edge (511);
		\path[->] (511) edge (5111);
		\path[->] (6) edge (61);
	\end{tikzpicture}
	}
	\caption{The tree of solution-embedding reductions for all considered problems.}
	\label{fig:reductions}
\end{figure}

\begin{samepage}
    \begin{mdframed}
    	\begin{description}
        \item[]\textsc{Satisfiability (Literal-based)} \hfill\\
        \textbf{Instances $\I$:} \quad Literal Set $L = \fromto{\ell_1}{\ell_n} \cup \fromto{\overline \ell_1}{\overline \ell_n}$, Clauses $C \subseteq \powerset{L}$.\\
        \textbf{Universe $\U$:} \quad Literal set $L$.\\
        \textbf{Solutions $\sol$:} \quad All sets $L' \subseteq L$ such that for all $i \in \fromto{1}{n}$ we have $|L' \cap \set{\ell_i, \overline \ell_i}| = 1$, and such that $|L' \cap c_j| \geq 1$ for all $c_j \in C$, $j \in \fromto{1}{|C|}$.   	
        \end{description}
        \hfill\\
        Reference.\quad  Cook.  \cite{DBLP:conf/stoc/Cook71} 
    \end{mdframed}
\end{samepage}

\begin{samepage}
    \begin{mdframed}
    	\begin{description}
        \item[]\textsc{Satisfiability (Variable-based)} \hfill\\
        \textbf{Instances $\I$:} \quad Variables $X = \fromto{x_1}{x_n}$, Clauses $C = \fromto{C_1}{C_m}$.\\
        \textbf{Universe $\U$:} \quad Variable set $X = \fromto{x_1}{x_n}$.\\
        \textbf{Solutions $\sol$:} \quad $S \subseteq X$ such that $S$ is a truth assignment that satisfies all of the clauses of $C$, and $S$ is the set of variables of $X$ that are true.   	
        \end{description}
        \hfill\\
        Reference.\quad  Cook.  \cite{DBLP:conf/stoc/Cook71} 
    \end{mdframed}
\end{samepage}
For both \textsc{Satisfiability (Literal-based)} and \textsc{Satisfiability (Variable-based)}, Cook's result on the $\NP$-completeness for \textsc{Satisfiability} also shows their $\NPS$-completeness.
The existence of a reduction between \textsc{Satisfiability (Literal-based)} and \textsc{Satisfiability (Variable-based)} and vice versa is provided by \Cref{satvandlequiv}.  

\begin{samepage}
    \begin{mdframed}
    	\begin{description}
            \item[]\textsc{3-Satisfiability (Literal-based)} \hfill\\
            \textbf{Instances $\I$:} \quad Literal Set $L = \fromto{\ell_1}{\ell_n} \cup \fromto{\overline \ell_1}{\overline \ell_n}$, Clauses $C \subseteq L^3$.\\
            \textbf{Universe $\U$:} \quad Literal set $L$.\\
            \textbf{Solutions $\sol$:} \quad All sets $L' \subseteq \U$ such that for all $i \in \fromto{1}{n}$ we have $|L' \cap \set{\ell_i, \overline \ell_i}| = 1$, and such that $|L' \cap c_j| \geq 1$ for all $c_j \in C$, $j \in \fromto{1}{|C|}$.
        \end{description}
        \hfill\\
        SE Reduction: \quad From \textsc{Satisfiability (Literal-based)}\\
        Reference.\quad Karp, \cite{DBLP:conf/coco/Karp72}
    \end{mdframed}
\end{samepage}

Analogous to \textsc{Satisfiability}, there is a variant of \textsc{3-Satisfiability (Variable-based)} referred to as \textsc{3-Satisfiability (Literal-based)} in which the universe consists of the literals.
Analogous to \Cref{satvandlequiv} for \textsc{Satisfiability}, \textsc{3Satisfiability (Literal-based)} and \textsc{3Satisfiability (Variable-based)} can be reduced to each other.

\begin{samepage}
    \begin{mdframed}
    	\begin{description}
            \item[]\textsc{3-Satisfiability (Variable-based)} \hfill\\
            \textbf{Instances $\I$:} \quad Variables $X = \fromto{x_1}{x_n}$, Clauses $C = \fromto{C_1}{C_m}$ of size 3.\\
            \textbf{Universe $\U$:} \quad Variable set $X$.\\
            \textbf{Solutions $\sol$:} \quad The subsets of $X$ that correspond to satisfying truth assignments.
        \end{description}
        \hfill\\
        SE Reduction: \quad From \textsc{Satisfiability (Variable-based)}\\
        Reference.\quad Karp, \cite{DBLP:conf/coco/Karp72}
    \end{mdframed}
\end{samepage}

\begin{samepage}
    \begin{mdframed}
    	\begin{description}
            \item[]\textsc{Vertex Cover} \hfill\\
            \textbf{Instances $\I$:} \quad Graph $G = (V, E)$, number $k \in \N$.\\
            \textbf{Universe $\U$:} \quad Vertex set $V$.\\
            \textbf{Solutions $\sol$:} \quad The set of all $S \subseteq V$ such that $|S \cap e| \geq 1$ for all $e \in E$ and $|S| \leq k$.
    	\end{description}
        \hfill\\
        SE Reduction.\quad From \textsc{3-Satisfiability (Literal-based)}.\\
        Reference.\quad  Garey and Johnson  \cite{DBLP:books/fm/GareyJ79}
    \end{mdframed}
\end{samepage}

\begin{samepage}
    \begin{mdframed}
    	\begin{description}
            \item[]\textsc{Dominating Set} \hfill\\
            \textbf{Instances $\I$:} \quad Graph $G = (V, E)$, number $k \in \N$.\\
            \textbf{Universe $\U$:} \quad Vertex set $V$.\\
            \textbf{Solutions $\sol$:} \quad The set of all $S \subseteq V$ such that $|S \cap N[v]| \neq \emptyset$ for all $v \in V$ and $|S| \leq k$.    	
        \end{description}
        \hfill\\
        SE Reduction: \quad From \textsc{Vertex Cover}.\\
        Reference.\quad  Folklore  
    \end{mdframed}
\end{samepage}

In the folklore reduction from \textsc{Vertex Cover} to \textsc{Dominating Set}, for each arc $(v, w) \in E$, the transformed problem has an additional vertex $vw$ and additional edges $(v, vw)$ and $(w, vw)$.
Any vertex cover $S \subseteq V$ is also a dominating set of the transformed graph $G' = (V', E')$.
However, this reduction is not solution-embedding because there may be a dominating set $S' \subseteq V'$ with $|S'| \le k$ and such that $vw \in S'$.
The transformation can be made solution-embedding as follows:
For each each $(u, v) \in E$, add $k+1$ vertices $vw_1, vw_2, \dots vw_{k+1}$, each of which is adjacent to $v$ and $w$. 

\begin{samepage}
    \begin{mdframed}
    	\begin{description}
            \item[]\textsc{Set Cover}\hfill\\
            \textbf{Instances $\I$:} \quad Sets $S_i \subseteq \fromto{1}{m}$ for $i \in \fromto{1}{n}$, number $k \in \N$.\\
            \textbf{Universe $\U$:} \quad $\{S_1, \dots, S_n\}$.\\
            \textbf{Solutions $\sol$:} \quad The set of all $S \subseteq \{S_1, \dots, S_n\}$ such that $\bigcup_{s \in S} s = \{1, \ldots, m\}$ and $|S| \le k$.
        \end{description}
        \hfill\\
        SE Reduction: \quad From \textsc{Vertex Cover}.\\
        Reference.\quad  Karp \cite{DBLP:conf/coco/Karp72}.
    \end{mdframed}
\end{samepage}

\begin{samepage}
    \begin{mdframed}
    	\begin{description}
            \item[]\textsc{Hitting Set}\hfill\\
            \textbf{Instances $\I$:} \quad Sets $S_j \subseteq \fromto{1}{n}$ for $j \in \fromto{1}{m}$, number $k \in \N$.\\
            \textbf{Universe $\U$:} \quad $\{1, \dots, n\}$.\\
            \textbf{Solutions $\sol$:} \quad The set of all $H \subseteq \{1, \ldots, n\}$ such that $H \cap S_j \neq \emptyset$ for all $j \in \fromto{1}{m}$ and $|H| \leq k$.
        \end{description}
        \hfill\\
        SE Reduction: \quad From \textsc{Vertex Cover}.\\
        Reference.\quad  Karp \cite{DBLP:conf/coco/Karp72}.
    \end{mdframed}
\end{samepage}

\begin{samepage}
    \begin{mdframed}
    	\begin{description}
            \item[]\textsc{Feedback Vertex Set} \hfill\\
            \textbf{Instances $\I$:} \quad Directed Graph $G = (V, A)$, number $k \in \N$.\\
            \textbf{Universe $\U$:} \quad Vertex set $V$.\\
            \textbf{Solutions $\sol$:} \quad The set of all $S \subseteq V$ such that $G[V \setminus S]$ is acyclic and $|S| \leq k$.
        \end{description}
        \hfill\\
        SE Reduction: \quad From \textsc{Vertex Cover}.\\
        Reference.\quad  Karp \cite{DBLP:conf/coco/Karp72}.
    \end{mdframed}
\end{samepage}

\begin{samepage}
    \begin{mdframed}
    	\begin{description}
            \item[]\textsc{Feedback Arc Set}\hfill\\
            \textbf{Instances $\I$:} \quad Directed Graph $G = (V, A)$, number $k \in \N$.\\
            \textbf{Universe $\U$:} \quad Arc set $A$.\\
            \textbf{Solutions $\sol$:} \quad The set of all $S \subseteq A$ such that the arc set $G' = (V, A \setminus S)$ is acyclic and $|S| \leq k$.
        \end{description}
        \hfill\\
        SE Reduction: \quad From \textsc{Vertex Cover}.\\
        Reference.\quad  Karp \cite{DBLP:conf/coco/Karp72}.
    \end{mdframed}
\end{samepage}
A modification of the reduction by Karp \cite{DBLP:conf/coco/Karp72} from \textsc{Vertex Cover} to \textsc{Feedback Arc Set} is an SSP reduction.
Let $I = (G,k) = ((V, E),k)$ be the \textsc{Vertex Cover} instance and $(G',k') = ((V', A'),k')$ the \textsc{Feedback Arc Set} instance.
First of all, we transform the vertices $v \in V$ to two vertices $v_0, v_1 \in V'$.
We define the injective embedding function $f_I$ by mapping each vertex $v \in V$ to the arc $(v_0, v_1) \in A'$, which is also the corresponding element in all solutions, that is $f_I(v) = (v_0, v_1)$.
At last, we transform each edge $\{v, w\} \in E$ to $|V + 1|$ once subdivided arcs from $v_1$ to $w_0$ and to $|V + 1|$ once subdivided arcs from $w_1$ to  $v_0$. Finally, we leave the parameter $k = k'$ unchanged.
By deleting the arc $(v_0, v_1)$, which corresponds to vertex $v$ in $G$ all of these induced cycles are disconnected.
This implies that every vertex cover of $G$ is translated to a feedback arc set by the function $f_I$.

\begin{samepage}
    \begin{mdframed}
    	\begin{description}
            \item[]\textsc{Uncapacitated Facility Location} \quad \textsc{(UFL)} \hfill\\
            \textbf{Instances $\I$:} \quad Set of potential facilities $F = \fromto{1}{n}$, set of clients $C = \fromto{1}{m}$, fixed cost of opening facility function $f: F \rightarrow \Z$, service cost function $c: F \times C \rightarrow \Z$, cost threshold $k \in \Z$.\\
            \textbf{Universe $\U$:} \quad Potential facilities $F$.\\
            \textbf{Solutions $\sol$:} \quad The set of of all $F' \subseteq F$ s.t. $\sum_{i \in F'} f(i) + \sum_{j \in C} \min_{i \in F'} c(i, j) \leq k$.
        \end{description}
        \hfill\\
        SE Reduction: \quad From \textsc{Vertex Cover}.\\
        Reference.\quad  Cornuéjols, Nemhauser, and Wolsey \cite{cornuejols1983uncapicitated} 
    \end{mdframed}
\end{samepage}

\begin{samepage}
    \begin{mdframed}
    	\begin{description}
            \item[]\textsc{p-Center}  \hfill\\
            \textbf{Instances $\I$:} \quad Set of potential facilities $F = \fromto{1}{n}$, set of clients $C = \fromto{1}{m}$, service cost function $c: F \times C \rightarrow \Z$, facility threshold $p \in \N$, cost threshold $k \in \Z$.\\
            \textbf{Universe $\U$:} \quad Potential facilities $F$.\\
            \textbf{Solutions $\sol$:} \quad The set of all $F' \subseteq F$ s.t. $|F'| \leq p$ and $\max_{j \in C} \min_{i \in F'} c(i, j) \leq k$.
        \end{description}
        \hfill\\
        SE Reduction: \quad From \textsc{Vertex Cover}.\\
        Reference.\quad  Cornuéjols, Nemhauser, and Wolsey \cite{cornuejols1983uncapicitated}.
    \end{mdframed}
\end{samepage}

A modified version of the reduction by Cornuéjols, Nemhauser, and Wolsey \cite{cornuejols1983uncapicitated} from \textsc{Vertex Cover} to \textsc{Uncapacitated Facility Location} is an SSP reduction.
Let $I = ((V,E), k)$ be the \textsc{Vertex Cover} instance and $(F, C, c, p, k')$ be the \textsc{p-Center} instance.
We map each $v \in V$ to $v \in F$ and each $e \in E$ to $e \in C$.
Further, we define $c(v, e) = 0$ if $v \in e$ and $c(v, e) = |V|+1$ otherwise.
At last, we set $p$ equal to the size $k$ of the vertex cover and $k' = 0$.
Note that this implies that in a solution the objective has to be $0$.

\begin{samepage}
    \begin{mdframed}
    	\begin{description}
            \item[]\textsc{\textsc{p-Median}}  \hfill\\
            \textbf{Instances $\I$:} \quad Set of potential facilities $F = \fromto{1}{n}$, set of clients $C = \fromto{1}{m}$, service cost function $c: F \times C \rightarrow \Z$, facility threshold $p \in \N$, cost threshold $k \in \Z$.\\
            \textbf{Universe $\U$:} \quad Potential facilities $F$.\\
            \textbf{Solutions $\sol$:} \quad The set of all $F' \subseteq F$ s.t. $|F'| \leq p$ and $\sum_{j \in C} \min_{i \in F'} c(i, j) \leq k$.
        \end{description}
        \hfill\\
        SE Reduction: \quad From \textsc{Vertex Cover}.\\
        Reference.\quad  Cornuéjols, Nemhauser, and Wolsey \cite{cornuejols1983uncapicitated}.
    \end{mdframed}
\end{samepage}

A modified version of the reduction by Cornuéjols, Nemhauser, and Wolsey \cite{cornuejols1983uncapicitated} from \textsc{Vertex Cover} to \textsc{Uncapacitated Facility Location} is an SSP reduction.
It is the same as for \textsc{p-Center}.
Let $I = ((V,E), k)$ be the \textsc{Vertex Cover} instance and $(F, C, c, p, k')$ be the \textsc{p-Median} instance.
We map each $v \in V$ to $v \in F$ and each $e \in E$ to $e \in C$.
Further, we define $c(v, e) = 0$ if $v \in e$ and $c(v, e) = |V|+1$ otherwise.
At last, we set $p = k$ equals to the size of the vertex cover and $k' = 0$.

\begin{samepage}
    \begin{mdframed}
    	\begin{description}
            \item[]\textsc{3-Dimensional Matching} \hfill\\
            \textbf{Instances $\I$:} \quad Sets $W,X,Y \in \N^q$, and set of triples $M \subseteq W \times X \times Y$.\\
            \textbf{Universe $\U$:} \quad The set of triples $M$\\
            \textbf{Solutions $\sol$:} \quad The set of all $S \subseteq M$ such that $|S| = q$.
    	\end{description}
        \hfill\\
        SE Reduction.\quad From \textsc{3-Satisfiability (Literal-based)}.\\
        Reference.\quad  Garey and Johnson \cite{DBLP:books/fm/GareyJ79}
    \end{mdframed}
\end{samepage}

\begin{samepage}
    \begin{mdframed}
    	\begin{description}
            \item[]\textsc{\textsc{Independent Set}} \hfill\\
            \textbf{Instances $\I$:} \quad Graph $G = (V,E)$, number $k \in \N$.\\
            \textbf{Universe $\U$:} \quad Vertex set $V$.\\
            \textbf{Solutions $\sol$:} \quad The set of all $S \in \subseteq V$ such that for all $i, j \in V'$, $\{i, j\} \notin E$, and $|S| \ge k$.
        \end{description}
        \hfill\\
        SE Reduction: \quad From \textsc{3-Satisfiability (Literal-based)}.\\
        Reference.\quad  Garey and Johnson \cite{DBLP:books/fm/GareyJ79} (Modification of the reduction to \textsc{Vertex Cover})
    \end{mdframed}
\end{samepage}

Let $I = (L,C)$ be the \textsc{3Sat} instance.
We define a  corresponding \textsc{Vertex Cover} instance $((V',E'), k')$.
Every literal $\ell \in L$ is transformed to a vertex $v_\ell \in V'$ and every pair $(\ell, \overline \ell) \in L \times L$ is again transformed to an edge $\{v_\ell, v_{\overline \ell}\} \in E'$ between the corresponding literal vertices.
Every clause $c \in C$ is again transformed to a 3-clique, where each vertex $v^{c}_{\ell_{i_1}}, v^{c}_{\ell_{i_2}}, v^{c}_{\ell_{i_3}}$ represents a literal in the clause.
In contrast to the \textsc{Vertex Cover} reduction, the clause vertices are connected to the opposite literal vertex.
Finally, we define the parameter $k'$ by $k' := |L|/2 + |C|$.
The universe elements of \textsc{3-Satisfiability (Literal-based)} are injectively mapped to the literal vertices in $V'$, where $f_{I}(\ell) = v_\ell \in V'$.

\begin{samepage}
    \begin{mdframed}
    	\begin{description}
            \item[]\textsc{Clique (Vertex-based)} \hfill\\
            \textbf{Instances $\I$:} \quad Graph $G = (V, E)$, number $k \in \N$.\\
            \textbf{Universe $\U$:} \quad Vertex set $V$.\\
            \textbf{Solutions $\sol$:} \quad $S \subseteq V$ such that for all $u, v \in V'$, $\{u, v\} \in E$ and $|S| \geq k$.  
        \end{description}
        \hfill\\
        SE Reduction: \quad From \textsc{Independent Set}.\\
        Reference.\quad  Garey and Johnson \cite{DBLP:books/fm/GareyJ79}.
    \end{mdframed}
\end{samepage} 

\begin{samepage}
    \begin{mdframed}
    	\begin{description}
            \item[]\textsc{Clique (Edge-based)}\hfill\\
            \textbf{Instances $\I$:} \quad Graph $G = (V, E)$, number $k \in \N$.\\
            \textbf{Universe $\U$:} \quad Edge set $E$.\\
            \textbf{Solutions $\sol$:} \quad The set of all $S \subseteq E$ such that for all pairs $(\{u,v\}, \{u,w\}) \in S^2$, $\{v,w\} \in S$, and $|S| \geq k$.
        \end{description}
        \hfill\\
        SE Reduction: \quad From \textsc{Clique (Vertex-based)}.
    \end{mdframed}
\end{samepage} 

Given an instance of \textsc{Clique (Vertex-based)} consisting graph $G = (V, E)$ and an integer $k$, we create an instance $G' = (V \cup V', E')$ and integer $k' = 2k$ for \textsc{Clique (Edge-based)} as follows.
The vertex subset $V'$ is a copy of $V$.
For each vertex $v \in V$, there are vertices $v$ and $v' \in V'$ and an edge $(v, v') \in E'$.
For each edge $(v, w) \in G$, there are edges $(v, w), (v', w), (v, w')$, and $(v', w') \in E'$.
It is easy to see that there is a one-to-one correspondence from cliques in $G$ of size $k$ to cliques in $G'$ of size $2k$.
The reduction is solution-embedding:
For each vertex $v$ of $G$, $f(v) = (v, v') \in E'$.

\begin{samepage}
    \begin{mdframed}
    	\begin{description}
                \item[]\textsc{Subset Sum} \hfill\\
                \textbf{Instances $\I$:} \quad Numbers $\fromto{a_1}{a_n} \subseteq \N$, and target value $M \in \N$.\\
                \textbf{Universe $\U$:} \quad $\fromto{a_1}{a_n}$.\\
                \textbf{Solutions $\sol$:} \quad $S \subseteq \U$ such that $\sum_{i\in S} a_i = M$.
        \end{description}
        \hfill\\
        SE Reduction: \quad From \textsc{3-Satisfiability (Literal-based)}\\
        Reference.\quad Sipser \cite{DBLP:books/daglib/0086373}
    \end{mdframed}
\end{samepage}

\begin{samepage}
    \begin{mdframed}
    	\begin{description}
            \item[]\textsc{Knapsack} \hfill\\
            \textbf{Instances $\I$:} \quad Objects with profits and weights $\{(p_1, w_1), \dots, (w_n, p_n)\} \in \N^2$, bounds $W, P \in \N$.\\
            \textbf{Universe $\U$:} \quad The objects $\{(p_1, w_1), \dots, (p_n, w_n)\}$.\\
            \textbf{Solutions $\sol$:} \quad The set of all $S \subseteq \U$ with $\sum_{(p_i, w_i) \in S} w_i \leq W$ and $\sum_{(p_i, w_i) \in S} p_i \geq P$.
        \end{description}
        \hfill\\
        SE Reduction: \quad From \textsc{Subset Sum}.\\
        Reference.\quad  Folklore.      	
    \end{mdframed}
\end{samepage}
For the reduction, set $p_i = w_i = a_i$ and $W = P = M$.

\begin{samepage}
    \begin{mdframed}
    	\begin{description}
                \item[]\textsc{Partition} \hfill\\
                \textbf{Instances $\I$:} \quad Numbers $\fromto{a_1}{a_n} \subseteq \N$.\\
                \textbf{Universe $\U$:} \quad $\fromto{a_1}{a_n}$.\\
                \textbf{Solutions $\sol$:} \quad $S \subseteq \U$ such that $a_1 \in S$, and $\sum_{i \in S} a_i = \sum_{i \notin S} a_i$.
        \end{description}
        \hfill\\
        SE Reduction: \quad From \textsc{Subset Sum}.\\
        Reference.\quad  Folklore.
    \end{mdframed}
\end{samepage}

Note that we require the first element to be in the solution.
If this requirement is relaxed, then there cannot be a solution-embedding reduction from an $\NPS$-complete problem.
To see why, note that for every subset $S \subseteq U$ such that $\sum_{i\in S} a_i = b$, it is also true that $\sum_{i\notin S} a_i = b$, and thus $U\setminus S \in \sol$.
Thus, if we relax the requirement that $1 \in S$, there cannot be a solution-embedding reduction from an $\NPS$-complete problem.    

Let $a_1, \ldots, a_n, M$ be the {\sc Subset Sum} instance.
We define the {\sc Partition} instance $a'_1, \ldots, a'_{n+2}$ by setting $a'_1 = M+1$, $a'_2 = \sum_i a_i + 1 - M$, and $a'_{i+2} = a_{i}$ for $i \in \{1, \ldots, n\}$.
The solution-embedding function $f$ is then defined by $f(a_i) = a'_{i+2}$.

\begin{samepage}
    \begin{mdframed}
    	\begin{description}
            \item[]\textsc{Two Machine Makespan Scheduling} \hfill\\
            \textbf{Instances $\I$:} \quad Processing times $\fromto{t_1}{t_n} \subseteq \N$, threshold $T \in \N$.\\
            \textbf{Universe $\U$:} \quad Processing times $\{t_1, \ldots, t_n\}$.\\
            \textbf{Solutions $\sol$:} \quad The set of all $S \subseteq \{t_1, \ldots, t_n\}$ such that $t_1 \in S$, $\sum_{t_i \in S} t_i \leq T$, and  $\sum_{t_i \notin S} t_i \leq T$.
        \end{description}
        \hfill\\
        SE Reduction: \quad From \textsc{Partition}.\\
        Reference.\quad  Folklore. 
    \end{mdframed}
\end{samepage}

Here we break the symmetry between the two machines as we did in \textsc{Partition}.
We require that the first task is processed on the first machine. 

For this, let $I = \{a_1, \ldots, a_n\}$ be the \textsc{Partition} instance.
We transform each number $a_i$ in the \textsc{Partition} instance to a job with processing time $a_i$ in the \textsc{Two Machine Makespan Scheduling} instance and set the threshold $T = \frac{1}{2} \sum_i a_i$.

\begin{samepage}
    \begin{mdframed}
    	\begin{description}
            \item[]\textsc{Directed Hamiltonian Path} \hfill\\
            \textbf{Instances $\I$:} \quad Directed Graph $G = (V, A)$, vertices $s, t \in V$.\\
            \textbf{Universe $\U$:} \quad Arc set $A$.\\
            \textbf{Solutions $\sol$:} \quad The set of all $S \subseteq A$ such that $S$ is an $s$-$t$-path of length $|V|-1$.
        \end{description}
        \hfill\\
        SE Reduction: \quad From \textsc{3-Satisfiability (Literal-based)}.\\
        Reference.\quad  Arora and Barak \cite{DBLP:books/daglib/0023084}.
    \end{mdframed}
\end{samepage}

\begin{samepage}
    \begin{mdframed}
    	\begin{description}
            \item[]\textsc{Directed Hamiltonian Cycle} \hfill\\
            \textbf{Instances $\I$:} \quad Directed Graph $G = (V, A)$.\\
            \textbf{Universe $\U$:} \quad Arc set $A$.\\
            \textbf{Solutions $\sol$:} \quad The set of all $S \subseteq A$ such that $S$ is a cycle of length $|V|$.
        \end{description}
        \hfill\\
        SE Reduction: \quad From \textsc{Directed Hamiltonian Path}.\\
        Reference.\quad  Arora and Barak \cite{DBLP:books/daglib/0023084}.
    \end{mdframed}
\end{samepage}

\begin{samepage}
    \begin{mdframed}
    	\begin{description}
            \item[]\textsc{Undirected Hamiltonian Cycle} \hfill\\
            \textbf{Instances $\I$:} \quad Graph $G = (V, E)$.\\
            \textbf{Universe $\U$:} \quad Edge set $E$.\\
            \textbf{Solutions $\sol$:} \quad The set of all $S \subseteq E$ such that $S$ is a cycle of length $|V|$.
        \end{description}
        \hfill\\
        SE Reduction: \quad From \textsc{Directed Hamiltonian Cycle}.\\
        Reference.\quad  Karp \cite{DBLP:conf/coco/Karp72}.
    \end{mdframed}
\end{samepage}

\begin{samepage}
    \begin{mdframed}
    	\begin{description}
            \item[]\textsc{Traveling Salesman Problem} \hfill\\
            \textbf{Instances $\I$:} \quad Complete Graph $G = (V, E)$, function $w: E \rightarrow \Z$, number $t \in \N$. \\
            \textbf{Universe $\U$:} \quad Edge set $E$.\\
            \textbf{Solutions $\sol$:} \quad The set of all $S \subseteq E$ such that $S$ is a Hamiltonian cycle and $w(S) \le t$.
        \end{description}
        \hfill\\
        SE Reduction:  \quad \textsc{Undirected Hamiltonian Cycle}\\
        Reference.\quad Folklore.
    \end{mdframed}
\end{samepage}
Let $I = (V, E)$ be the \textsc{Undirected Hamiltonian Cycle} instance and $(V', E', w', k')$ the \textsc{Traveling Salesman Problem} instance.
Every vertex $v \in V$ is mapped to itself $v \in V'$.
Furthermore, we map each edge $e \in E$ to itself in $E'$ and add additional edges to form a complete graph.
The weight function $w': E' \rightarrow \Z$ is defined for all $e' \in E'$ as
$$
    w(e') = \begin{cases}
        0, \quad \text{if} \ e' \in E\\
        1, \quad \text{if} \ e' \notin E
    \end{cases}
$$
At last, we set $k' = 0$ resulting that only the edges from $E$ are usable.
Thus, we preserve the one-to-one correspondence between the edges with $f_I(e) = e$.

\begin{samepage}
    \begin{mdframed}
    	\begin{description}
            \item[]\textsc{Directed Two Vertex Disjoint Paths}  \hfill\\
            \textbf{Instances $\I$:} \quad Directed graph $G = (V, A)$, $s_i, t_i \in V$ for $i \in \{1, 2\}$.\\
            \textbf{Universe $\U$:} \quad Arc set A.\\
            \textbf{Solutions $\sol$:} \quad The set of all $S \subseteq A$ such that $S$ is the union of two disjoint paths $P_1$ and $P_2$, where $P_1$ is a path from from $s_1$ to $t_1$, and $P_2$ is a path from $s_2$ to $t_2$.
        \end{description}
        \hfill\\
        SE Reduction: \quad From \textsc{3-Satisfiability (Literal-based)}.\\
        Reference.\quad  Fortune, Hopcroft and Wyllie \cite{DBLP:journals/tcs/FortuneHW80}.  
    \end{mdframed}
\end{samepage}

\begin{samepage}
    \begin{mdframed}
    	\begin{description}
        \item[]\textsc{Directed} $k$-\textsc{Vertex Disjoint Path}\hfill\\
        \textbf{Instances $\I$:} Directed graph $G = (V, A)$, $s_i, t_i \in V$ for $i \in \fromto{1}{k}$.\\
        \textbf{Universe $\U$:} Arc set $A =: \U$.\\
        \textbf{Solutions $\sol$:} The sets of all sets $S \subseteq A$ such that $S = \bigcup^k_{i = 1} A(P_i)$, where all $P_i$ are pairwise vertex-disjoint paths from $s_i$ to $t_i$ for $1 \leq i \leq k$.
    	\end{description}
    \end{mdframed}
\end{samepage}
We introduce $k-2$ additional vertex pairs $s_i, t_i$ for $i \in \fromto{3}{k}$, which we connect by adding arcs $(s_i, t_i)$ for all $i \in \fromto{3}{k}$.

\begin{samepage}
    \begin{mdframed}
    	\begin{description}
            \item[]\textsc{Steiner Tree}\hfill\\
            \textbf{Instances $\I$:} Graph $G = (S \cup T, E)$, set of Steiner vertices $S$, set of terminal vertices $T$, edge weights $c: E \rightarrow \N$, number $k \in \N$.\\
            \textbf{Universe $\U$:} Edge set $E$.\\
            \textbf{Solutions $\sol$:} The set of all $S \subseteq E$ such that $S$ is a tree connecting all terminal vertices from $T$ and $\sum_{e' \in E'} c(e') \leq k$.
    	\end{description}
    \end{mdframed}
\end{samepage}

There is a folklore reduction from \textsc{3Sat} to \textsc{Steiner Tree}, which is an SE reduction.
First, there are designated terminal vertices $s$ and $t$.
For every literal $\ell \in L$, there is a Steiner vertex $\ell$.
Additionally for every literal pair $(\ell_i, \overline \ell_i)$, $1 \leq i \leq |L|/2-1$, we add a Steiner vertex $v_i$. We define $v_0 := s$.
Then all of the above vertices are connected into a \enquote{diamond chain}, where we begin with $s$ connected to both $\ell_1$ and $\overline \ell_1$.
Both Vertices $\ell_1$ and $\overline \ell_1$ are connected to $v_1$.
This vertex $v_1$ is then connected to vertices $\ell_2$ and $\overline \ell_2$ and so on.
At last, $\ell_{|L|}$ and $\overline \ell_{|L|}$ are connected to $t$.
Furthermore, for every clause $c_j \in C$, we add a corresponding terminal vertex $c_j$.
The vertex $c_j$ is then connected its corresponding literals $\ell \in c_j$ via a path of Steiner vertices of length $|L| + 1$.
The costs of every edge is set to $1$ and the threshold is set to $k = |L| + |C| \cdot (|L| + 1)$.
We define the solution-embedding function $f_I$ for all $\ell \in L$ with $f_I(\ell) = \set{v_{i-1}, \ell}$.

\newpage

\bibliography{bib_general,bib_interdiction,bib_regret,bib_two-stage,bib_recoverable_robust,bib_reductions,bib_multistage_adjustable,bib_nps}

\end{document}